\documentclass[journal]{IEEEtran}
\usepackage[explicit]{titlesec}
\usepackage{graphicx}
\usepackage{easyReview}
\usepackage{verbatim}
\usepackage{epstopdf}
\usepackage{setspace}
\usepackage{subfigure} 
\usepackage{algorithmic}
\usepackage{amsmath}
\usepackage{amsthm}
\usepackage{mathrsfs}
\usepackage{extarrows}
\usepackage{cite}
\newtheorem{theorem}{Theorem}
\newtheorem{example}{Example}

\usepackage{bbm}
\usepackage[bookmarks=false, colorlinks=true, citecolor=blue]{hyperref}
\usepackage{cleveref}
\usepackage{amssymb}
\newtheorem{remark}{Remark}

\newcommand{\RNum}[1]{\uppercase\expandafter{\romannumeral #1\relax}}
\newtheorem{Lemma}{Lemma}

\newtheorem{problem}{Problem}
\newtheorem{definition}{Definition}
\newtheorem{iteration}{Iteration}
\newtheorem{corollary}{Corollary}
\usepackage{stfloats}
\usepackage{color, soul}
\usepackage{booktabs}
\usepackage{tikz,xcolor}
\usetikzlibrary{arrows.meta,automata, positioning,,backgrounds,matrix,shadows}

\newcommand{\finished}{\hfill$\blacksquare$}
\newenvironment{proofsketch}{\begin{proof}[\textit{Proof Sketch}]}{\end{proof}}
\definecolor{lime}{HTML}{A6CE39}
\usepackage[left=0.68in,right=0.68in,top=0.67in,bottom=0.99in]{geometry}

\titlespacing{\section}{0pt}{1.2ex plus .0ex minus .0ex}{.3ex plus .0ex}
\titlespacing{\subsection}{0pt}{1.2ex plus .0ex minus .0ex}{.3ex plus .0ex}
\definecolor{mygreen}{RGB}{46,125,50}
\makeatletter
\DeclareRobustCommand{\orcidicon}{%
	\begin{tikzpicture}
		\draw[lime, fill=lime] (0,0) 
		circle [radius=0.16] 
		node[white] {{\fontfamily{qag}\selectfont \tiny ID}};    \draw[white, fill=white] (-0.0625,0.095) 
		circle [radius=0.007];    \end{tikzpicture}
	\hspace{-2mm}}
\foreach \x in {A, ..., Z}{%
	\expandafter\xdef\csname orcid\x\endcsname{\noexpand\href{https://orcid.org/\csname orcidauthor\x\endcsname}{\noexpand\orcidicon}}
}

\newcommand*\bigcdot{\mathpalette\bigcdot@{.5}}
\newcommand*\bigcdot@[2]{\mathbin{\vcenter{\hbox{\scalebox{#2}{$\m@th#1\bullet$}}}}}
\makeatother

\usepackage[linesnumbered,ruled,vlined]{algorithm2e}
	
	\IEEEoverridecommandlockouts
	\makeatletter
	\def\@thmnote#1{\textit{[#1]}} 
	\makeatother

\begin{document}
		\title{\textit{From Freshness to Effectiveness}: Goal-Oriented Sampling for  Remote Decision Making}
		\author{
			Aimin Li,
			Shaohua Wu,
			Gary C.F. Lee, 
			and Sumei Sun,
			\emph{Fellow, IEEE}
			
			
			\thanks{
				An earlier version of this work was presented in part by IEEE Information Theory Workshop (IEEE ITW) 2024 \cite{li2024sampling}.
				
				Aimin Li is with the Guangdong Provincial Key Laboratory of Aerospace Communication and Networking Technology, Harbin Institute of Technology (Shenzhen), Shenzhen 518055, China. This work is accomplished in part during his visit at Institute for Infocomm Research, Agency for Science, Technology and Research, 136832, Singapore (e-mail: liaimin@stu.hit.edu.cn).
				
				Shaohua Wu is with the Guangdong Provincial Key Laboratory of Aerospace Communication and Networking Technology, Harbin Institute of Technology (Shenzhen), Shenzhen 518055, China, and also with the Peng Cheng Laboratory, Shenzhen 518055, China (e-mail: hitwush@hit.edu.cn). 
				
				Gary C.F. Lee is with the Institute for Infocomm Research, Agency for Science, Technology and Research, 138632, Singapore (e-mail: gary\_lee@i2r.a-star.edu.sg).
				
				S. Sun is with the Institute for Infocomm Research, Agency for Science, Technology and Research, 138632, Singapore (e-mail: sunsm@i2r.a-star.edu.sg).
				
				This work has been supported in part by the National Key Research and Development Program of China under Grant no. 2020YFB1806403, and in part by the Guangdong Basic and Applied Basic Research Foundation under Grant no. 2022B1515120002. 
			}
		}
		
		\maketitle
		\allowdisplaybreaks

		\begin{abstract}
			Data freshness, measured by Age of Information (AoI), is highly relevant in networked applications such as Vehicle to Everything (V2X), smart health systems, and Industrial Internet of Things (IIoT). However, freshness alone does not always equate to utility in decision-making. In decision-critical settings, some \textit{stale} data may be more valuable than \textit{fresh} updates. \textcolor{black}{Motivated by this, we move beyond AoI-centric policies and investigate how data \textit{staleness} affects remote decision-making effectiveness under random delay and limited communication resources.} To this end, we propose {AR-MDP}, an \textbf{A}ge-\textbf{a}ware \textbf{R}emote \textbf{M}arkov \textbf{D}ecision \textbf{P}rocess framework, which co-designs optimal sampling and remote decision-making under a sampling frequency constraint and random delay. To efficiently solve this problem, we design a new \textit{two-stage} hierarchical algorithm, namely \textbf{Q}uick \textbf{B}ellman-\textbf{Linear-Program} ({\normalfont \textsc{QuickBLP}}), where
			the first stage involves solving the Dinkelbach root of a Bellman variant and the second stage involves solving a streamlined linear program (LP). For the tricky first stage, we propose a new One-layer Primal-Dinkelbach Synchronous Iteration ({\normalfont{\textsc{OnePDSI}}}) method, which overcomes the \textit{re-convergence} and \textit{non-expansive divergence} present in existing \textit{per-sample} multi-layer algorithms. Through rigorous convergence analysis of our proposed algorithms, we establish that the worst-case optimality gap in {\normalfont{\textsc{OnePDSI}}} exhibits exponential decay with respect to iteration $K$ at a rate of $\mathcal{O}(\frac{1}{R^K})$. Through \textit{sensitivity analysis}, we derive a threshold for the sampling frequency, beyond which additional sampling does not yield further gains in decision-making. Simulation results validate our analyses. 
		\end{abstract}

		
		\IEEEpeerreviewmaketitle
		
		\section{Introduction}\label{sectionI}	
		Age of Information (AoI) is a crucial metric for evaluating information freshness in status update systems, garnering broad attention from both academia and industry \cite{DBLP:journals/ftnet/KostaPA17,10105150}. {Currently, AoI has been applied in a wide range of applications such as queue control \cite{DBLP:conf/infocom/KaulYG12,talak2020age,costa2016age,bedewy2019minimizing,huang2015optimizing,DBLP:journals/tcom/DoganA21,yates2018age,kam2015effect,moltafet2020age,pappas2015age}, source coding\cite{10715699,mayekar2020optimal}, remote estimation \cite{DBLP:conf/isit/SunPU17,sun2019wiener,mitra2021distributed,chen2021real,sun2019samplingwiener,ornee2019sampling,DBLP:journals/tcom/ArafaBSP21,10.1145/3492866.3549732,tsai2021unifying,ornee2021sampling,arafa2020sample}, and network design \cite{li2022age,pan2022age,xie2020age,meng2022analysis,DBLP:journals/automatica/MenaN23,DBLP:journals/tcom/CaoZJWS21,DBLP:conf/globecom/Long0GLN22,Yusi-tvt-uav-aoi,DBLP:journals/tcom/FengWFCD24,DBLP:journals/tmc/PanCLL23,9036969,Yusi-conf-uav-aoi-drl}} (see \cite{DBLP:journals/jsac/YatesSBKMU21a} for a comprehensive review). Central to this field is the question: ``\textit{How can we minimize the Age of Information?}'' The conventional wisdom in AoI optimization lies in an intuitively compelling yet mathematically non-trivial heuristic: ``\textit{fresher information holds greater value}''. This heuristic finds validation across real-world applications. In Internet of Vehicles (IoV) systems, \textit{timely} status updates are essential for enabling safety-critical driving maneuvers. In financial markets, access to \textit{first-hand} information directly impacts the effectiveness of trading decisions. These examples empirically demonstrate that minimizing AoI can improve estimation accuracy or enhance subsequent information-driven decision-making.
		
		\textcolor{black}{A significant challenge in the field lies in the lack of unified analytical frameworks that link information \textit{freshness} to its \textit{effectiveness} in real-time decision-making. In many scenarios, freshness alone does not determine how beneficial an update is to the downstream decision-making task.} Instead, the \textit{effectiveness} of information may depend on multiple interrelated factors beyond AoI, including the semantic content of transmitted packets and the underlying dynamics of the monitored source\cite{9919752,10579545,9955525,10639525}. This recognition has driven the development of various AoI variants. One approach introduces nonlinear AoI penalties, implemented through both empirical configurations \cite{sun2019sampling,kosta2020cost,cho2003effective,bastopcu2020information} and theoretically derived functions \cite{truong2013effects,8445873,wangFrameworkCharacterisingValue2021,chen2022uncertainty}. These nonlinear formulations aim to quantify the loss resulting from information \textit{staleness}. Additionally, researchers have proposed and optimized various heuristic metrics, including Age of Synchronization (AoS) \cite{zhong2018two}, Age of Incorrect Information (AoII) \cite{AoII,10818531,chen2024minimizing}, Age of Changed Information (AoCI) \cite{AoCI}, and Age of Collected Information \cite{10841476,8825510}. These metrics customize time-related penalization from a wider perspective than what can be captured with age, particularly in applications involving rapidly evolving source dynamics. For remote estimation, mean square estimation error (MSEE) \cite{DBLP:conf/isit/SunPU17} and context-aware Urgency of Information (UoI) \cite{zheng2020urgency} are leveraged to penalize the real-time reconstruction distortion. \textcolor{black}{Despite these advances, the relationship between AoI and decision-making performance is not fully characterized. Specifically, existing studies often optimize communication metrics (e.g., AoI or AoII) without explicitly modeling how delayed information influences sequential decision-making outcomes. This motivates the question: \textit{How does delayed or outdated information affect the quality of remote decisions, and how should communication policies adapt accordingly?}} 
		
		Several works have proposed heuristic approaches to characterize this relationship. In \cite{dong2019age}, Dong \textit{et al.} introduced Age upon Decisions (AuD), which measures the time elapsed between data generation and its use in decision-making, where the decision epoch follows a stochastic distribution. In a similar vein, \cite{nikkhah2023age} proposed Age of Actuation (AoA). \textcolor{black}{In \cite{Kountouris2020SemanticsEmpoweredCF,pappas2021goal,10409276}, Cost of Actuation Error (CoAE) was proposed to penalize \textit{distortion-induced} error actuation. In this setting, a penalty $C_{i,j}$ is incurred when the true system state is $i$ while the remote controller makes decisions based on an estimated value $\hat{X}_t = j$. This line of work primarily focuses on the estimation of a discrete-time Markov chain (DTMC), and quantifies the semantic mismatch between state and inferred action due to delayed or lossy communication.} In \cite{zou2023costly}, three types of decisions: correct decisions, incorrect decisions, and missed decisions are assigned different time-cost functions. A new metric termed Penalty upon Decision (PuD) was proposed. In \cite{10579545} and \cite{10562359}, a tensor-based metric termed Goal-oriented Tensor (GoT) was proposed as a unified framework for existing metrics. However, while prior work has advanced communication optimization through a variety of metrics, it often overlooks how decision-making systems actually function when operating with potentially \textit{stale} information. \textcolor{black}{Bridging this gap requires a framework that explicitly models how stochastic queuing delays influence decision quality in dynamic systems.}
		
		\textcolor{black}{In this paper, we aim to examine how \textit{stale} information impacts remote stochastic decision-making under communication constraints.} To this end, we propose the \textit{Age-aware Remote Markov Decision Process}\footnote{In \cite{10273599}, the term \textit{remote MDP} was first proposed as a pathway to pragmatic or goal-oriented communications. Our paper focuses on the communication delay and introduces the \textit{age} to enhance remote decision-making to achieve a certain goal, hence the term \textit{age-aware remote MDP}.} (AR-MDP), a comprehensive framework that jointly optimizes sampling and sequential decision-making, with a specific purpose to achieve goal-oriented effective decisions. The relationship between sampling and decision-making exhibits inherent \textit{bidirectional coupling}. \textcolor{black}{Sampling decisions affect the freshness of information available for the remote decision maker, which can result in unsatisfactory decision outcomes.} Conversely, decision-making processes affect the stochastic evolution of the source system, which in turn impacts the effectiveness of content-driven goal-oriented sampling mechanisms. To decouple these two processes and achieve optimal decision-making under random delay and a sampling frequency constraint, we formulate the problem as a constrained partially observable semi-Markov decision process, where AoI no longer serves as a typical indicator, \textcolor{black}{but as side information that informs delay-aware decision-making.} We design efficient algorithms to solve this problem. \textcolor{black}{\emph{To the best of our knowledge, this is among the first attempts to treat AoI as dynamic side information in remote decision-making systems, and to systematically integrate it into a formal decision-theoretic framework.}}
		
		\section{Related Work and Our Novelty}\label{section II}
		\subsection{Sampling Under Random Delay}
		The results in this paper contribute to the optimal sampling design under random delay. In real-world network environments, communication channels inevitably experience random delays due to various factors: network handover, congestion, variable sample sizes, and packet retransmissions. These fundamental characteristics have driven research into developing optimal sampling policies under random delay. Existing literature has focused on optimizing three key aspects: $i$) \textbf{information freshness} \cite{DBLP:journals/tit/SunUYKS17,DBLP:journals/jcn/SunC19,DBLP:journals/tit/TangCWYT23,DBLP:journals/ton/PanBSS23,DBLP:journals/tit/LiyanaarachchiU24,peng2024online}; $ii$) \textbf{remote estimation} \cite{ornee2021sampling,DBLP:conf/isit/SunPU17,sun2019wiener,DBLP:journals/ton/TangST24,chen2023sampling,10807024}; and $iii$) \textbf{remote inference} \cite{shisher2024timely,shisher2022does,shisher2023learning,ari2024goal} under random delay. A particularly noteworthy and \textit{counter-intuitive} finding in this field reveals that optimal sampling may require the source to \textit{deliberately wait} before submitting a new sample to the channel, challenging the conventional wisdom of throughput-optimal \textit{zero-wait} sampling policy. 
		
		\textbf{Information Freshness:} In the seminal work  \cite{DBLP:journals/tit/SunUYKS17}, Sun \textit{et al.} derived an AoI-optimal sampling policy under random delay. This paper revealed that under a maximum rate constraint, the AoI-optimal sampling follows a threshold structure, where sampling is activated only when the current AoI exceeds a specific threshold determinable through a low-complexity bisection search method. In \cite{DBLP:journals/jcn/SunC19}, the optimal sampling policy for a non-linear monotonic function of AoI was designed. Tang \textit{et al.} 
		\cite{DBLP:journals/tit/TangCWYT23} extended this framework to scenarios with unknown delay statistics, employing stochastic approximation methods to determine the AoI-optimal sampling policy under unknown delay statistics. Further advancements were made by Pan \textit{et al.} \cite{DBLP:journals/ton/PanBSS23}, who developed AoI-optimal sampling policies under unreliable transmission with random two-way delay.
		Most recently, Liyanaarachchi and Ulukus extended \cite{DBLP:journals/tit/SunUYKS17} by incorporating random ACK delay, demonstrating that sampling before receiving acknowledgment can potentially achieve superior AoI performance \cite{DBLP:journals/tit/LiyanaarachchiU24}. In \cite{peng2024online}, Peng \textit{et al.} designed optimal sampling policies that achieve minimal Age of Changed Information (AoCI) \cite{wang2021age}---a metric capable of detecting source changes---under known and unknown delay statistics. \textcolor{black}{Most recently, Chen \textit{et al.} \cite{chen2024minimizing} derived AoII-optimal sampling policies under random delay. }
			
			\textbf{{Remote Estimation:}} The theoretical foundations of remote estimation-oriented sampling under random delay were established through \cite{DBLP:conf/isit/SunPU17,sun2019wiener}. These studies developed an optimal sampling policy for the Wiener process that minimizes the mean square estimation error (MSEE) while adhering to sampling frequency constraints. Their research revealed that the optimal sampling policy exhibits a threshold structure, where sampling is initiated only when the real-time MSEE surpasses a predetermined threshold. Building upon this foundation, Ornee \textit{et al.} \cite{ornee2021sampling} expanded this theoretical framework by investigating MSEE-optimal sampling for the Ornstein-Uhlenbeck (OU) process, a stationary Gauss-Markov process, under random delay. In \cite{DBLP:journals/ton/TangST24} and \cite{chen2023sampling}, the MSEE-optimal sampling policies are derived for the Wiener process and the OU process under unknown delay statistics. In \cite{10807024}, Chen \textit{et al.} derived the optimal sampling policy that achieves minimum uncertainty of information under random delay, where UoI is defined as the conditional entropy of the source at the receiver given the observation history \cite{chen2022uncertainty}---mathematically expressed as $H(X_t|\mathcal{I}_t)$, with \textcolor{black}{$X_t$} representing the source state at time $t$ and $\mathcal{I}_t$ denoting the available observation history at the receiver. 
			
			\textbf{Remote Inference:} Recent research has revealed insights into remote inference performance and its relationship with information freshness metrics. In \cite{shisher2022does,shisher2024timely}, Shisher \textit{et al.} demonstrated that the loss function in remote inference may not be  monotonic in terms of the age of the samples (features) used, if the source sequence is not Markovian. Upon making this remarkable observation, the authors developed policies that allow selection of aged samples from the buffer, rather than the freshest one. This was termed the  ``\textit{selection-from-buffer}'' model. In \cite{shisher2023learning}, a learning and communication co-design problem that jointly optimizes feature length selection and
			transmission scheduling is proposed. In \cite{ari2024goal}, Ari \textit{et al.} expanded previous works by incorporating time-varying statistics of random delay and delayed feedback, developing optimal sampling policies to minimize long-term inference error within the ``\textit{selection-from-buffer}'' model. All these works reveal that the remote inference utility may not be a monotonic function in terms of AoI. \textcolor{black}{Together, these works demonstrate that in remote inference settings, the utility of a sample is not necessarily a monotonic function of its age.}
			
			\textit{To the best of our knowledge, the optimal sampling policy for \textbf{
					Remote 
					Decision Making} under random delay remains an open research problem, which we address in this paper.}
			\begin{table}[t]
				\centering
				\caption{Comparisons of Time-Lag MDPs}
				\label{1}
				\begin{tabular}{ccc}
					\toprule
					\textbf{Type} & \textbf{Observation} & \textbf{Reference} \\
					\midrule
					Standard MDP & \( O(t) = X_t \) & \cite{bellman1966dynamic}\\
					DDMDP & \( O(t) = X_{t-d} \) & \cite{DBLP:conf/sigmetrics/AltmanN92} \\
					SDMDP & \( O(t) = X_{t-D} \) & \cite{DBLP:journals/tac/KatsikopoulosE03}\\
					Age-Aware Remote MDP & \( O(t) = X_{t-\Delta(t)} \) & {This Work}\\
					\bottomrule
				\end{tabular}
			\end{table}
			\begin{figure}[tbp]
				\centering
				\includegraphics[width=1\linewidth]{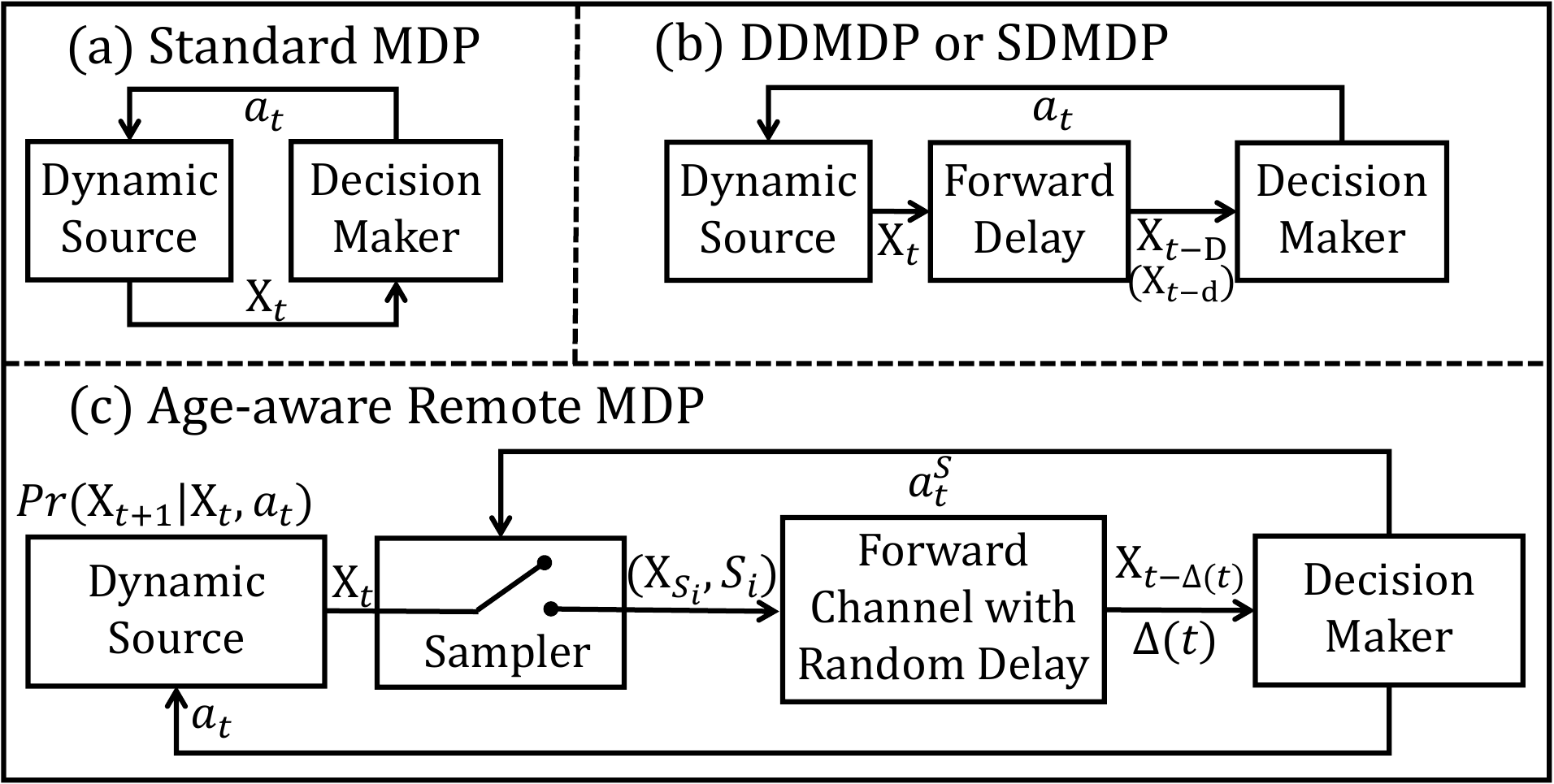}
				\caption{Comparisons among standard MDP, DDMDP, SDMDP, and AR-MDP.}
				\label{fig:figure2}
			\end{figure}
			\subsection{Decision-Making over Stale Status}
			
			The proposed AR-MDP in this paper also enriches the family of time-lag MDP, whose focus is on making decisions based on \textit{stale} status. As illustrated in \textcolor{black}{Table \ref{1}} and Fig. \ref{fig:figure2}, two primary types of MDPs address observation delay at the decision maker: deterministic delayed MDP (DDMDP) \cite{DBLP:conf/sigmetrics/AltmanN92} and stochastic delayed MDP (SDMDP) \cite{DBLP:journals/tac/KatsikopoulosE03}. The DDMDP introduces a constant observation delay $d$ to the standard MDP framework. At any given time $t$, the decision-maker accesses the time-varying data as $O(t)=X_{t-d}$. The main result of the DDMDP problem is its reducibility to a standard MDP without delays through \textit{state augmentation}, as detailed by Altman and Nain \cite{DBLP:conf/sigmetrics/AltmanN92}. The SDMDP extends DDMDP by treating the observation delay not as a static constant but as a random variable $D$ following a given distribution $\Pr(D=d)$, with $O(t)=X_{t-D}$. In 2003, V. Katsikopoulos and E. Engelbrecht showed that an SDMDP is also reducible to a standard MDP problem without delay \cite{DBLP:journals/tac/KatsikopoulosE03}. Thus, it becomes clear to solve an SDMDP problem by solving its equivalent standard MDP.
			
			\begin{figure*}[!t]
				\centering
				\fcolorbox{black}{white}{%
					\includegraphics[width=0.75\linewidth]{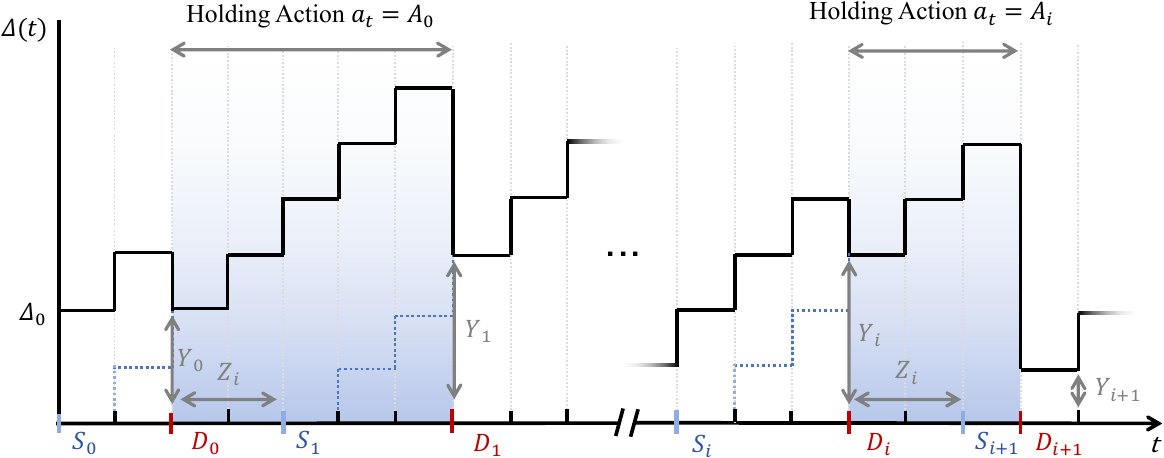}%
				}
				\caption{\textcolor{black}{AoI evolution in slotted time. The $i$-th sample is generated at $S_i$ and delivered at $D_i$ with random delay $Y_i$. The control is updated only upon delivery and is held constant until the next delivery. Shaded areas indicate action-holding intervals, during which the delivered observation remains $X_{S_i}$ for $t\in[D_i,D_{i+1})$ while the staleness $\Delta(t)$ increases linearly.}}
				\label{fig:figure1}
			\end{figure*}
			
			However, the above time-lag MDPs, where the observation delay follows a given distribution (DDMDP can be regarded as a special type of SDMDP), potentially assume that the state is sampled and transmitted to the decision maker \textit{at every time slot}\footnote{In this case, each state $X_{i}, \forall i \in \{0,1,...n\}$ are all sampled and forwarded to the decision maker. The observation delay $D$ is an i.i.d. random variable and is independent of the sampling policy.}. This setup presumes that the system can transmit every state update without encountering any \textit{backlog}. In practice, constantly sampling and transmitting may result in infinitely accumulated packets in the queue, resulting in severe congestion. This motivates the need for queue control and adaptive sampling policy design in the network \cite{DBLP:conf/isit/Yates15,DBLP:journals/tit/SunUYKS17,DBLP:journals/tcom/ArafaBSP21,DBLP:journals/tit/TangCWYT23,panjiayu2023,BZJSSU}, where Age of Information (AoI) serves as a key performance indicator. Suppose the $i$-th sample is generated at time $S_i$ and is delivered at the receiver at time $D_i$. Then, \textcolor{black}{in a time slotted system, }AoI is defined as: \begin{equation}\label{eq1}\Delta(t)=t-S_i, D_{i}\le t< D_{i+1},\qquad \forall i\in\mathbb{N},\quad{\color{black}t\in\mathbb{N}},\end{equation}
			as shown in Fig. \ref{fig:figure1}. From this definition, the most recently available information at the receiver at time slot $t$ is $O(t)=X_{t-\Delta(t)}$. \textcolor{black}{In slotted time, AoI evolves deterministically between deliveries and resets to the realized delivery delay; specifically}
			\begin{equation}\label{eq:aoi-recursion}
				\textcolor{black}{\Delta(t+1)=
					\begin{cases}
						\Delta(t)+1, & \text{if no delivery occurs at } t+1,\\
						Y_i, & \text{if a delivery occurs at } t+1.
				\end{cases}}
			\end{equation}
			
			Different from the DDMDP and SDMDP where the time lag is a constant $d$ or an i.i.d. random variable $D$, with $O(t)=X_{t-d}$ or $O(t)=X_{t-D}$, \textcolor{black}{the effective delay in AR-MDP is \emph{sampling-dependent} (through $S_i$) and coupled with random transmission delay, so the time lag is policy-dependent rather than an exogenous fixed/i.i.d. variable\footnote{\textcolor{black}{While the AoI evolves deterministically between successful updates (i.e., linearly increases), its reset events depend on both the sampling actions and the stochastic delay process. Thus, it can be viewed as a process indirectly governed by the sampling actions.}}.}   
			
			\textit{To the best of our knowledge, the design of optimal remote decision-making \textcolor{black}{in the presence of sampling-dependent stochastic observation delays, as captured by the AoI process,} remains an unexplored research direction, which we address in this paper by studying AR-MDP.}
			
			\subsection{The Novelty of Our Work}
			\begin{itemize}
				\item\textbf{System Model}: This paper proposes AR-MDP, a novel theoretical framework integrating optimal sampling and decision-making under random delay. Differing from prior sampling designs which often treat information as an end in itself---optimizing for freshness, accuracy, or estimation quality under random delays  \cite{DBLP:journals/tit/SunUYKS17,DBLP:journals/jcn/SunC19,DBLP:journals/tit/TangCWYT23,DBLP:journals/ton/PanBSS23,DBLP:journals/tit/LiyanaarachchiU24,peng2024online,shisher2024timely,shisher2022does,shisher2023learning,ari2024goal,ornee2021sampling,DBLP:conf/isit/SunPU17,sun2019wiener,DBLP:journals/ton/TangST24,chen2023sampling,10807024}, we treat information as a \textbf{means to action}, where its \textcolor{black}{contribution} is defined not by how precise it is, but by how well it enables timely and effective decisions. \textcolor{black}{This goes beyond distortion-based formulations such as the Cost of Actuation Error (CoAE) \cite{Kountouris2020SemanticsEmpoweredCF,pappas2021goal,fountoulakis2023goal,10409276}, by embedding the staleness-induced impact directly into the sequential decision-making process.}
				\textcolor{black}{Different from classical time-lag MDPs (e.g., DDMDP/SDMDP) that model the delay as an \emph{exogenous} quantity where either a fixed constant $d$ or a random variable $D$ independent of the sampling policy, our AR-MDP adopts a different remote-decision information structure in which AoI $\Delta(t)$ is an \emph{endogenous} staleness process shaped jointly by controllable sampling or waiting decisions and random delay $Y_i$. Under the sample-and-hold information structure, we establish an exact fixed-dimensional sufficient statistic and an embedded lifted MDP (Lemma~\ref{l1}), enabling tractable average-cost analysis.}

				\item\textbf{Methodology Design:} We design \textsc{QuickBLP}, a computationally efficient single-layer algorithm that addresses limitations associated with iterative re-convergence common in multi-layer per-sample algorithms, \emph{e.g.}, \cite[Algorithm 1]{10807024}, \cite[Algorithm 1]{bedewy2021optimal}, and \cite[Section IV.C]{10621420}. This algorithm is designed based on a key analytical insight that the optimal solution exhibits a threshold structure and can be obtained through a two-stage process. The first stage determines the \textit{Dinkelbach root} of a Bellman variant, for which we develop \textsc{OnePDSI}, a \textit{Cauchy sequence} that converges asymptotically to the root without requiring \textit{re-convergence}. The second stage involves finding the \textit{Dinkelbach root} of a \textit{per-sample} constrained Markov Decision Process (CMDP), which traditionally necessitates multiple CMDP solutions and suffers from \textit{re-convergence}. We resolve this challenge by proving that the \textit{Dinkelbach root} can be explicitly calculated through the optimal value of a linear program (LP), enabling direct solution to the root through a single LP solution. To the best of our knowledge, \textsc{QuickBLP} is the first framework to tackle constrained partially observable SMDPs \textbf{through a streamlined single-layer flow}, fundamentally improving efficiency by eliminating re-convergence overhead.
				
				\item \textbf{Theoretical Rigor and Convergence:} We significantly advance our previous work \cite{li2024sampling} by resolving critical convergence limitations. To overcome inherent divergence risks, we design two novel algorithms in this paper: Bisec-$\tau$-RVI and \textsc{OnePDSI}. Although convergence analyses often rely on the Banach Contraction Mapping Theorem \cite[Theorem 6]{banach1922operations} \cite{DBLP:journals/ton/PanBSS23,10621420}, this method falls short in capturing the behavior of our models. Our approach departs from this tradition, providing rigorous proofs that both algorithms guarantee efficient \textit{exponential convergence} in \textit{worst-case} settings. These guarantees reinforce the reliability and efficiency of our algorithms in practical applications.
			\end{itemize}
			\subsection{Notations}
			\textcolor{black}{The main notations throughout this paper are summarized in Table \ref{tab:notation}.}
			\begin{table}[t]
				\color{black}
				\caption{Summary of Notations}
				\label{tab:notation}
				\centering
				\resizebox{0.5\textwidth}{!}{%
					\begin{tabular}{cl}
						\hline
						\textbf{Symbol} & \textbf{Description} \\ \hline
						$X_t$ & System state at time slot $t$ \\
						$a_t$ & Response action taken by the decision maker at time $t$ \\
						$u_t$ & Sampling action at time slot $t$ \\
						$\pi_t$ & History-dependent policy \\
						$\phi_t$ & State-dependent policy \\
						$\boldsymbol{\psi}$ & Policy composed by $\phi_{0:\infty}$\\
						$\boldsymbol{\psi}^{\lambda}$ & Policy in Problem \ref{p4} induced by parameter $\lambda$\\
						$\mathcal{H}_t(\cdot)$ & Sufficient statistics function \\
						$\Delta(t)$ & AoI at time $t$ \\
						$S_i$ & Time slot when the $i$-th sample is taken \\
						$D_i$ & Time slot when the $i$-th sample is delivered \\
						$Y_i$ & Random delivery delay between $S_i$ and $D_i$ \\
						$Z_i$ & Waiting time to sample the $i$-th sample \\
						$A_i$ & Holding-action taken at the $i$-th epoch \\
						$G_i$ & Lifted State of Lifted MDP at epoch $i$ \\
						$C(x,a)$ & Cost function given state $x$ and action $a$ \\
						$\lambda$ & Dinkelbach parameter \\
						$\mathscr{P}_\mathrm{MDP}(\lambda)$ & Transformed MDP tuple give parameter $\lambda$ \\
						$\phi_\lambda^*$ & Optimal policy of MDP $\mathscr{P}_\mathrm{MDP}(\lambda)$ \\
						$U(\lambda)$ & Optimal value of Problem \ref{p4} \\
						$V^\star(\cdot;\lambda)$ & Optimal relative value function given $\lambda$ \\
						${U}_K(\lambda),{V}_K(\gamma;\lambda)$ & RVI values in \eqref{RVIiteration} \\
						$\tilde{U}_K,\tilde{V}_K(\gamma;\lambda)$ & $\tau$-RVI values in \eqref{MRVI} \\
						$e_{\mathrm{U}}^{(K)}(\lambda)$,$e_{\mathrm{V}}^{(K)}(\cdot;\tau,\lambda)$ & Relative errors for the $K$-iteration in $\tau$-RVI given $\lambda$\\
						$\kappa$ & A parameter for \textsc{OnePDSI}\\
						$e_{\rho}^{(K)}$,$e_{\mathrm{W}}^{(K)}(\cdot;\kappa)$ & Relative errors for the $K$-iteration in \textsc{OnePDSI} given $\kappa$\\
						$W^\star(\cdot)$ & Variables in fixed-point equations \eqref{eqfunction} \\
						$\rho_K$,$\widetilde{W}_K(\cdot)$ & \textsc{OnePDSI} values in \eqref{prop2}\\
						$f_{\max}$ & Maximum average sampling rate\\
						$H(\lambda;f_{\max})$ & Optimal value of Problem \ref{p5}\\
						$\theta$ & Lagrangian multiplier\\
						$\mathcal{L}(\boldsymbol{\psi};\theta,\lambda,f_{\max})$ & Lagrangian function\\
						$\Upsilon(\theta,\lambda;f_{\max})$ & Lagrangian dual function\\
						$d(\lambda;f_{\max})$ & Optimal value of Problem \ref{p6}\\
						$Q^\star(f_{\max})$ & Optimal value of Problem \ref{p7}\\
						$\theta_{\lambda}^{\star}$ & Optimal variable $\theta$ for fixed $\lambda$ in Problem \ref{p6}\\
						$\mathcal{Q}^{\lambda}$ & Long-term average cost given policy $\phi_{\lambda}^*$ \\
						$\mathcal{F}^{\lambda}$ & Long-term average sampling rate given policy $\phi_{\lambda}^*$\\
						$\mathcal{F}^{\lambda^+}$ & Right limit of $\mathcal{F}^{\lambda}$\\
						$\mathcal{F}^{\lambda^-}$ & Left limit of $\mathcal{F}^{\lambda}$\\
						$f_{\max}^T$ & Sapling frequency threshold\\
						$\mathbf{P}_a$ & Transition probability matrix given action $a$ \\
						$\mathbf{P}^{n}$ & $n$-step transition probability matrix given action $a$ \\
						$\mathbf{P}_{i\times j}$ & the $(i,j)$-th entry of the transition matrix $\mathbf{P}$ \\
						$\rho^{\star}$ & Optimal value of Problem \ref{p3}; Root of $U(\lambda)$ \\
						$h^{\star}$ & Optimal value of Problem \ref{p1}; Root of $H(\lambda;f_{\max})$\\
						$\pi_a(\cdot)$ & Stationary distribution over states under $\mathbf{P}_a$\\
						\hline
					\end{tabular}
				}
			\end{table}
			
			\section{System Model and Problem Formulation }\label{section III}
			

			We consider a time-slotted\footnote{\textcolor{black}{As the proposed AR-MDP is formulated as an extension of the discrete-time MDP framework, a time-slotted system model is employed to maintain structural consistency with MDP. Therefore, all key variables including transmission delay $Y_i$, AoI $\Delta(k)$, and sampling time $S_i$ are accordingly defined and evolve over discrete time slots.}} \textit{age-aware remote MDP} problem illustrated in Fig. \ref{fig:figure2}(c). Let $X_t\in\mathcal{S}$ be the controlled source of interest at time slot $t$. The evolution of the source is a Markov decision process, characterized by the transition probability $\Pr(X_{t+1}|X_t,a_t)$\footnote{For short-hand notations, we use the transition probability matrix $\mathbf{P}_a$ to encapsulate the dynamics of the source given an action $a_t=a$.}, where $a_t\in\mathcal{A}$ represents the controlled action taken by the remote decision maker to control the source in the desired way. \textcolor{black}{We assume that both the state space $\mathcal{S}$ and the action space $\mathcal{A}$ are finite. This finite setting is a common starting point in infinite-horizon average-cost MDP literature \cite{puterman2014markov,bertsekas2012dynamic} and allows the proposed AR-MDP to inherit well-established \textit{optimality} and \textit{convergence results}. We view AR-MDP as a stepping stone toward more general models, and refer readers to \cite[Chap. 4.6]{bertsekas2012dynamic} for approaches that can be used to extend our framework to infinite spaces.}
			
			The sampler conducts the sampling action $u_t\in\{0,1\}$, with $u_t=1$ representing the sampling action and $u_t=0$ otherwise. Let $S_i$ be the sampling time of the $i$-th delivered packet, and $D_i$ be the corresponding delivery slot. The random channel delay of the $i$-th packet is denoted as $Y_i\in\mathcal{Y}\subseteq\mathbb{N}^+$, which is independent of the source $X_t$ and is bounded $\max[{Y_i}]<\infty$. The sampling times $S_0,S_1,\cdots$  record the time stamp when $u_t=1$, given by
			\begin{equation}\label{eq2}
				S_i=\max\{\textcolor{black}{t\in\mathbb{N}}\left|t\le D_i,u_t=1\right.\},\quad \forall i\in\mathbb{N},
			\end{equation}
			where the initial state of the system is $S_0=0$ and $\Delta(0)=\Delta_0$. 
			
			\textcolor{black}{Since the system is time-slotted, the AoI evolves in discrete steps. Specifically, for any slot $t$, let
				$i(t)\triangleq \max\{i: D_i\le t\}$ be the index of the most recently delivered update. Then $
				\Delta(t)=t-S_{i(t)}$. Hence, upon a delivery at $t=D_i$, the AoI drops to
				\begin{equation}\label{eq:AoI_drop_item}
					\Delta(D_i)=D_i-S_i=Y_i,
				\end{equation}
				which is \emph{not necessarily zero} unless the channel delay is zero. Moreover, during the interval
				$t\in[D_i,D_{i+1})$, no new observation is delivered and the age $\Delta(t)$ drifts with  $\Delta(t+1)=\Delta(t)+1$. The discrete-time AoI dynamics are shown in Fig. \ref{fig:figure1}.}
			
			At the sampling time $S_i,\forall i\in\mathbb{N}$, the state $X_{S_i}$ along with the corresponding time stamp $S_i$ is encapsulated into a packet $(X_{S_i},S_i)$, which is transmitted to a remote decision maker. \textcolor{black}{Upon receipt of the packet $(X_{S_i},S_i)$ at delivery time slot $D_i$, the \textit{observation history} at the decision maker is
				$
				\{(X_{S_j}, S_j,D_j): j \le i\}, t \in [D_i, D_{i+1}).
				$
				By employing \eqref{eq1}, for any time slot \( t \in\mathbb{N} \), the freshest available sample at the decision maker is $(X_{t-\Delta(t)},t-\Delta(t))$. As $t$ is known to the decision maker, the \textit{observation history} up to time slot $t$ is equivalently expressed by AoI, given as:}
			$
			\color{black}
			\{(X_{k-\Delta(k)}, \Delta(k)): k \le t,k\in\mathbb{N}\}.
			$
			
			\subsection{\textcolor{black}{Protocol and Assumptions}}
			Similar to \cite{DBLP:journals/tit/SunUYKS17}, we impose the following assumptions in sampling:
			\begin{itemize}
				\item[(S1)]  A new sample cannot be generated until the previous sample has been delivered. Specifically,
				\begin{equation}\label{eq4}
					S_{i+1}=D_i+Z_i,\quad Z_i\ge 0,\quad i\in\mathbb{N},
				\end{equation}
				where $Z_i$ is the sampling waiting time after the delivery at $D_i$.
				Consequently, the delivery time satisfies $D_i=S_i+Y_i$ for all $i\in\mathbb{N}$, where $Y_i$ is the random delay of the $i$-th sample.
				\textcolor{black}{Moreover, under (S1) and the delivery timeline, the observation is piecewise constant between two consecutive deliveries. 
					Specifically, for any $t\in[D_i,D_{i+1})$ we have $t-\Delta(t)=S_i$, and hence $
					O(t)=X_{t-\Delta(t)}=X_{S_i}, D_i\le t< D_{i+1},
					$
					while the staleness $\Delta(t)$ increases deterministically within the interval.}
				
				\item[(S2)] The inter-sample times $G_i=S_{i+1}-S_i$ form a regenerative process \cite[Section 6.1]{haas2006stochastic}. Hence, almost surely\footnote{This assumption also implies that the waiting time $Z_i$ is bounded, belonging to a subset of nature numbers with $Z_i\in\mathcal{Z}\subseteq\mathbb{N}$.},
				\begin{equation}\label{eq3}
					\lim_{i\to\infty} S_i=\infty,\qquad \lim_{i\to\infty} D_i=\infty.
				\end{equation}
			\end{itemize}
			
			\textcolor{black}{In addition, we adopt a \textit{holding-action} paradigm consistent with the timeline in Fig. \ref{fig:figure1}:}
			\begin{itemize}
				\item[(A1)] 
				The controlled action is updated only upon the delivery of a sample.
				That is, upon the delivery time $D_i$, the decision maker selects an updated action $A_i\in\mathcal{A}$ based on the available history, and then holds it constant until the next delivery:
				\begin{equation}\label{eq8}
					a_t = A_i,\quad D_i \le t < D_{i+1},\quad i\in\mathbb{N}.
				\end{equation} 	 
				\textcolor{black}{As a result, the underlying state continues to evolve according to the controlled Markov kernel, i.e., $X_{t+1}\sim P(\cdot\mid X_t,A_i)$ for all $t\in[D_i,D_{i+1})$.}
			\end{itemize}
			
			\begin{remark}
				\textcolor{black}{The modeling assumption (A1) is widely used in networked and remote control systems where the controller can revise its command only when a new measurement is received; see, e.g., \cite{Lee2006PassiveTeleop,dimarogonas2009event,hespanha2007survey}.
					Typical real-world systems that follow the \textit{holding-action} paradigm include:
					\begin{itemize}
						\item \textbf{Robotic manipulation and teleoperation:} In robotic teleoperation under constant communication delays, the controller continuously updates the motion command based on the current local state and the most recently received delayed remote state, effectively maintaining past remote information during delay periods \cite{Lee2006PassiveTeleop}.		
						\item \textbf{Multi-UAV or vehicular coordination:} each autonomous agent maintains its last chosen coordination strategy (e.g., formation control gain or following distance target) until updated state information is received from neighboring agents \cite{dimarogonas2009event}.			
						\item \textbf{Industrial supervisory control:} in process plants or smart grids, a supervisory controller holds the previously assigned operation mode (e.g., heater on/off state, pump flow rate setpoint) during communication gaps between control center and local devices \cite{hespanha2007survey}.
				\end{itemize}}
			\end{remark}

			\begin{remark}[Beyond holding actions]\label{rem:beyond_holding}
				\textcolor{black}{The holding-action rule \eqref{eq8} induces an epoch-based decision structure: between two consecutive deliveries, no new observation is received and the applied action is fixed, and the state continues to evolve according to the controlled Markov kernel under the held action. This structure is essential for compressing the growing history into a finite-dimensional sufficient statistic (Lemma~\ref{l1}) and for constructing a finite-state lifted MDP amenable to our average-cost optimality analysis and algorithms\footnote{\textcolor{black}{If (A1) is relaxed so that the decision maker can adapt $a_t$ within $t\in[D_i,D_{i+1})$ based on the deterministic AoI drift, then the sufficient statistic generally becomes belief-based: the conditional distribution of $X_t$ depends on the staleness and the within-epoch action sequence. This leads to a (belief-)MDP/POMDP formulation with either a continuous state space (belief simplex) or a significantly enlarged epoch action space (age-indexed action plans). We leave this non-holding extension as future work.}}.}
			\end{remark}

			\subsection{\textcolor{black}{Joint Sampling and Decision-Making Policy}}
			Let $\mathcal{I}_t=\{(X_{k-\Delta(k)},\Delta(k),u_{k-1},a_{k-1}): k\le t\}$ denote the history available to the decision maker up to time $t$.
			Under the protocol in \eqref{eq4} and \eqref{eq8}, the slot-level actions $(u_t,a_t)$ are induced by the epoch-level decisions $\{(A_i,Z_i)\}$ made at delivery epochs. A possibly randomized decision policy is a sequence of mappings from the history to a distribution over the joint action space $\{0,1\}\times\mathcal{A}$:
			\begin{equation}\label{eq3new}
				\pi_t:\mathcal{I}_t \rightarrow \mathcal{P}(\{0,1\}\times \mathcal{A}),
			\end{equation}
			where $\mathcal{P}(\cdot)$ is a \textit{simplex} space which represents the probability that an action is taken. 
			
			\textcolor{black}{An epoch-based policy is defined as
				\begin{equation}
					\phi:\mathcal{S}\times\mathcal{Y}\times\mathcal{A}\to \mathcal{P}(\mathcal{A}\times\mathcal{Z}),
				\end{equation}
				which induces the slot-level policy $\pi_t$ via $S_{i+1}=D_i+Z_i$ and $a_t=A_i$ on $[D_i,D_{i+1})$.
				For completeness, a slot-level policy can be written as $\pi_t:\mathcal I_t\to \mathcal P(\{0,1\}\times\mathcal A)$, which in our setting is induced by $\phi$. (See Lemma \ref{l1}).}

			We consider a bounded cost function $\mathcal{C}(X_t,a_t)<\infty$, which represents the \textit{immediate cost} incurred when action $a_t$ is taken in state $X_t$. Under the above assumptions, the objective of the system is to design the optimal joint sampling and decision policies at each time slot, \emph{i.e.}, $\pi_0, \pi_1,\pi_2\cdots$, to minimize the \textit{long-term average cost}, subject to a \textit{long-term average sampling frequency constraint}:
			
						\begin{problem}[\textit{Joint Design of Sampling and Decision Processes under Sampling Frequency Constraint}]\label{p1}
				\begin{subequations}
					\begin{align}\label{eq5}
						&\inf _{\pi_{0: \infty}} \limsup _{T \rightarrow \infty} \frac{1}{T} \mathbb{E}_{\pi_{0: \infty}}\left[\sum_{t=1}^T \mathcal{C}(X_t, a_t)\right]\\
						&\text{s.t. } \liminf_{N \to \infty} \frac{1}{N} \mathbb{E}_{\pi_{0: \infty}}\left[\sum_{i=0}^{N-1}(S_{i+1}-S_i)\right] \geq \frac{1}{f_{\text{max}}},
					\end{align}
				\end{subequations}
				where $\pi_t$ is the joint sampling and decision-making policy defined by \eqref{eq3new}, $f_{\text{max}}$ represents the maximum allowed sampling frequency, and the expectation $\mathbb{E}_{\pi_{0:\infty}}$ is taken over the stochastic processes $(X_1,X_2,\cdots)$ and $(Y_0,Y_1,\cdots)$ under given policies $\pi_{0:\infty}$.	
			\end{problem}
			
			In practice, the overhead for information updates will increase with the average sampling frequency. Hence, Problem \ref{p1} represents a tradeoff between remote decision-making utility and communication overhead. Since \textit{age} $\Delta(k)$ is available at the decision maker as \textit{side information} to facilitate more informed decision-making, we call this problem an {age-aware remote MDP} problem. This problem aims at determining the distribution of joint sampling and controlled actions $(u_t,a_t)$ based on the history $\mathcal{I}_t$, such that the {long-term average cost} subject to the {sampling frequency constraint} is minimized.
			
			We remark that we will study Problem \ref{p1} both without \textcolor{black}{(see Section \ref{sectionIV} and \ref{sectionIV2}) and with (see Section \ref{sectionV} and \ref{Onelayerwithrate})} the sampling frequency constraint. To distinguish these two problems, we use $h^\star$ to denote the optimal value of the problem with the rate constraint and $\rho^\star$ to denote the optimal value of the problem without the sampling frequency constraint.
			
			\subsection{Sufficient Statistics of History}\label{subsectionA}
			\textcolor{black}{In principle, the joint policy $\pi_t$ maps the growing history $\mathcal I_t$ to a distribution over $(u_t,a_t)$, which leads to the classical \emph{curse of history} and makes direct dynamic programming intractable.}
			
			\textcolor{black}{More importantly, under delayed observations, optimal decisions generally depend not only on the latest delivered sample value but also on its \emph{staleness}. 
				In our time-slotted model, the staleness at delivery epochs is $\Delta(D_i)=Y_i$. 
				{Within the interval $t\in[D_i,D_{i+1})$, the freshest delivered sample remains $X_{S_i}$ and the AoI drifts deterministically as $\Delta(t)=Y_i+(t-D_i)$, i.e., the time since sampling equals $t-S_i=\Delta(t)$.}
				This determines how many Markov transitions have occurred since the sampled state. Consequently, even conditioned on the same delivered sample $X_{S_i}$ and held actions, different $Y_i$ (and hence different $\Delta(t)$ in the epoch $t\in[D_i,D_{i+1})$) induce different conditional distributions (beliefs) of the current state.}
			
			\textcolor{black}{Crucially, under the holding-action rule \eqref{eq8}, no new observation is delivered between two consecutive deliveries and the action is kept constant. This induces an epoch-based decision structure, under which the growing history can be compressed into a finite-dimensional sufficient statistic. The resulting statistic, given in Lemma~\ref{l1}, serves as the state of the lifted MDP and enables our subsequent analysis and algorithms.}

			\textcolor{black}{A sufficient statistic is defined as follows.}
			\begin{definition}\label{definition1}
				A sufficient statistic of $\mathcal{I}_t$ is a function $\mathcal{H}_t(\mathcal{I}_t)$, such that
				\begin{equation}\label{definition1eq}
					\min_{a_{t:T}}\mathbb{E}\left[\sum_{k=t}^T \mathcal{C}(X_k, a_k)|\mathcal{I}_t\right]=\min_{a_{t:T}}\mathbb{E}\left[\sum_{k=t}^T \mathcal{C}(X_k, a_k)|\mathcal{H}_t(\mathcal{I}_t)\right]
				\end{equation} holds for any $T>t$. 
			\end{definition} 
			This definition suggests that decision-making that leverages the \textit{sufficient statistics} $\mathcal{H}_t(\mathcal{I}_t)$ can achieve the same performance as using the complete history $\mathcal{I}_t$. Thus, the compression $\mathcal{H}_t(\mathcal{I}_t)$ is \textit{sufficient} for the agent to implement an optimal action, enabling the design of efficient policies that maintain optimality while overcoming the \textit{curse of history}. \textcolor{black}{In particular, the time stamp carried by each packet makes the staleness observable at delivery: upon receiving $(X_{S_i},S_i)$ at time $D_i$, the decision maker can compute $\Delta(D_i)=D_i-S_i=Y_i$. This staleness information enters the sufficient statistic and therefore the optimal delivery-epoch policy.}
			The following Lemma establishes efficient \textit{sufficient statistics} of the \textit{history} $\mathcal{I}_t$ in a piecewise manner.
			\begin{Lemma}\label{l1}
				(Sufficient Statistics). During the interval $t\in [D_{i},D_{i+1})$, $\mathcal{G}_i=(X_{S_i},Y_i,A_{i-1})\in\mathcal{S}\times\mathcal{Y}\times\mathcal{A}$ is a sufficient statistic of $\mathcal{I}_t$. In addition, determining the optimal sampling actions $u_t$ under condition (\ref{eq4}) is equivalent to determining the optimal sampling time $S_{i+1}$, or the optimal waiting time $Z_i$.
			\end{Lemma}
			\begin{proof}
				\textcolor{black}{See Appendix \ref{l1proof}.}
			\end{proof}

			Lemma~\ref{l1} indicates that solving Problem~\ref{p1} \textcolor{black}{over the history-dependent slot-level policies}
			\begin{equation}
				\pi_t:\mathcal{I}_t\rightarrow\mathcal{P}(\{0,1\}\times\mathcal{A})
			\end{equation}
			is equivalent to \textcolor{black}{solving over delivery-epoch policies that depend only on $\mathcal{G}_i$:}
			\begin{equation}
				\phi_i:\mathcal{S}\times\mathcal{Y}\times\mathcal{A}\rightarrow\mathcal{P}(\mathcal{A}\times\mathcal{Z}),\quad \forall i\in\mathbb{N},
			\end{equation}
			which maps the sufficient statistic $\mathcal{G}_i=(X_{S_i},Y_i,A_{i-1})$ to a distribution over the joint epoch actions $(A_i,Z_i)\in\mathcal{A}\times\mathcal{Z}$.
			\textcolor{black}{Consequently, the AoI at delivery, equivalently the realized delay $Y_i=\Delta(D_i)$, is an explicit argument of the optimal policy via $\mathcal{G}_i$. This provides a concrete AoI-related conclusion: even for the same delivered sample value $X_{S_i}$, different staleness values $Y_i$ can induce different beliefs and thus the optimal epoch sampling and decision-making may depend on $Y_i$.}
			\textcolor{black}{Because the Cartesian product $\mathcal{S}\times\mathcal{Y}\times\mathcal{A}$ does not grow with time, this reformulation avoids the exponential growth of the original information set $\mathcal{I}_t$.}
			From assumptions (S1) and (S2), we reformulate Problem~\ref{p1} as:
			
			\begin{problem}[\textit{From History-Dependent to State-Dependent Policy}]\label{p2}
				\begin{subequations}
					\begin{align}\label{eq6}
						&{\color{black}\inf _{\boldsymbol{\psi}} \limsup _{N \rightarrow \infty} \frac{\mathbb{E}_{\boldsymbol{\psi}}\left[\sum_{i=0}^{N-1}\sum_{t=D_{i}}^{D_{i+1-1}} \mathcal{C}(X_t, a_t)\right]}{\mathbb{E}_{\boldsymbol{\psi}}[D_N]}} \\
						&\text{s.t. } \liminf_{N \to \infty} \frac{1}{N} \mathbb{E}_{\boldsymbol{\psi}}\left[\sum_{i=1}^{N}(S_{i+1}-S_i)\right] \geq \frac{1}{f_{\text{max}}},
					\end{align}
				\end{subequations}
				Here, we use $\boldsymbol{\psi}$ to denote the policy $\{\phi_i\}_{i=0}^{\infty}$.
			\end{problem}

			\section{Optimal Sampling Without Rate Constraint: A Two-Layer Persample Solution}\label{sectionIV}
			
			In this section, we address the unconstrained problem to determine $\rho^{\star}$. A series of theoretical results are developed through a divide-and-conquer approach. In subsection \ref{subsectionB}, we rewrite the unconstrained problem into a \textit{non-linear fractional program}. By utilizing a \textit{Dinkelbach-like} method \cite{dinkelbach1967nonlinear}, we transform the \textit{non-linear fractional programming} into an infinite horizon MDP problem given a \textit{Dinkelbach parameter}. The problem then becomes finding the \textit{Dinkelbach parameter} such that the optimal value of the standard MDP is zero, which is equivalently a root-finding problem. To search for the root, in subsection \ref{sectionIIC} we review a typical \textit{two-layer nested} algorithm, namely \textit{Bisec-RVI} \cite[Algorithm 2]{li2024sampling}\footnote{The idea of the {two-layer nested} algorithm has been applied in \cite{DBLP:journals/tit/SunUYKS17,sun2019samplingwiener,DBLP:journals/tit/BedewySKS21} and \cite{sun2019sampling} to achieve Age-optimal and Mean Square Error (MSE)-optimal sampling and scheduling.}.  We note that the inner-layer \textit{Relative Value Iteration} (RVI) algorithm of \textit{Bisec-RVI} may suffer from divergence and thus propose an improved \textit{Bisec-$\tau$-RVI} algorithm to achieve provable convergence. 
			
			\subsection{A Reformulation of Problem \ref{p2}}\label{subsectionB}
			\textcolor{black}{From the action-holding assumption (A1), we can rewrite the objective function in Problem \ref{p2} as: } 
			\begin{subequations}\color{black}
				\begin{align}
						&\limsup _{N \rightarrow \infty} \frac{\mathbb{E}_{\boldsymbol{\psi}}\left[\sum_{i=0}^{N-1}\sum_{t=D_{i}}^{D_{i+1-1}} \mathcal{C}(X_t, a_t)\right]}{\mathbb{E}_{\boldsymbol{\psi}}[D_N]} \\&=\lim _{\mathrm{n} \rightarrow \infty} \frac{\sum_{i=0}^{N-1} \mathbb{E}_{\boldsymbol{\psi}}\left[\sum_{t=\it{D}_{i}}^{\it{D}_{i+1}-1} \mathcal{C}(X_t, A_i)\right]}{\sum_{i=0}^{N-1} \mathbb{E}_{\boldsymbol{\psi}}\left[\it{D}_{i+1}-\it{D}_{i}\right]}\\&=\lim _{\mathrm{n} \rightarrow \infty} \frac{\sum_{i=0}^{n-1} \mathbb{E}_{\boldsymbol{\psi}}\left[\sum_{t=\it{D}_{i}}^{\it{D}_{i+1}-1} \mathcal{C}(X_t, A_i)\right]}{\sum_{i=0}^{n-1} \mathbb{E}_{\boldsymbol{\psi}}\left[Y_{i+1}+Z_i\right]}.
				\end{align}
			\end{subequations}
			\textcolor{black}{We can then express the \textit{unconstrained} version of Problem~\ref{p2} as the epoch-by-epoch variant:}
			
			\begin{problem}[\textit{\textcolor{black}{Epoch-by-Epoch} Reformulation}]\label{p3}
				\begin{equation}
					\rho^{\star} \triangleq \inf _{\boldsymbol{\psi}}\lim _{\mathrm{N} \rightarrow \infty} \frac{\sum_{i=0}^{N-1} \mathbb{E}_{\boldsymbol{\psi}}\left[\sum_{t=\it{D}_{i}}^{\it{D}_{i+1}-1} \mathcal{C}(X_t, A_i)\right]}{\sum_{i=0}^{N-1} \mathbb{E}_{\boldsymbol{\psi}}\left[Y_{i+1}+Z_i\right]}.
				\end{equation}
			\end{problem}

			Problem \ref{p3} is a \textit{non-linear fractional program}, which is challenging due to its fractional nature. To simplify this problem, we consider the following sequential decision process with Dinkelbach parameter $\lambda\ge0$: 
			
			\begin{problem}[\textit{Standard Infinite-Horizon Sequential Decision Process with Dinkelbach Parameter $\lambda$}]\label{p4}
				\begin{equation}\label{p4eq}
					\begin{aligned}
						&U(\lambda)\triangleq\\&\inf _{\boldsymbol{\psi}^\lambda} \lim _{\mathrm{n} \rightarrow \infty}\frac{1}{n} {\sum_{i=0}^{n-1}\left\{ \mathbb{E}\left[\sum_{t=\it{D}_{i}}^{\it{D}_{i+1}-1} \mathcal{C}(X_t, A_i)\right]-\lambda\mathbb{E} \left[Z_i+Y_{i+1}\right]\right\}}.
					\end{aligned}
				\end{equation}
			\end{problem}
			
			By similarly applying the Dinkelbach-like method for \textit{non-linear fractional programming} \cite{dinkelbach1967nonlinear}, we can obtain the following lemma:
			\begin{Lemma}\label{l3}
				The following assertions hold:\\
				(i). $\rho^{\star} \gtreqless \lambda \text { if and only if } U(\lambda) \gtreqless 0 \text {. }$\\
				(ii). When $U(\lambda)=0$, the policy solutions to Problem \ref{p4} are equivalent to those of Problem \ref{p3}.\\
				(iii). $U(\lambda)=0$ has a unique root, with $U(\rho^{\star})=0$.
			\end{Lemma}
			\begin{proof}
				\textcolor{black}{See Appendix \ref{proofofl2}.}
			\end{proof}
			
			This key lemma enables the formulation of the following two-layer nested approach to determine $\rho^\star$.
			
			\begin{figure}[tbp]
				\centering
				\includegraphics[width=0.85\linewidth]{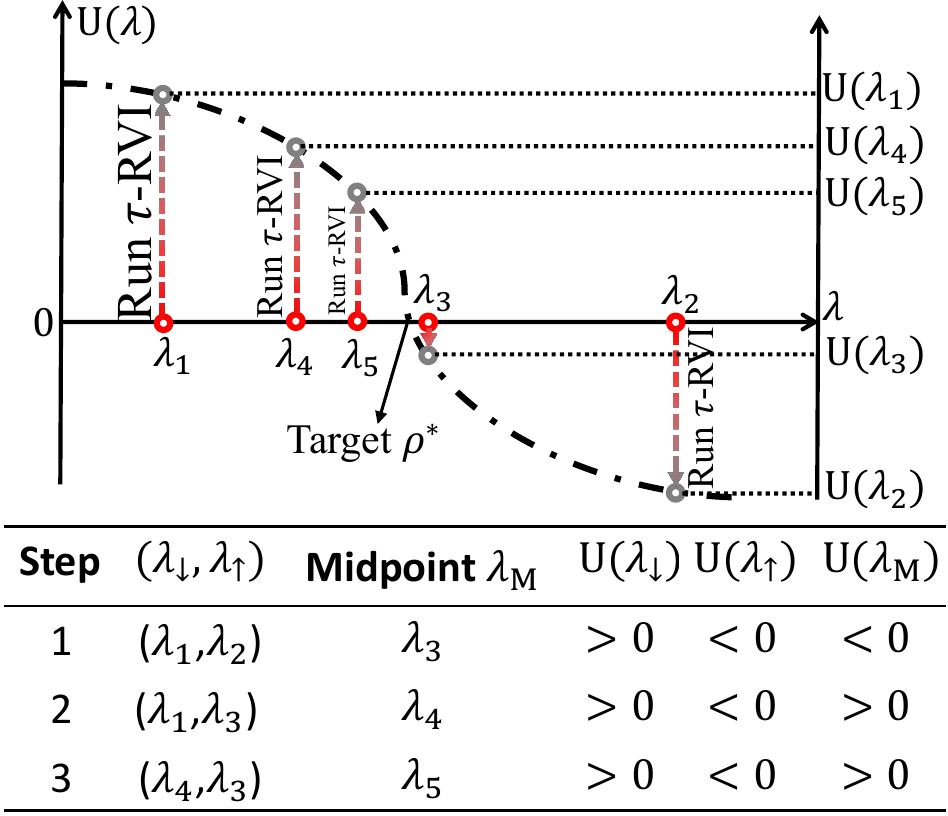}
				\caption{Bisection search to find the root of the implicit function $U(\lambda)$. The implicit function $U(\lambda)$ is approximated using a value iteration approach. The interval containing the root, denoted by $(\lambda_{\downarrow}, \lambda_{\uparrow})$, is halved at each outer-layer iteration, and this process eventually converges to the unique root $\rho^{\star}$.}
				\label{fig:figurealgorithm1}
			\end{figure}
			\begin{algorithm}[t]
				\caption{\textit{Two-layer approaches for $\rho^\star$}}
				\label{Algorithm 1}
				\LinesNumbered
				\KwIn{Tolerence $\epsilon>0$, MDP $\mathscr{P}_{\mathrm{MDP}}(\lambda)$}
				Initialization: $\lambda_{\uparrow}=\min_a\sum_{s\in\mathcal{S}}{\pi}_a(s)\cdot \mathcal{C}(s,a)$, $\lambda_{\downarrow}=\min_{s,a}\mathcal{C}(s,a)$\;
				\While {$\lambda_{\uparrow}-\lambda_{\downarrow}\ge \epsilon$}
				{
					$\lambda=(\lambda_{\uparrow}+\lambda_{\downarrow})/2$\;
					Run RVI (Iteration \ref{RVI}) or $\tau$-RVI (Iteration \ref{convergence:theorem2}) to solve $\mathscr{P}_{\mathrm{MDP}}(\lambda)$ and calculate $U(\lambda)$\;
					\If{$U(\lambda)>0$}
					{
						$\lambda_{\downarrow}=\lambda$\;
					}
					\Else
					{$\lambda_{\uparrow}=\lambda$\;}			
				}
				\KwOut{$\rho^{\star}=\lambda$}
			\end{algorithm}
			\subsection{Two-layer Approaches: Bisec-RVI and Bisec-$\tau$-RVI}\label{sectionIIC}
			Following Lemma \ref{l3}, solving Problem \ref{p3} reduces to solving Problem \ref{p4} to determine the optimal value $U(\lambda)$ while simultaneously finding the unique root $\rho^{\star}$ of the implicit function $U(\lambda)$, which can be solved using a {two-layer nested} algorithm as shown in Fig. \ref{fig:figurealgorithm1} and Algorithm \ref{Algorithm 1}. In the inner layer, value iteration approaches are applied to approximate the optimal value of Problem \ref{p4}, $U(\lambda)$, by resolving the reformulated MDP $\mathscr{P}_{\mathrm{MDP}}(\lambda)$ detailed in subsection \ref{MDP}. In the outer layer, a bisection search algorithm approximates the unique root $\rho^{\star}$. Conventionally, the RVI algorithm is applied in the inner layer to iteratively approximate $U(\lambda)$ \cite{li2024sampling,sun2019samplingwiener}, but it proves to be divergent in our formulation (see Fig. \ref{convergence1}). To address this limitation, we propose $\tau$-RVI in this paper and present a rigorous convergence analysis. The condition and the rationale of the convergence will be discussed in this section.

			\subsubsection{Inner-Layer MDP Given Dinkelbach Parameter $\lambda$}\label{MDP}
			To solve the inner layer sequential decision process in Problem \ref{p4}, we reformulate it as an equivalent standard infinite-horizon MDP. A standard MDP is typically described by a \textit{quadruple}: the state space, the action space, the transition probability, and the cost function. This subsection details this \textit{quadruple}. The MDP with Dinkelbach parameter $\lambda$ is denoted as $\mathscr{P}_{\mathrm{MDP}}(\lambda)$:
			\begin{itemize}
				\item \textbf{\ul{State Space}}: the state of the equivalent MDP is the sufficient statistics $\mathcal{G}_i=(X_{S_i},Y_i,A_{i-1})\in\mathcal{S}\times\mathcal{Y}\times\mathcal{A}$, as established in Lemma \ref{l1}.
				\item \textbf{\ul{Action Space}}: the action space of the MDP is composed of the tuple $(Z_i,A_i)\in\mathcal{\mathcal{Z}\times\mathcal{A}}$, where $Z_i$ is the sampling waiting time and $A_i$ is the controlled action that controls the source.
				\item \textbf{\ul{Transition Probability}}: The transition probability is defined by $\Pr(\mathcal{G}_{i+1}|\mathcal{G}_i,Z_i,A_i)$. We have the transition probability established in \eqref{tran}, \textcolor{black}{whose detailed proof is given in Appendix \ref{appendixc}}:
				\begin{equation}\label{tran}
					\begin{aligned}
						&\Pr(\mathcal{G}_{i+1}=(s',\delta',a')|\mathcal{G}_i=(s,\delta,a),Z_i,A_i)
						\\
						&=\Pr(Y_{i+1}=\delta')\cdot[\mathbf{P}_a^{\delta}\cdot\mathbf{P}_{A_i}^{Z_i}]_{s\times s'}\cdot\mathbbm{1}\{a'=A_i\},
					\end{aligned}
				\end{equation}
				where $[\mathbf{P}]_{s\times s'}$ denotes the element located at the $s$-th row and $s'$-th column of the matrix $\mathbf{P}$.
				
				\item \textbf{\ul{Cost Function}}: the cost function is typically a real-valued function over the state space and the action space. We denote it as $g(\mathcal{G}_i, Z_i, A_i)$ and will show how to tailor the cost function to establish equivalence with Problem \ref{p4}. 
				\begin{Lemma}\label{l4}
					If the cost function is defined by 
					\begin{equation}\label{cfun}
						\begin{aligned}
							g(\mathcal{G}_i,Z_i,A_i;\lambda)\triangleq q(\mathcal{G}_i,Z_i,A_i)-\lambda f(Z_i),
						\end{aligned}
					\end{equation}
					where
					\begin{align}\label{15}
						&q(\mathcal{G}_i,Z_i,A_i)\triangleq\notag\\&\mathbb{E}\left[\sum_{s'\in\mathcal{S}}\left[\sum_{t=0}^{Z_i+\mathrm{Y}_{i+1}-1} \mathbf{P}_{A_{i-1}}^{Y_i}\cdot\mathbf{P}_{A_{i}}^{t}\right]_{X_{S_i}\times s'}\cdot\mathcal{C}(s',A_i)\right],
					\end{align}
					with the expectation $\mathbb{E}$ taken over the random delay $Y_{i+1}$, and
					\begin{equation}
						f(Z_i)\triangleq\mathbb{E}\left[Y_{i+1}\right]+Z_i.
					\end{equation}
					
					Then Problem $\mathscr{P}_{\mathrm{MDP}}(\lambda)$ is equivalent to Problem \ref{p4}.
				\end{Lemma}
				\begin{proof}
					\textcolor{black}{See Appendix \ref{appendixd}.}
				\end{proof}
			\end{itemize}
			
			\textcolor{black}{In what follows, we refer to the MDP with transition dynamics $\mathbf{P}_a$ as the \emph{primal MDP}, 
			and denote by $\mathscr{P}_{\mathrm{MDP}}(\lambda)$ the transformed (or \emph{lifted}) MDP. 
			The next theorem provides a simple sufficient condition under which the lifted MDP $\mathscr{P}_{\mathrm{MDP}}(\lambda)$ 
			is \emph{unichain}\footnote{An MDP is said to be \emph{unichain} if, under any stationary policy, the induced Markov chain has a single recurrent class (with all other states being transient).}, 
			thereby ensuring the existence of an optimal stationary deterministic policy.
			\begin{theorem}[Sufficient condition for a unichain lifted MDP]\label{thm:lifted-unichain}
				Let $\mathcal S$, $\mathcal A$, and $\mathcal Y$ be the finite state, action, and delay sets of the primal MDP. 
				Suppose there exist a state $s^{\star}\in\mathcal S$, an integer $m\ge 1$, and a constant $\epsilon>0$ such that, 
				for every initial $(s,a)\in\mathcal S\times \mathcal A$ and for every admissible length-$m$ sequence 
				\(
				\{A_{t},Z_{t},\delta_t\}_{t=0}^{m-1}
				\):
				\begin{equation}\label{eq:entrance}
					\big[\mathbf P_{a}^{\delta_0}\mathbf P_{A_0}^{Z_0}\cdots 
					\mathbf P_{A_{m-2}}^{\delta_{m-1}}\mathbf P_{A_{m-1}}^{Z_{m-1}}\big]_{s\times s^\star}\ge\ \epsilon .
				\end{equation}
				Then, for \emph{every} stationary deterministic policy 
				$\pi:\mathcal S\times\mathcal Y\times\mathcal A\to\mathcal A\times\mathcal Z$, 
				the Markov chain on $\mathcal S\times\mathcal Y\times\mathcal A$ induced by $\pi$ has a single recurrent class 
				(all other states are at most transient). 
				In particular, the lifted MDP $\mathscr{P}_{\mathrm{MDP}}(\lambda)$ is \emph{unichain}.
			\end{theorem}
			\begin{proof}
				See Appendix \ref{appendixe}.
			\end{proof}}
			If the lifted MDP $\mathscr{P}_{\mathrm{MDP}}(\lambda)$ is a \emph{unichain}, an optimal \textit{stationary deterministic} policy exists, and one can establish the following \textit{Average Cost Optimality Equations} (ACOE) \cite[Eq. 4.1]{howard1960dynamic}:
			\begin{align}\label{eq21}
				\mathbf{ACOE}:\nonumber &V^\star(\gamma;\lambda)+U(\lambda)=\min _{A_i, Z_i}\Big\{g(\gamma,Z_i,A_i;\lambda)+\\& \mathbb{E}[V^\star\left(\gamma';\lambda\right)|\gamma,Z_i,A_i]\Big\}, \forall \gamma\in\mathcal{S}\times\mathcal{Y}\times\mathcal{A},
			\end{align}
			where $V^\star(\gamma;\lambda)\in\mathbb{R}$ is the optimal value function and $U(\lambda)\in\mathbb{R}$ is the optimal long-term average value of $\mathscr{P}_{\text{MDP}}(\lambda)$ in \eqref{p4eq}. By solving $U(\lambda)$ and $V^\star(\gamma;\lambda)$ for $\gamma\in\mathcal{S}\times\mathcal{Y}\times\mathcal{A}$ from the ACOE \eqref{eq21}, we can obtain the optimal \textit{stationary deterministic} sampling and remote decision-making policy for  $\mathscr{P}_{\text{MDP}}(\lambda)$, defined as: $\phi^{*}_{\lambda}:\mathcal{S}\times\mathcal{Y}\times\mathcal{A}\rightarrow\mathcal{Z}\times\mathcal{A}$.
			\begin{align}\label{eq19}
				\phi^{*}_{\lambda}(\gamma)=\mathop{\arg\min}_{Z_i,A_i} &\Big\{g(\gamma,Z_i,A_i;\lambda)+\nonumber\\& \mathbb{E}[V^\star\left(\gamma';\lambda\right)|\gamma,Z_i,A_i]\Big\}, \forall \gamma\in\mathcal{S}\times\mathcal{Y}\times\mathcal{A}.
			\end{align}
			The ACOE in \eqref{eq21} represents a series of nonlinear equations, which is mathematically intractable to solve explicitly. One can resort to the typical Dynamic Programming (DP)-like RVI algorithm \cite{white1963dynamic} to iteratively generate the sequences $\{U_{K}(\lambda)\}^{K\in\mathbb{N}^+}$ and $\{V_{K}(\gamma;\lambda)\}_{\gamma\in\mathcal{S}\times\mathcal{Y}\times\mathcal{A}}^{K\in\mathbb{N}^+}$ that conditionally converge asymptotically to the solutions of the ACOE. 
			\begin{iteration}\label{RVI}
				({{RVI} Algorithm \cite{white1963dynamic}}). For a given $\lambda$, the RVI starts with a given initial value $\{V_0(\gamma;\lambda)\}_{\gamma\in\mathcal{S}\times\mathcal{Y}\times\mathcal{A}}$ and computes $U_{K+1}(\lambda)$ and $V_{K+1}(\gamma;\lambda)$ iteratively:		
				\begin{subequations}\label{RVIiteration}
					\begin{align}
						U_{K+1}(\lambda) &= \min_{A_i, Z_i} \Big\{\, g(\gamma^{\text{r}}, Z_i, A_i; \lambda)\notag \\
						&\quad +\, \mathbb{E}\big[ V_{K}(\gamma'; \lambda) \mid \gamma^{\text{r}}, Z_i, A_i \big] \Big\}, \\
						V_{K+1}(\gamma; \lambda) &= \min_{A_i, Z_i} \Big\{\, g(\gamma, Z_i, A_i; \lambda)\notag \\
						&\quad +\, \mathbb{E}\big[ V_{K}(\gamma'; \lambda) \mid \gamma, Z_i, A_i \big] \Big\}\notag \\
						&\quad -\, U_{K+1}(\lambda), \quad \forall \gamma \in \mathcal{S} \times \mathcal{Y} \times \mathcal{A}.
					\end{align}
				\end{subequations}		
				where $\gamma^r\in\mathcal{S}\times\mathcal{Y}\times\mathcal{A}$ is an arbitrarily chosen \textit{reference state} and the conditional expectation $\mathbb{E}$ is taken with respect to the conditional distribution of the next state $\gamma'$ given the current state and the current action. The iterative process continues until a predefined convergence criterion is satisfied.
			\end{iteration}
			\begin{figure*}[t!]
				\centering
				\includegraphics[width=0.88\linewidth]{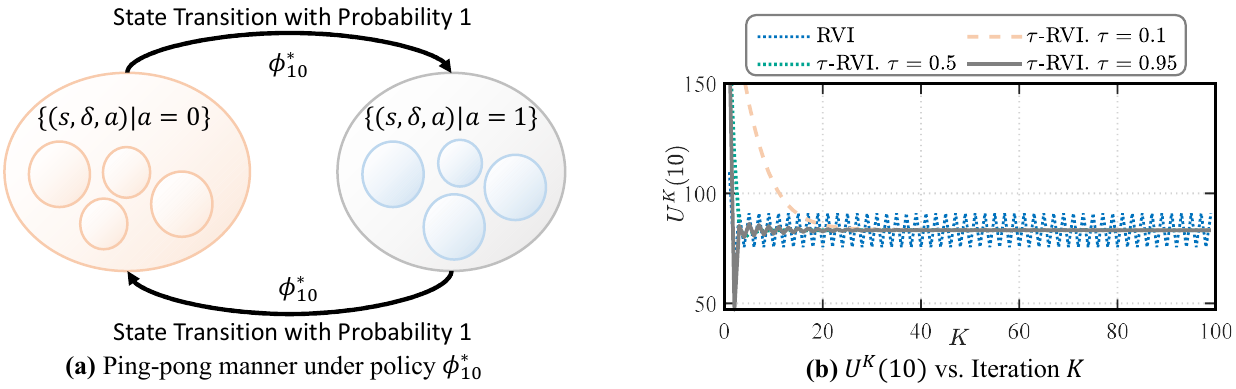}
				\caption{Algorithmic behavior of RVI versus $\tau$-RVI: Divergence mechanisms and comparative performance.}
				\label{fig4}
			\end{figure*}
			
			However, in our specific context, the convergence of the RVI algorithm is not necessarily guaranteed (see Example \ref{ex1}). To further investigate this issue and develop a convergent alternative, we present sufficient conditions for its convergence in the following Lemma.
			\begin{Lemma}\label{Lemma4}
				(Restatement of \cite[Proposition 4.3.2]{bertsekas2012dynamic2}). If an MDP satisfies the following conditions: \\
				(a) \textit{the MDP is a {unichain} MDP}; \\
				(b) \textit{the optimal stationary policy for the MDP yields an \textit{aperiodic} transition probability matrix};\\
				then the sequences $\{U_{K}(\lambda)\}_{K\in\mathbb{N}}$ and $\{V_{K}(\gamma;\lambda)\}_{K\in\mathbb{N}}$ will converge to the solution to the ACOE \eqref{eq21}:
				\begin{equation}
					\begin{aligned}
						&\lim_{K\rightarrow \infty} U_K(\lambda)=U(\lambda)\\
						&\lim_{K\rightarrow \infty} V_K(\gamma;\lambda)=V^{\star}(\gamma;\lambda),\forall \gamma\in\mathcal{S}\times\mathcal{Y}\times\mathcal{A}.
					\end{aligned}
				\end{equation}
			\end{Lemma}
			
			In the transformed MDP $\mathscr{P}_{\text{MDP}}(\lambda)$, condition (a) in Lemma \ref{Lemma4} is ensured by Theorem \ref{thm:lifted-unichain}. Nevertheless, the following \textbf{counter-example} demonstrates that condition (b) does not necessarily hold in $\mathscr{P}_{\text{MDP}}(\lambda)$. Specifically, even if the \textit{primal MDP} characterized by $\left\langle \mathcal{S},\mathcal{A},\mathbf{P}_a,\mathcal{C}\right\rangle $ is \textit{aperiodic}, the transformed MDP $\mathscr{P}_{\text{MDP}}(\lambda)$ may still exhibit \textit{periodicity}, which will cause the RVI algorithm to diverge. 
			\begin{example}\label{ex1}(A divergence example of RVI).
				Consider the parameter setup described in Appendix \ref{appendixh} where the delay is constant with $p=0$. In this case, the optimal policy for the transformed MDP $\mathscr{P}_{\text{MDP}}(10)$ is:
				\begin{subequations}\label{pingpong}
					\begin{align}
						&\phi^{*}_{10}(\overbrace{0}^{X_{S_i}},\overbrace{10}^{Y_i},\overbrace{0}^{A_{i-1}})=(\overbrace{0}^{Z_i},\overbrace{1}^{A_i}),\\
						&\phi^{*}_{10}(\overbrace{1}^{X_{S_i}},\overbrace{10}^{Y_i},\overbrace{0}^{A_{i-1}})=(\overbrace{0}^{Z_i},\overbrace{1}^{A_i}),\\		
						&\phi^{*}_{10}(\overbrace{0}^{X_{S_i}},\overbrace{10}^{Y_i},\overbrace{1}^{A_{i-1}})=(\overbrace{0}^{Z_i},\overbrace{0}^{A_i}),\\ &\phi^{*}_{10}(\overbrace{1}^{X_{S_i}},\overbrace{10}^{Y_i},\overbrace{1}^{A_{i-1}})=(\overbrace{0}^{Z_i},\overbrace{0}^{A_i}).
					\end{align}
				\end{subequations}		
				The optimal stationary policy in \eqref{pingpong} induces the sub-state sequence $\{A_i\}_{i \in \mathbb{N}^+}$ to alternate in a $(0, 1, 0, 1, \dots)$ ping-pong manner, as shown in Fig. \ref{fig4}-\textbf{(a)}. This alternating behavior results in the RVI oscillations as shown in Fig.  \ref{fig4}-\textbf{(b)}.	\end{example} 	
			
			Lemma \ref{Lemma4} and Example \ref{ex1} show that the RVI algorithm \cite{white1963dynamic} may not asymptotically converge to $U(\lambda)$, and the reason is the \textit{periodic} nature inherent in the transformed MDP $\mathscr{P}_{\text{MDP}}(\lambda)$. Consequently, the existing two-layer Bisec-RVI Algorithm (e.g., \cite[Algorithm 1]{li2024sampling}, \cite[Algorithm 1]{10807024}, and \cite[Algorithm 1]{bedewy2021optimal}) cannot reliably determine the root $\rho^\star$. To circumvent this
			problem, we propose a new iterative approach, namely $\tau$-RVI, in Iteration \ref{convergence:theorem2}. This approach eliminates the need for condition (b) in Lemma \ref{Lemma4} but guarantees rigorous convergence. 
			\begin{iteration}\label{convergence:theorem2} ($\tau$-{RVI} Algorithm). For a given $\lambda$ and a parameter  $0<\tau\le1$, the $\tau$-RVI iteratively generate sequences $\{\tilde{U}_{K}(\lambda)\}^{K\in\mathbb{N}^+}$ and $\{\tilde{V}_{K}(\gamma;\lambda)\}_{\gamma\in\mathcal{S}\times\mathcal{Y}\times\mathcal{A}}^{K\in\mathbb{N}^+}$ with a starting initial value $\{\tilde{V}_0({\gamma;\lambda})\}_{\gamma\in\mathcal{S}\times\mathcal{Y}\times\mathcal{A}}$.
				\begin{subequations}\label{MRVI}
					\begin{align}
						\tilde{U}_{K+1}(\lambda) &= \min_{A_i, Z_i} \bigg\{\, g(\gamma^{\text{r}}, Z_i, A_i; \lambda)\notag \\
						&\quad\quad\quad + \tau\, \mathbb{E}\big[ \tilde{V}_{K}(\gamma'; \lambda) \mid \gamma^{\text{r}}, Z_i, A_i \big] \bigg\},\label{23a} \\
						\tilde{V}_{K+1}(\gamma; \lambda) &= (1 - \tau)\, \tilde{V}_{K}(\gamma; \lambda)\notag \\
						&\quad + \min_{A_i, Z_i} \bigg\{\, g(\gamma, Z_i, A_i; \lambda)\notag \\
						&\quad+ \tau\, \mathbb{E}\big[ \tilde{V}_{K}(\gamma'; \lambda) \mid \gamma, Z_i, A_i \big] \bigg\}\notag \\&\quad- \tilde{U}_{K+1}(\lambda), \forall \gamma \in \mathcal{S} \times \mathcal{Y} \times \mathcal{A},\label{23b}
					\end{align}
				\end{subequations}			 		 
				where $\gamma^{\text{r}}{\in\mathcal{S}\times\mathcal{Y}\times\mathcal{A}}$ is a predefined fixed \textit{reference state} with initial condition $\tilde{V}_{0}(\gamma^{\text{r}};\lambda)=0$.
			\end{iteration}
			\begin{remark}
				If $\tau=1$, then $\tau$-{RVI} reduces to {RVI}.
			\end{remark}
			
			Fig. \ref{fig4}-\textbf{(b)} illustrates the convergence of $\tau$-RVI across various values of the parameter $\tau$, compared to the standard RVI algorithm \cite{white1963dynamic}. The results demonstrate that $\tau$-RVI overcomes oscillatory behavior encountered by RVI, with convergence rates depending on the selected $\tau$ values. A rigorous convergence analysis of $\tau$-RVI will be presented in in \cref{ctaorvi}.
			
			\subsubsection{Outer-Layer Bisection and Bounds on $\rho^\star$}
			
			In the outer layer of Algorithm \ref{Algorithm 1}, the search interval $(\lambda_{\downarrow},\lambda_{\uparrow})$ is bisected on a slow time scale to approximate the root of $U(\lambda)$. This process relies on the uniqueness of the root of $U(\lambda)$ established in Lemma \ref{l3}. For the bisection search process, the complexity mainly depends on the initialization of the search interval $(\lambda_{\downarrow},\lambda_{\uparrow})$, which requires establishing upper and lower bounds on $\rho^\star$. To address this, we establish the bounds on $\rho^{\star}$ to initialize the bisection search:
			\begin{Lemma}[Upper and Lower Bounds on $\rho^\star$]\label{l5}
				The lower bound of $\rho^\star$ can be defined by the minimum value of the cost function, given by \begin{equation}\rho^\star\ge\min_{s,a}\mathcal{C}(s,a).\end{equation}
				The upper bound of $\rho^\star$ can be defined by the minimum stationary cost achievable under a constant action, given by 
				\begin{equation}
					\rho^\star\le\min_a\sum_{s\in\mathcal{S}}{\pi}_a(s)\cdot \mathcal{C}(s,a),
				\end{equation}
				where ${\pi}_a(s)$ represents the stationary distribution of state $s$, corresponding to the transition probability matrix $\mathbf{P}_a$.
			\end{Lemma}
			\begin{proof}
				\textcolor{black}{See Appendix \ref{appendixf}.}
			\end{proof}
			\subsection{Convergence of $\tau$-RVI}\label{ctaorvi}
			In this subsection, we present convergence results for $\tau$-RVI in Iteration \ref{convergence:theorem2}. We theoretically show that the generated sequences in $\tau$-RVI will asymptotically approach the solution to the ACOE \eqref{eq21}. To quantify the convergence behavior, we define the \textit{relative error} for the $K$-th iteration value $\tilde{U}_K(\lambda)$ with respect to the ACOE solution $U(\lambda)$ as: \begin{equation}
				e_{\mathrm{U}}^{(K)}(\lambda)\triangleq|\tilde{U}_{K}(\lambda)-U(\lambda)|.
			\end{equation} Similarly, define the relative error for the $K$-th iteration value $\tilde{V}_K(\gamma;\lambda)$ with respect to the ACOE solution $V^\star(\gamma;\lambda)$ as: \begin{equation}
				e_{\mathrm{V}}^{(K)}(\gamma;\tau,\lambda)\triangleq\left|\tilde{V}_{K}(\gamma;\lambda)-\frac{V^\star(\gamma;\lambda)}{\tau}\right|.
			\end{equation} The main convergence results for $\tau$-RVI are summarized in Theorem \ref{convergence1} below.		
			\begin{theorem}\label{convergence1}
				(Convergence of $\tau$-RVI). If the MDP $\mathscr{P}_{\text{MDP}}(\lambda)$ is a \textit{unichain} MDP, then the following limits hold true:\\
				($i$). For $\forall \lambda\in\mathbb{R}$, 
				\begin{equation}
					\lim_{K \to \infty}e_{\mathrm{U}}^{(K)}(\lambda)=0.
				\end{equation}
				($ii$). For $\forall \lambda\in\mathbb{R}$, $\forall 0<\tau<1$, and $\forall \gamma\in\mathcal{S}\times\mathcal{Y}\times\mathcal{A}$,
				\begin{equation}
					\lim_{K \to \infty}e_{\mathrm{V}}^{(K)}(\gamma;\tau,\lambda)=0.
				\end{equation}
			\end{theorem}
			\begin{proofsketch}
				The motivation behind \eqref{MRVI} is to formulate an alternative MDP problem, denoted as $\widetilde{\mathscr{P}_{\text{MDP}}}(\lambda)$, by eliminating potential \textit{periodicity} in the transition probabilities in \eqref{tran}. We denote the transition probability from state $i$ to state $j$, given action $a$, as $p_{ij}(a)$ in ${\mathscr{P}_{\text{MDP}}}(\lambda)$, and as $\widetilde{p_{ij}}(a)$ in $\widetilde{\mathscr{P}_{\text{MDP}}}(\lambda)$. The alternative transition probability $\widetilde{p_{ij}}(a)$ is defined as:
				\begin{equation}\label{ptransform}
					\widetilde{p_{ij}}(a)\triangleq\begin{cases}
						\tau{p}_{ij}(a), & \text{if } i\ne j\\
						1-\tau+\tau p_{ii}(a), &\text{if } i=j.
					\end{cases}
				\end{equation}
				It is easy to verify that $\sum_{j}\widetilde{p_{ij}}(a)=1$ and $\widetilde{p_{ii}}(a)>0$ for $\forall i\in\mathcal{S}\times\mathcal{Y}\times\mathcal{A}$, and thus the new MDP $\widetilde{\mathscr{P}_{\text{MDP}}}(\lambda)$ is \textit{aperiodic}. Then, we can formulate the ACOE for the alternative $\widetilde{\mathscr{P}_{\text{MDP}}}(\lambda)$:
				\begin{equation}\label{eq26}
					\begin{aligned}			\tilde{V}^\star(\gamma;\lambda)&+\tilde{U}(\lambda)=\min _{A_i, Z_i}\Big\{g(\gamma,Z_i,A_i;\lambda)+\\ &\sum_{\gamma'}\widetilde{p_{\gamma \gamma'}}(Z_i,A_i)\tilde{V}^\star(\gamma';\lambda)\Big\}, \forall \gamma\in\mathcal{S}\times\mathcal{Y}\times\mathcal{A},
					\end{aligned}
				\end{equation}
				which can be solved by the traditional RVI with guaranteed convergence because of its \textit{unichain} and \textit{aperiodic} property. Then, comparing \eqref{eq21} and \eqref{eq26}, we can establish the relationship that $\tau\tilde{V}^\star(\gamma;\lambda)={V}^\star(\gamma;\lambda)$ for $\forall\gamma\in\mathcal{S}\times\mathcal{Y}\times\mathcal{A}$ and $\tilde{U}(\lambda)={U}(\lambda)$. See Appendix \ref{proof:convergence:theorem2} for the detailed proof.
			\end{proofsketch}
			
			The next theorem characterizes an upper bound on the \textit{relative error} of the $\tau$-RVI Algorithm, whose proof is included in Appendix \ref{proof:convergence:theorem2}.
			\begin{theorem} \label{theorem3:convergence}(Upper Bound of Relative Error).
				If the MDP $\mathscr{P}_{\text{MDP}}(\lambda)$ is a unichain MDP, then up to iteration $K$, the relative error $e_{\mathrm{U}}^{(K)}(\lambda)$ is upper bounded above by
				\begin{equation}
					e_{\mathrm{U}}^{(K)}(\lambda)\le\frac{\tau M(1-\epsilon)^{(K-1)/L}}{1-(1-\epsilon)^{1/L}}=\mathcal{O}\left(\frac{1}{R^K}\right),
				\end{equation}
				where $M$ is a scaling factor, $L$ is defined by \eqref{Ldefinition}. The term $R=\frac{1}{(1-\epsilon)^{1/L}}>1$ captures the asymptotic convergence rate.
			\end{theorem}
			Theorem \ref{theorem3:convergence} demonstrates that the upper bound of the \textit{relative error} decreases \textbf{exponentially} with respect to the number of iterations $K$. This indicates that the proposed method exhibits faster convergence, as the number of inner iterations required to achieve a given accuracy $\delta$ is at most \textbf{logarithmic}, \emph{i.e.}, $K\le\mathcal{O}\left(\log(1/\delta)\right)$.
			
			
			\section{Optimal Sampling Without Rate Constraint: A One-layer Primal-Dinkelbach Approach}\label{sectionIV2}
			The two-layer algorithm requires repeatedly executing the \textit{computation-intensive} $\tau$-RVI to evaluate $U(\lambda)$ at the fast time scale and update $\lambda$ on the slow timescale. This results in high complexity since each update of $\lambda$ in the outer-layer search necessitates running the inner layer value iterations. To address this, we develop efficient \textit{one-layer} iterations in this section that eliminate the need for outer-layer bisection search. The key idea here is to treat the constraint $U(\rho^\star) = 0$ as an intrinsic condition within the Markov Decision Process (MDP) $\mathscr{P}_{\mathrm{MDP}}(\rho^\star)$, allowing us to directly establish the \textit{optimality equations}, which take the form of \textit{fixed-point equations}, as shown in \cite{li2024sampling}. In this work, we show through an example that the \textit{fixed-point operation} in \cite{li2024sampling} may not converge due to the \textit{non-contractive} nature of the operator (see Section \ref{noconverge}). To address divergence, we develop a new one-layer iterative algorithm that guarantees provable convergence to the \textit{fixed point} (see Section \ref{iiC}).
			
			\subsection{{Fixed-Point Equations and Iterations}}\label{noconverge}
			We demonstrate here that the root finding process on $U(\lambda)$ is equivalent to solving the following non-linear equations, which we will then show to be \textit{fixed-point equations}.
			\begin{theorem}\label{the2}
				Solving Problem $\mathscr{P}_{\mathrm{MDP}}(\rho^{\star})$ with $U(\rho^{\star})=0$ is equivalent to solving the following nonlinear equations:
				\begin{equation}\label{eqfunction}
					\left\{
					\begin{aligned}
						&W^\star(\gamma)=\min _{A_i, Z_i}\Big\{g(\gamma,A_i,Z_i;\rho^{\star})+\mathbb{E}[W^\star\left(\gamma'\right)|\gamma,Z_i,A_i]\Big\},\\&\quad\quad\quad\quad\quad\quad\quad\quad\quad\quad\quad\quad\quad\quad\quad\quad\quad \gamma\in\mathcal{S}\times\mathcal{Y}\times\mathcal{A},\\
						&\rho^{\star}=\min _{A_i, Z_i}\left\{\frac{q\left(\gamma^{\mathrm{r}}, A_i, Z_i\right)+
							\mathbb{E}[W^\star\left(\gamma'\right)|\gamma^{\mathrm{r}}, A_i, Z_i]}{f\left(Z_i\right)}\right\},\\
					\end{aligned}\right.
				\end{equation}
				where $\gamma^{\mathrm{r}}\in\mathcal{S}\times\mathcal{Y}\times\mathcal{A}$ is the reference state and can be arbitrarily chosen.
			\end{theorem}
			\begin{proof}
				\textcolor{black}{See Appendix \ref{appendixg}.}
			\end{proof}
			
			In \eqref{eqfunction}, there are $|\mathcal{S}|\times|\mathcal{Y}|\times|\mathcal{A}|+1$ variables, \emph{i.e.}, $\rho^{\star}$ and $W^\star(\gamma),\gamma\in\mathcal{S}\times\mathcal{Y}\times\mathcal{A}$, which are matched by an equal number of equations. Solving non-linear equations is generally challenging, however, we establish that the equations (\ref{eqfunction}) are \textit{fixed-point equations}. Let $\mathbf{W}^\star$ denote the vector consisting of $W^\star(\gamma)$ for all $\gamma\in\mathcal{S}\times\mathcal{Y}\times\mathcal{A}$, (\ref{eqfunction}) can be succinctly represented as follows:
			\begin{equation}\label{eqfunctiontransformation}
				\left\{
				\begin{aligned}
					&\mathbf{W}^\star=T(\mathbf{W}^\star,\rho^{\star})\\
					&\rho^{\star}=H(\mathbf{W}^\star),
				\end{aligned}\right.
			\end{equation}
			where $T(\cdot):\mathbb{R}^{|\mathcal{S}\times\mathcal{Y}\times\mathcal{A}|}\times\mathbb{R}^+\rightarrow\mathbb{R}^{|\mathcal{S}\times\mathcal{Y}\times\mathcal{A}|}$ is a non-linear operator corresponding to the first equations of \eqref{eqfunctiontransformation} and $H(\cdot):\mathbb{R}^{|\mathcal{S}\times\mathcal{Y}\times\mathcal{A}|}\rightarrow\mathbb{R}^+$ is an operator corresponding to the second equation of \eqref{eqfunctiontransformation}. Substituting the second equation $\rho^{\star}=H(\mathbf{W}^\star)$ into the first equation of (\ref{eqfunctiontransformation}) yields $\mathbf{W}^\star=T(\mathbf{W}^\star,H(\mathbf{W}^\star))$. Define a new operator $G$ as $G(\mathbf{W})\triangleq T(\mathbf{W},H(\mathbf{W}))$, as implied by (\ref{eqfunctiontransformation}), we have the \textit{fixed-point equation}: \begin{equation}\label{eq31}
				\mathbf{W}^\star=G(\mathbf{W}^\star).
			\end{equation}
			Then, as was done in \cite[Algorithm 2]{li2024sampling}, we can conduct the following \textit{fixed-point iterations} to asymptotically approximate the solution to \eqref{eqfunction}:
			\begin{equation}\label{FPBI}
				\left\{\begin{aligned}
					&\mathbf{W}^{K+1}=G(\mathbf{W}^K), \\
					&\rho_{K+1}=H(\mathbf{W}^{K+1}).
				\end{aligned}\right.
			\end{equation}
			Typically, the convergence of fixed-point iterations is established using the \textit{Banach Contraction Mapping Theorem} \cite[Theorem 6]{banach1922operations}. According to this theorem, if the operator $Q(\cdot): \mathbb{R}^{|\mathcal{S}\times\mathcal{Y}\times\mathcal{A}|}\rightarrow\mathbb{R}^{|\mathcal{S}\times\mathcal{Y}\times\mathcal{A}|}$ is a \textit{contraction mapping}, \emph{i.e.}, for any $\mathbf{W}_1,\mathbf{W}_2\in\mathbb{R}^{ |\mathcal{S}\times\mathcal{Y}\times\mathcal{A}|}$, there exists a constant $0<\alpha<1$ such that
			\begin{equation}\label{alpha}
				||G(\mathbf{W}_1)-G(\mathbf{W}_2)||_2\le\alpha||\mathbf{W}_1-\mathbf{W}_2||_2, 
			\end{equation}
			then the operator $G$ has a unique fixed point, and the fixed-point iteration described in \eqref{FPBI} converges to the solution $\mathbf{W}^\star$, where $\mathbf{W}^\star=G(\mathbf{W^\star})$. 
			
			However, the following \textbf{counter-example} shows that the operator $G: \mathbb{R}^{|\mathcal{S}\times\mathcal{Y}\times\mathcal{A}|}\rightarrow\mathbb{R}^{|\mathcal{S}\times\mathcal{Y}\times\mathcal{A}|}$ may be a \textit{non-expansive mapping}\footnote{\textit{Non-expansive mapping} indicates $\alpha=1$ in \eqref{alpha}. Fig. \ref{fig:convergenceFPBI} demonstrates that the fix-point operation in \eqref{FPBI} may be \textit{non-expansive}.} rather than a \textit{contraction mapping}, which results in potential divergence.
			\begin{example}
				(Divergence Example of FPBI \cite[Algorithm 2]{li2024sampling}). Consider the parameter setup described in Appendix \ref{appendixh} where the delay is $p=0$. In this case, the fixed-point iteration starting from $\mathbf{W}^0=\mathbf{0}$ will diverge, as shown in Fig. \ref{fig:convergenceFPBI}.
			\end{example}
			\begin{figure}
				\centering
				\includegraphics[width=0.9\linewidth]{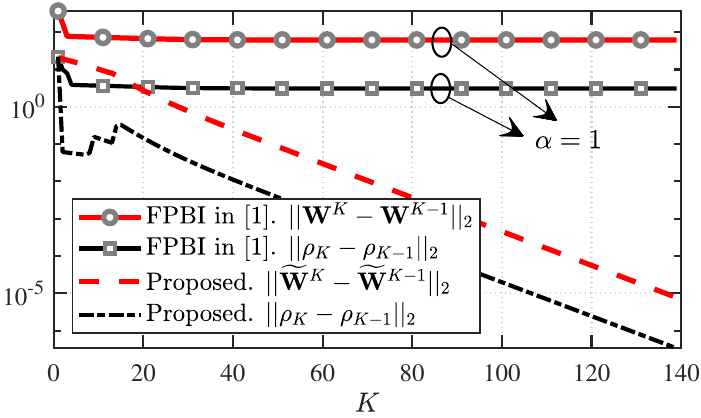}
				\caption{Convergence comparison between FPBI \cite[Algorithm 2]{li2024sampling} and \textsc{OnePDSI} (Iteration \ref{propsition2}). FPBI becomes \textit{non-expansive} and thus diverges. The proposed \textsc{OnePDSI} converges.}
				\label{fig:convergenceFPBI}
			\end{figure}
			
			\subsection{Primal-Dinkelbach Synchronous Iteration (Proposed \normalfont{\textsc{OnePDSI}})}\label{iiC}
			To overcome the divergence limitations of FPBI, we develop a novel one-layer iterative approach that guarantees rigorous convergence to the solution of the \textit{fixed-point equations} \eqref{eqfunction}. The details of the iteration are given as:
			\begin{iteration}\label{propsition2} ({\normalfont{\textsc{OnePDSI}}}): For a given $0<\kappa<1$, we can iteratively generate sequences $\{\rho_{K}\}^{K\in\mathbb{N}^+}$ and $\{\widetilde{W}_{K}(\gamma)\}_{\gamma\in\mathcal{S}\times\mathcal{Y}\times\mathcal{A}}^{K\in\mathbb{N}^+}$ with a starting initial value $\{\widetilde{W}_{0}(\gamma)\}_{\gamma\in\mathcal{S}\times\mathcal{Y}\times\mathcal{A}}$:
				\begin{subequations}\label{prop2}
					\begin{align}
						\rho_{K+1} &= \min_{A_i, Z_i} \Bigg\{ \frac{q(\gamma^{\text{r}}, Z_i, A_i) - \kappa \widetilde{W}_K(\gamma^{\text{r}}) \cdot \mathbb{E}[Y_i]}{f(Z_i)} \notag\\
						&+ \frac{\kappa \mathbb{E}\left[\widetilde{W}_K(\gamma') \mid \gamma^{\text{r}}, Z_i, A_i \right] \cdot \mathbb{E}[Y_i]}{f(Z_i)} \Bigg\}+\widetilde{W}_K(\gamma^{\text{r}}), \\
						\widetilde{W}_{K+1}(\gamma) &= \min_{A_i, Z_i} \Bigg\{ \frac{q(\gamma, Z_i, A_i) - \kappa \widetilde{W}_K(\gamma) \cdot \mathbb{E}[Y_i]}{f(Z_i)} \notag\\
						& + \frac{\kappa \mathbb{E}\left[\widetilde{W}_K(\gamma') \mid \gamma, Z_i, A_i \right] \cdot \mathbb{E}[Y_i]}{f(Z_i)} \Bigg\} \notag\\
						& +\widetilde{W}_K(\gamma) - \rho_{K+1}, \quad\forall \gamma \in \mathcal{S} \times \mathcal{Y} \times \mathcal{A}.\label{38b}
					\end{align}
				\end{subequations}	
				where $\gamma^{\text{r}}{\in\mathcal{S}\times\mathcal{Y}\times\mathcal{A}}$ is a fixed \textit{reference state} with an initial condition $\widetilde{W}_{0}(\gamma^{\text{r}})=0$.
			\end{iteration}
			
			In what follows, we present a comprehensive convergence analysis demonstrating that \textsc{OnePDSI} ensures robust and provable convergence.
			
			\subsection{Convergence of \textsc{OnePDSI}}\label{iiD}
			In this subsection, we theoretically demonstrate that the sequences in \textsc{OnePDSI} will approach the solution to \eqref{eqfunction}. To quantify the convergence, we define the \textit{relative error} of ${\rho_K}$ at iteration $K$ as:
			\begin{equation}
				e_{\rho}^{(K)}=|\rho_K-\rho^\star|.
			\end{equation}
			Meanwhile, define the \textit{relative error} of the sequence $\widetilde{W}_{K}(\gamma)$  at $K$-th iteration as
			\begin{equation}
				e_{\mathrm{W}}^{(K)}(\gamma;\kappa)=\left|\widetilde{W}_K(\gamma)-\frac{{W}^\star(\gamma)}{\kappa\cdotp\mathbb{E}[Y_i]}\right|.
			\end{equation}
			The following theorem demonstrates the convergence of \textsc{OnePDSI}.
			\begin{theorem}\label{Lemma6}
				(Convergence of \textsc{OnePDSI}). If the transformed MDP is unichain and $0<\kappa<1$, then the \textsc{OnePDSI} in \eqref{prop2} is convergent, with:
				\begin{equation}
					\begin{aligned}
						&\lim_{k \to \infty}	e_{\rho}^{(K)}=0,\\
						&\lim_{k \to \infty} e_{\mathrm{W}}^{(K)}(\gamma)=0, \forall \gamma\in\mathcal{S}\times\mathcal{Y}\times\mathcal{A}.
					\end{aligned}
				\end{equation}	
			\end{theorem}
			\begin{proof}
				See Appendix \ref{appendixJ}.
			\end{proof}
			
			The next theorem characterizes an upper bound on the \textit{relative error} of the proposed \textsc{OnePDSI}, whose proof is provided in Appendix \ref{appendixJ}:
			\begin{theorem} (Upper Bound of Relative Error). \label{the6upperbound2}
				If the MDP $\mathscr{P}_{\text{MDP}}(\lambda)$ is a unichain MDP, then up to iteration $K$, the relative error \textcolor{black}{$e_{\mathrm{\rho}}^{(K)}$} is upper bounded above by
				\begin{equation}\color{black}
					e_{\rho}^{(K)}\le\frac{\tau M(1-\epsilon)^{(K-1)/L}}{1-(1-\epsilon)^{1/L}}=\mathcal{O}\left(\frac{1}{R^{K}}\right),
				\end{equation}
				where $M$ is a scaling factor, $L$ is defined by \eqref{181}. The term $R=\frac{1}{(1-\epsilon)^{1/L}}$ captures the asymptotic convergence rate. 
			\end{theorem}
			
			Theorem \ref{the6upperbound2} demonstrates that the upper bound of the \textit{relative error} decreases \textbf{exponentially} with respect to the number of iterations $K$. This indicates that the number of inner iterations required to achieve a given optimality gap $\delta$ is at most \textbf{logarithmic}, \emph{i.e.}, $K\le\mathcal{O}\left(\log(1/\delta)\right)$.


				
				\section{Optimal Sampling With Rate Constraint: A Typical Three-Layer Approach}\label{sectionV}
				In this section, we investigate the optimal sampling and decision-making policy that minimizes the long-term average cost under a sampling frequency constraint, as formulated in Problem~\ref{p1}. Our goal is to derive this optimal policy and its corresponding value $h^{\star}$.  
				\subsection{Lagrangian Dual Techniques}
				
				Following the steps in \eqref{eq6}-\eqref{p4eq} and applying Dinkelbach's method for non-linear fractional programming as in \cite{dinkelbach1967nonlinear} and \cite[Lemma 2]{sun2019samplingwiener}, the problem of determining the optimal policy for Problem \ref{p1} is equivalent to solving the following alternative problem given the \textit{Dinkelbach} parameter $\lambda$:
				
				\begin{problem}[\textit{Standard Infinite-Horizon Constrained Markov Decision Process (CMDP) with Dinkelbach Parameter $\lambda$}]\label{p5}
					\begin{equation}
						\begin{aligned}
							&H(\lambda;f_{\max})\triangleq\\
							&\inf _{\boldsymbol{\psi}} \lim _{\mathrm{n} \rightarrow \infty}\frac{1}{n} {\sum_{i=0}^{n-1}\left\{ \mathbb{E}_{\boldsymbol{\psi}}\left[\sum_{t=\it{D}_{i}}^{\it{D}_{i+1}-1} \mathcal{C}(X_t, A_i)\right]-\lambda\mathbb{E}_{\boldsymbol{\psi}} \left[Z_i+Y_{i+1}\right]\right\}}\\
							&\text{s.t. } \lim_{T \to \infty} \frac{1}{T} \mathbb{E}_{\boldsymbol{\psi}}\left[\sum_{i=1}^{T}(S_{i+1}-S_i)\right] \geq \frac{1}{f_{\text{max}}},
						\end{aligned}
					\end{equation}
				\end{problem}
				The following lemma characterizes the relationship between $h^\star$ and the optimal value of Problem~\ref{p5}. The proof follows directly from modifying the policies in \cite[Appendix B]{li2024sampling} to satisfy the sampling frequency constraint, and is therefore omitted.
				\begin{Lemma}\label{l5v2}
					The following assertions hold:\\
					(i). $h^\star \gtreqless \lambda \text { if and only if } H(\lambda;f_{\max}) \gtreqless 0 \text {. }$\\
					(ii). When $H(\lambda;f_{\max})=0$, the solutions to Problem \ref{p5} coincide with those of Problem \ref{p1}.\\
					(iii). $H(\lambda;f_{\max})=0$ has a unique root, and the root is $h^\star$.
				\end{Lemma}
				
				With Lemma \ref{l5v2} in hand, solving Problem \ref{p1} is equivalent to solving for the root of the implicit function $H(\lambda;f_{\max})$. To obtain $H(\lambda;f_{\max})$ given $\lambda$, we first transform the CMDP into an unconstrained Lagrangian MDP problem. Specifically, define the \textit{Lagrange Function} as:
				\begin{equation}\label{eq37}
					\begin{aligned}
						&\mathcal{L}(\boldsymbol{\psi};\theta,\lambda,f_{\max})=\frac{\theta}{f_{\max}}+\\&\lim _{n \rightarrow \infty} \frac{1}{n} \sum_{i=0}^{n-1}\mathbb{E}_{\boldsymbol{\psi}}\left\{\sum_{t=D_i}^{D_{i+1}-1} \mathcal{C}\left(X_t, A_i\right)-(\lambda+\theta) \left(Z_i+Y_{i+1}\right)\right\},
					\end{aligned}
				\end{equation}
				where $\theta\ge0$ is the \textit{Lagrangian multiplier}. Let the \textit{Lagrangian Dual Function} defined as:
				\begin{equation}\label{eq38}
					\Upsilon(\theta,\lambda;f_{\max})\triangleq\inf_{\boldsymbol{\psi}}\mathcal{L}(\boldsymbol{\psi};\theta,\lambda,f_{\max}).
				\end{equation}
				Since an optimal \textit{stationary deterministic} policy exists for the MDP problem $\inf_{\boldsymbol{\psi}}\mathcal{L}(\boldsymbol{\psi};\theta,\lambda,f_{\max})$ (as indicated in Theorem \ref{thm:lifted-unichain}), we can use a short-hand notation $\phi$ to denote $\boldsymbol{\psi}$, and the \textit{Lagrangian Dual Problem} of Problem \ref{p5} is:
				
				\begin{problem}[\textit{Lagrangian Dual Problem}]\label{p6}
					\begin{equation}\label{eq39}
						\begin{aligned}
							d(\lambda;f_{\max})=\max_{\theta\ge0}\Upsilon(\theta,\lambda;f_{\max}),
						\end{aligned}
					\end{equation}
					where $\Upsilon(\theta,\lambda;f_{\max})=\inf_{\phi}\mathcal{L}(\phi;\theta,\lambda,f_{\max})$ .
				\end{problem}
				The \textit{weak duality principle} \cite[Chapter 5.2.2]{boyd2004convex} implies that $d(\lambda;f_{\max})$ is a lower bound of $H(\lambda;f_{\max})$, \emph{i.e.}, $d(\lambda;f_{\max})\le H(\lambda;f_{\max})$. In the following lemma, we establish the conditions where the \textit{strong duality} holds true and thus $d(\lambda;f_{\max})=H(\lambda;f_{\max})$. Under these conditions, it is sufficient to solve $d(\lambda;f_{\max})$ to obtain $H(\lambda;f_{\max})$.
				\begin{Lemma}\label{Lemma6v2}
					(Restatement of \cite[Chapter 5.5.3]{boyd2004convex}) The duality gap between Problem \ref{p5} and Problem \ref{p6} is zero, \emph{i.e.}, $d(\lambda;f_{\max})=H(\lambda;f_{\max})$, if and only if for any given $\lambda$, we can find $\phi^\star_{\lambda+\theta^\star_\lambda}$ and $\theta^\star_\lambda$ such that the Karush--Kuhn--Tucker (KKT) conditions are satisfied:
					\begin{subequations}
						\begin{align}
							&\theta^\star_\lambda\ge0,\label{eq40}\\
							&\phi^\star_{\lambda+\theta^\star_\lambda}=\mathop{\arg\min}_{\phi}\lim _{n \rightarrow \infty} \frac{1}{n} \sum_{i=0}^{n-1}\mathbb{E}_{\phi}\Bigg\{\sum_{t=D_i}^{D_{i+1}-1} \mathcal{C}\left(X_t, A_i\right)-\notag\\&\quad\quad\quad\quad\quad\quad\quad\quad\quad(\lambda+\theta^\star_\lambda) \left(Z_i+Y_{i+1}\right)\Bigg\}+\frac{\theta^\star_\lambda}{f_{\max}},\label{kkt2}\\
							&\lim _{\mathrm{n} \rightarrow \infty}\frac{1}{n} {\sum_{i=0}^{n-1}\mathbb{E}_{\phi^\star_{\lambda+\theta^\star_\lambda}} \left[Z_i+Y_{i+1}\right]}\ge\frac{1}{f_{\max}},\label{KKT3}\\
							&\theta^\star_\lambda\left\{\lim _{\mathrm{n} \rightarrow \infty}\frac{1}{n} {\sum_{i=0}^{n-1}\mathbb{E}_{\phi^\star_{\lambda+\theta^\star_\lambda}} \left[Z_i+Y_{i+1}\right]}-\frac{1}{f_{\max}}\right\}=0.\label{kktT4}
						\end{align}	
					\end{subequations}
				\end{Lemma}
				
				By leveraging Lemma \ref{Lemma6v2}, we reformulate the constrained problem as an unconstrained Problem \ref{p6}. Given a fixed \textit{Dinkelbach} parameter $\lambda$, the goal of Problem \ref{p6} is to determine the \textit{saddle point} $(\phi^\star_{\lambda+\theta^\star_\lambda}, \theta^\star_\lambda)$ of the function $\mathcal{L}(\phi; \theta, \lambda)$. The inner layer of Problem \ref{p6} is a standard MDP problem with fixed \textit{Lagrangian multiplier} $\theta\ge0$ and \textit{Dinkelbach parameter} $\lambda$, while the outer layer seeks the optimal \textit{Lagrangian multiplier} $\theta^\star_{\lambda}$ that maximizes the \textit{Lagrangian Dual Function} \eqref{eq38} under the KKT conditions \eqref{eq40}-\eqref{kktT4}. 
				
				\subsection{Three-Layer Solutions and the Structure of Optimal Policies}
				In this subsection, we propose a \textbf{three-layer} algorithm outlined in Algorithm \ref{Algorithm 2}. The basic framework of this algorithm is inspired by \cite[Section IV.C]{10621420}. This algorithm consists of inner, middle, and outer layers, whose implementations are detailed below.
				\subsubsection{Inner-Layer: A Standard MDP Given $\lambda$ and $\theta$}
				For any given $\theta$ and $\lambda$, the inner layer  $\inf_{\phi}\mathcal{L}(\phi;\theta,\lambda,f_{\max})$ is a standard unconstrained infinite horizon MDP as defined in \ref{MDP}, which is denoted by $\mathscr{P}_{\text{MDP}}(\theta+\lambda)$ with the optimal value $U(\lambda+\theta)$. This problem can be efficiently solved using the $\tau$-RVI algorithm, as outlined in Iteration \ref{convergence:theorem2}. Notably, \cite[Section IV.C]{10621420} applies a value iteration process in the inner layer, while we apply the $\tau$-RVI to ensure the rigorous convergence in our context. 
				
				\subsubsection{Middle-Layer: Update Lagrangian multiplier $\theta$ Given $\lambda$}
				Given a \textit{Dinkelbach parameter} $\lambda$, the middle layer involves solving $\max_{\theta\ge0}\Upsilon(\theta,\lambda;f_{\max})$ in Problem \ref{p6} to accurately approximate $d(\lambda; f_{\max})$, which searches for the optimal $\theta^\star_{\lambda}$ and its corresponding optimal policy $\phi^\star_{\lambda+\theta^\star_{\lambda}}$ that satisfy the KKT conditions in Lemma \ref{Lemma6v2}. Different from \cite[Section IV.C]{10621420} where the sub-gradient method is employed to update the Lagrangian multiplier leveraging the convexity of the Lagrangian Dual function, we conduct a monotonicity analysis and explicitly state the conditions for the optimal multiplier. We begin by introducing the following notations that help clarify the role of the multiplier $\theta$ given a fixed $\lambda$. Specifically, 
				\begin{equation}\label{eq44v2}
					\mathcal{Q}^{\lambda+\theta}\triangleq\limsup _{\mathrm{n} \rightarrow \infty}\frac{1}{n} {\sum_{i=0}^{n-1} \mathbb{E}_{\phi^\star_{\lambda+\theta}}\left[\sum_{t=\it{D}_{i}}^{\it{D}_{i+1}-1} \mathcal{C}(X_t, A_i)\right]},
				\end{equation}
				\begin{equation}\label{eq45v2}
					\mathcal{F}^{\lambda+\theta}\triangleq\liminf _{\mathrm{n} \rightarrow \infty}\frac{1}{n} {\sum_{i=0}^{n-1}\mathbb{E}_{\phi^\star_{\lambda+\theta}} \left[Z_i+Y_{i+1}\right]},
				\end{equation}
				where $\phi^\star_{\lambda+\theta}$ is the optimal \textit{stationary deterministic} policy for the unconstrained problem $\mathscr{P}_{\text{MDP}}(\lambda+\theta)$ given in \ref{MDP}. The subsequent Lemma presents key properties of $\mathcal{F}^{\lambda+\theta}$, $\mathcal{Q}^{\lambda+\theta}$, $U(\lambda+\theta)$, and $\Upsilon(\theta, \lambda; f_{\max})$ with respect to $\theta$, which are useful for the explicit solutions.
				\begin{Lemma} \label{Lemma7v2}(A variant of \cite[Lemma 3.1]{BEUTLER1985236}) The following assertions hold true:
					
					($i$) $U({\lambda+\theta})$ is non-increasing with respect to $\theta$;
					
					($ii$) $\mathcal{Q}^{\lambda+\theta}$ and $\mathcal{F}^{\lambda+\theta}$ are non-decreasing functions with respect to $\theta$; 
					
					($iii$) $\mathcal{Q}^{\lambda+\theta}$ and $\mathcal{F}^{\lambda+\theta}$ are both step functions with respect to $\theta$;
					
					($iv$) If $\mathcal{F}^{\lambda+\theta}\ge1/f_{\max}$, $\Upsilon(\theta,\lambda;f_{\max})$ is non-increasing with respect to $\theta$.
				\end{Lemma}
				\begin{proof}
					Part ($i$) and part ($ii$) are supported by \cite[Lemma 3.1]{BEUTLER1985236}. Part ($iii$) is supported by \cite[Lemma 3.2]{BEUTLER1985236}. See Appendix \ref{appendixkv3} for the detailed proof of part ($iv$).
				\end{proof}
				
				Lemma \ref{Lemma7v2} leads to the following corollary, which provides an explicit solution for the optimal value $\theta_{\lambda}^\star$ in Problem \ref{p6} that satisfies the KKT conditions \eqref{eq40}-\eqref{kktT4}. 
				\begin{figure}
					\centering
					\includegraphics[width=1\linewidth]{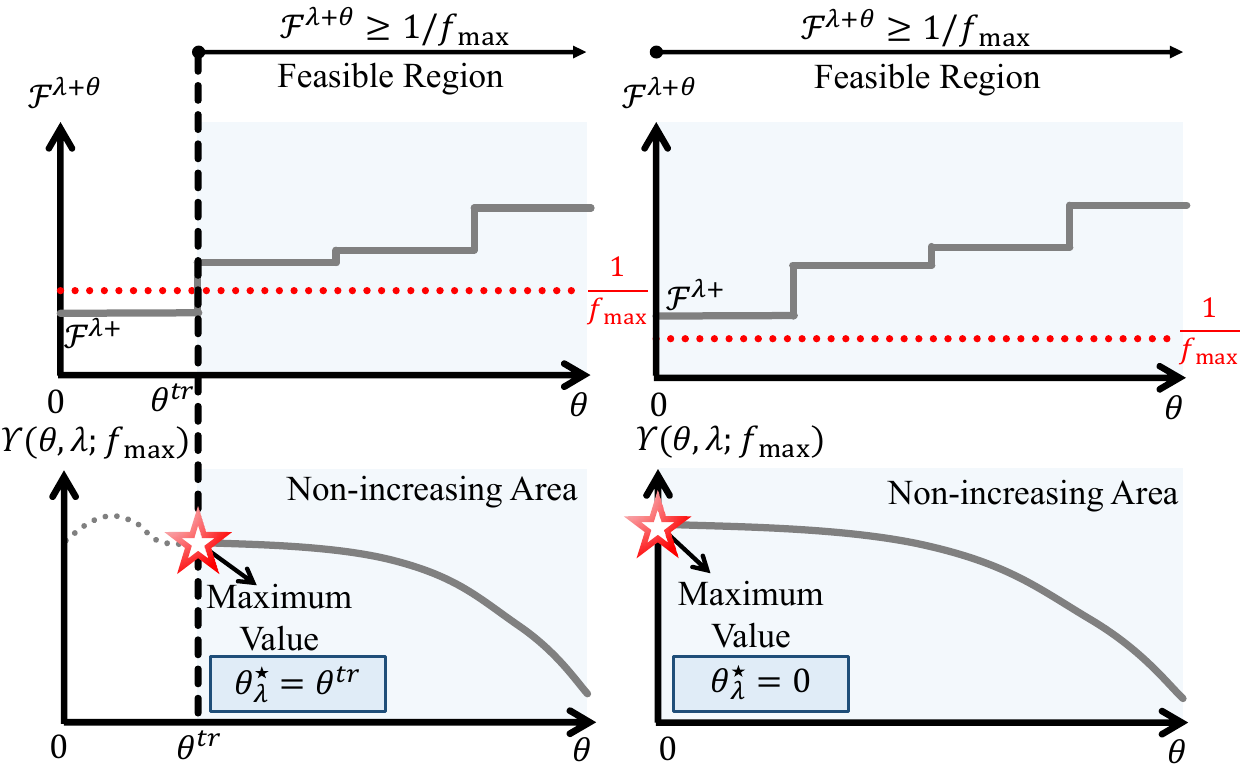}
					\caption{Illustrations of Case ($i$) and case ($ii$) in Corollary \ref{coro1}. In this figure, $\mathcal{F}^{\lambda+\theta}$ is a non-increasing step function with respect to $\theta$, as Lemma \ref{Lemma7v2}-($ii$) and Lemma \ref{Lemma7v2}-($iii$) indicate. In addition,  $\Upsilon(\theta,\lambda;f_{\max})$ is non-increasing with respect to $\theta$ if $\mathcal{F^{\lambda+\theta}}\ge1/f_{\max}$, as Lemma \ref{Lemma7v2}-($iv$) indicates.} 
					\label{fig:lagrangianv2}
				\end{figure}
				\begin{corollary}\label{coro1} Denote $\mathcal{F}^{\lambda^+}$ as the right limit\footnote{Since $\mathcal{F}^{\lambda}$ is a step function, it does not necessarily follow that $\mathcal{F}^{\lambda^+} = \mathcal{F}^{\lambda}$. Specifically, when $\lambda$ is a break point, we have $\mathcal{F}^{\lambda^+}>\mathcal{F}^{\lambda}$.} of $\mathcal{F}^\lambda$: 
					\begin{align}\label{eq46v2}
						\mathcal{F}^{\lambda^+}=\lim_{\Delta\lambda \rightarrow 0}\mathcal{F}^{\lambda+\Delta\lambda},
					\end{align}
					The following assertions hold true:
					
					$(i)$. If $\mathcal{F}^{\lambda^+} \geq 1/f_{\max}$, then $\theta_{\lambda}^\star=0$;
					
					$(ii)$. If $\mathcal{F}^{\lambda^+} < 1/f_{\max}$, then $\theta_{\lambda}^\star$ is equal to a positive break point $\theta^{\text{tr}}>0$, which satisfies:
					\begin{equation}\label{eq46}
						\begin{aligned}
							\mathcal{F}^{(\theta^{\text{tr}}+\lambda)^-}<\frac{1}{f_{\max}}\le\mathcal{F}^{(\theta^{\text{tr}}+\lambda)^+},\\			
						\end{aligned}
					\end{equation}
				\end{corollary}
				\begin{proofsketch}
					If $\mathcal{F}^{\lambda+}\ge\frac{1}{f_{\max}}$, which is illustrated in the right panel of Fig. \ref{fig:lagrangianv2}, the function $\Upsilon(\theta,\lambda;f_{\max})$ is non-increasing with $\theta$ for $\theta\ge0$. In this case, the maximum value of $\Upsilon(\theta,\lambda;f_{\max})$ is obtained at $\theta_{\lambda}^{\star}=0$. If $\mathcal{F}^{\lambda+}<\frac{1}{f_{\max}}$, the feasible region under the KKT condition \eqref{KKT3} is shown in the left panel of Fig. \ref{fig:lagrangianv2} with $\theta\ge\theta^{\text{tr}}$. In this feasible region, the function $\Upsilon(\theta,\lambda;f_{\max})$ is non-increasing in $\theta$. Therefore, the maximum value of $\Upsilon(\theta,\lambda;f_{\max})$ is obtained at $\theta_{\lambda}^{\star}=\theta^{\text{tr}}$. A detailed proof is provided in Appendix \ref{appendixlv2}.
				\end{proofsketch}
				
				Having established Corollary \ref{coro1}, our remaining task is to identify the threshold value $\theta^{\text{tr}}$ that satisfies \eqref{eq46} under the condition that $\mathcal{F}^{\lambda^+} < 1/f_{\max}$. \textcolor{black}{We here introduce two algorithms for searching this threshold:}
			\textcolor{black}{	\begin{itemize}
					\item \textbf{Bisection Search}: Given that $\mathcal{F}^{\lambda + \theta}$ is non-decreasing in $\theta$, we can apply a \textit{bisection search} to gradually converge on the threshold $\theta^{\text{tr}}$. This method is a classical approach for locating thresholds in monotonic functions.
					\item \textbf{Intersection Search} \cite[Algorithm 2]{11007613}: To improve the efficiency of solving constrained MDPs, \cite{11007613} further exploits the piece-wise linear and concave (PWLC) structure Lagrangian cost function under finite state and action spaces. Unlike traditional bisection methods that locate the zero-crossing of a single monotonic function, \textit{intersection search} identifies the intersection point of two tangents to the Lagrangian curve, thereby accelerating the search for the optimal Lagrange multiplier $\theta^{\text{tr}}$ with fewer iterations. We refer readers to \cite[Section V.A]{11007613} for a detailed discussion of this algorithm.
				\end{itemize}}	
				 Once this threshold value is obtained, setting $\theta_\lambda^\star = \theta^{\text{tr}}$ ensures that the KKT conditions \eqref{eq40}--\eqref{KKT3} are satisfied. What remains is to determine the optimal policy $\phi_{\lambda + \theta^\star_{\lambda}}$ that guarantees the KKT condition \eqref{kktT4}. The structure of this optimal policy is provided in the following theorem. 
				\begin{theorem}\label{theorem5}(Structure of the optimal policy) The following assertions hold true:
					
					($i$). If $\mathcal{F}^{\lambda^+}\ge1/f_{\max}$, the optimal policy is a stationary deterministic policy $\phi^\star_{\lambda}$ determined in \eqref{eq19}; 
					
					($ii$). If $\mathcal{F}^{\lambda^+}<1/f_{\max}$ and $\mathcal{F}^{(\lambda+\theta^\star_\lambda)^+}=1/f_{\max}$, the optimal policy is a stationary deterministic policy $\phi^\star_{(\lambda+\theta^\star_\lambda)}$;
					
					($iii$). If $\mathcal{F}^{\lambda^+}<1/f_{\max}$ and $\mathcal{F}^{(\lambda+\theta^\star_\lambda)^+}>1/f_{\max}$, the optimal policy is a random mixture of two stationary deterministic policy $\phi^\star_{(\lambda+\theta^\star_\lambda)^+}$ and $\phi^\star_{(\lambda+\theta^\star_\lambda)^-}$:
					\begin{equation}\label{eq52}
						\phi^\star_{\lambda+\theta^\star_\lambda}(\gamma)=\begin{cases}
							\phi^\star_{(\lambda+\theta^\star_\lambda)^+}(\gamma),  &\text{w.p. } \eta \\
							\phi^\star_{(\lambda+\theta^\star_\lambda)^-}(\gamma), &\text{w.p. }1-\eta
						\end{cases}, \text{ for } \forall \gamma\in\mathcal{S}\times\mathcal{Y}\times\mathcal{A},
					\end{equation}
					where $\eta$ is a randomization factor given by
					\begin{equation}\label{eq53v2}
						\eta=\frac{\frac{1}{f_{\max}}-\mathcal{F}^{\lambda+\theta^\star_\lambda}}{\mathcal{F}^{(\lambda+\theta^\star_\lambda)^+}-\mathcal{F}^{\lambda+\theta^\star_\lambda}}.
					\end{equation}
				\end{theorem}
				\begin{proof}
					See Appendix \ref{appendixkM}.
				\end{proof}
				
				With Theorem \ref{theorem5}, we can determine the optimal policy once the optimal $\theta_{\lambda}^{\star}$ is obtained. In this way, the optimal value $d(\lambda;f_{\max})$ is determined as:
				\begin{equation}
					\begin{aligned}
						d(\lambda;f_{\max})&=
						\Upsilon(\theta_{\lambda}^{\star},\lambda;f_{\max})\\&=
						\begin{cases}
							U(\lambda)&\text{if } \mathcal{F}^{\lambda^+}\ge\frac{1}{f_{\max}}\\
							U(\lambda+\theta_{\lambda}^{\star})+\frac{\theta_{\lambda}^{\star}}{f_{\max}}&\text{if } \mathcal{F}^{\lambda^+}<\frac{1}{f_{\max}}
						\end{cases}.
					\end{aligned}
				\end{equation}
				The following subsection aims at searching the root of $d(\lambda;f_{\max})$.
				\begin{algorithm}[t]
					\caption{\textit{A Three-layer Algorithm (\textbf{A Variant of \cite[Section IV.C]{10621420}})}}
					\label{Algorithm 2}
					\LinesNumbered
					\KwIn{Tolerence $\epsilon_1,\epsilon_2>0$, MDP $\mathscr{P}_{\mathrm{MDP}}(\lambda)$, maximum sampling frequency $f_{\max}$}
					Initialization: sufficiently large $\lambda_{\uparrow}$, and $\lambda_{\downarrow}=\min_{s,a}\mathcal{C}(s,a)$\;\tcp{Outer-Layer Bisection Search on Dinkelbach parameter $\lambda$}
					\While {$\lambda_{\uparrow}-\lambda_{\downarrow}\ge \epsilon_1$}
					{$\lambda=(\lambda_{\uparrow}+\lambda_{\downarrow})/2$\;
						Run $\tau$-RVI to solve $\mathscr{P}_{\mathrm{MDP}}(\lambda)$ and calculate $\mathcal{F}^{\lambda^+}$ and $U(\lambda)$\;
						\If{$\mathcal{F}^{\lambda^+}\ge1/f_{\max}$}{$d(\lambda;f_{\max})=U(\lambda)$\;}
						\Else{
							Initialization: sufficiently large $\theta_{\uparrow}$, and $\theta_{\downarrow}=0$\;\tcp{Middle-Layer Bisection Search on Lagrangian multiplier $\theta$}
							\While {$\theta_{\uparrow}-\theta_{\downarrow}\ge \epsilon_2$}{
								$\theta=(\theta_{\uparrow}+\theta_{\downarrow})/2$\;\tcp{Inner-Layer MDP Given $\lambda$ and $\theta$}Run $\tau$-RVI to solve $\mathscr{P}_{\text{MDP}}(\lambda+\theta)$ and calculate $\mathcal{F}^{(\lambda+\theta)^+}$ and $U(\lambda+\theta)$\;
								\If{$\mathcal{F}^{(\lambda+\theta)^+}\ge1/f_{\max}$}
								{
									$\theta_{\uparrow}=\theta$\;
								}
								\Else
								{$\theta_{\downarrow}=\theta$\;}}
							$d(\lambda;f_{\max})=U(\lambda+\theta)+\frac{\theta}{f_{\max}}$\;
						}
						\If{$d(\lambda;f_{\max})>0$}
						{
							$\lambda_{\uparrow}=\lambda$\;
						}
						\Else
						{$\lambda_{\downarrow}=\lambda$\;}			
					}
					\KwOut{$h^\star=\lambda$}
				\end{algorithm}
				
				\subsubsection{Outer-Layer: Update \textit{Dinkelbach parameter} $\lambda$}
				With the approximation of $d(\lambda;f_{\max})$ in hand, the outer layer updates $\lambda$ in a \textit{bisection-search} fashion by leveraging Lemma \ref{l5v2} to finally approach the root $h^\star$ such that $d(h^\star;f_{\max})=0$. The flow of the algorithm is demonstrated in Algorithm \ref{Algorithm 2}. 
				
				\section{Optimal Sampling With Rate Constraint: {One}-{L}ayer {I}teration is All You Need}\label{Onelayerwithrate}
				\begin{algorithm}[t]
					\caption{\textit{\textbf{Proposed} One-layer \textsc{QuickBLP}}}
					\label{Algorithm 3}
					\LinesNumbered
					\KwIn{MDP $\mathscr{P}_{\mathrm{MDP}}(\lambda)$ and maximum sampling frequency $f_{\max}$}
					Run \textsc{OnePDSI} in Iteration \ref{propsition2} to obtain $\rho^\star$\;
					Calculate the left limit $\mathcal{F}^{(\rho^\star)^-}$\;
					\If{$\mathcal{F}^{(\rho^\star)^-}\ge1/f_{\max}$}{$h^\star=\rho^\star$\tcp*{Case ($i$) of Theorem \ref{the6}}}
					\Else{Solve LP Problem \ref{p7} to obtain $Q^\star(f_{\max})$\; 
						Calculate the root $h^\star=f_{\max}\cdot Q^\star(f_{\max})$\tcp*{Case ($ii$) of Theorem \ref{the6}}}
					\KwOut{$h^\star$}
				\end{algorithm}
				In the three-layer algorithm, each update of $\lambda$ or $\theta$ necessitates solving the inner-layer MDP using the $\tau$-RVI algorithm. This process incurs a high computational complexity due to the repeated execution of the $\tau$-RVI algorithm required to iteratively search and optimize the parameters $\lambda$ and $\theta$.
				
				\subsection{QuickBLP}
				
				In this section, we design a one-layer \textit{two-stage} hierarchical algorithm namely, {\normalfont \textsc{QuickBLP}}, which reduces the computational complexity by explicitly solving for the root $h^\star$. Rather than relying on iterative \textit{bisection search} to find the root $h^\star$, our approach leverages a direct structural exploration of the solution space, allowing us to bypass the need for multiple $\tau$-RVI executions. Consequently, the {\normalfont \textsc{QuickBLP}} solves this problem in two stages. The first stage solves a Bellman variant and the second stage explicitly expresses the root $h^\star$ as a function of the solution to an LP problem, thus directly obtaining the root $h^\star$ in a more computationally efficient manner. The main structural results are summarized in the following theorem. 
				\begin{theorem}\label{the6}
					(Structural Results of the root $h^\star$) The following assertions hold true:
					
					($i$). If the root of $U(\rho)$, denoted as $\rho^{\star}$, satisfies  
					$\mathcal{F}^{(\rho^{\star})^-}\ge 1/f_{\max}$, then the optimal value of Problem \ref{p1} is $h^\star=\rho^{\star}$, and the optimal policy for Problem \ref{p1} is $\phi_{\rho^{\star}}^\star$, as defined in \eqref{eq19};
					
					($ii$). If  $\mathcal{F}^{(\rho^{\star})^-}< 1/f_{\max}$, the optimal value of Problem \ref{p1} is
					\begin{equation}
						h^\star=f_{\max}\cdot Q^\star(f_{\max}),
					\end{equation}where $Q^\star(f_{\max})$ is the optimal value of the following Linear Programming: 
					
					\begin{problem}[\textit{Linear Programming Reformulation}]\label{p7}
						\begin{subequations}
							\begin{align}
								&Q^\star(f_{\max})=\min_{\mathbf{x}} \sum_{\gamma,z,a}q(\gamma,z,a)x(\gamma,z,a) \\
								\mathrm{s.t.}\quad& \sum_{\gamma,z,a} f(z)x(\gamma,z,a)=\frac{1}{f_{\max}},\label{eq53}\\
								& \sum_{z,a}x(\gamma',z,a)=\sum_{\gamma,z,a}p(\gamma'|\gamma,z,a)x(\gamma,z,a),\notag\\&\quad\quad\quad\quad\quad\quad\quad\quad\quad\quad\forall \gamma'\in\mathcal{S}\times\mathcal{Y}\times\mathcal{A},\\
								&\sum_{\gamma,z,a}x(\gamma,z,a)=1, \\
								& x(\gamma,z,a)\ge0, ~~\forall \gamma\in\mathcal{S}\times\mathcal{Y}\times\mathcal{A},z\in\mathcal{Z},a\in\mathcal{A},\label{eq56} 
							\end{align}
						\end{subequations}
						and the corresponding optimal policy is a randomized policy given by:
						\begin{equation}
							\begin{aligned}
								\phi^\star(\gamma)=(z,a), \text{w.p. } \frac{x(\gamma,z,a)}{\sum_{\zeta,\alpha}x(\gamma,\zeta,\alpha)}, \\\forall \gamma\in\mathcal{S}\times\mathcal{Y}\times\mathcal{A},z\in\mathcal{Z},a\in\mathcal{A}.
							\end{aligned}
						\end{equation}
					\end{problem}
				\end{theorem}
				\begin{proof}
					{See Appendix \ref{appendixkv2}}.
				\end{proof}
				
				Theorem \ref{the6} leads to a one-layer two-stage algorithm presented in Algorithm \ref{Algorithm 3}. In the first stage, the algorithm solves the unconstrained MDP $\mathscr{P}_{\text{MDP}}(\lambda)$ and determines the root of $U(\lambda)$, denoted by $\rho^\star$. This root is obtained by implementing \textsc{OnePDSI} in Iteration \ref{propsition2}. If the condition specified in part ($i$) of Theorem \ref{the6} holds true, then the root $h^\star$ is immediately found as $h^\star = \rho^\star$. If the condition of part ($i$) of Theorem \ref{the6} is not met, the algorithm proceeds to the second stage. Here, the LP problem\footnote{Once an LP problem is established, it can be solved to global optimality using well-established algorithms such as the \textit{simplex} and modern \textit{primal--dual interior-point} methods. \textcolor{black}{For problems whose variable set is large, mature decomposition techniques such as \textit{column generation} allow solving the LP without ever forming the entire matrix, while preserving exact optimality.}} Problem \ref{p7} is solved. The solution of this LP provides the optimal value $Q^\star(f_{\max})$. Finally, the root $h^\star$ is explicitly determined as $h^\star=f_{\max}\cdot Q^\star(f_{\max})$. Compared to the previous Algorithm \ref{Algorithm 2}, this newly proposed algorithm eliminates the need for multiple executions of the $\tau$-RVI for searching $h^\star$ and $\theta^\star_{h^\star}$.
				
				\subsection{\textcolor{black}{Polynomial Complexity Results}}
				\textcolor{black}{In this subsection, we analyze that both the methods proposed in this paper are \textbf{polynomial} in time and space complexity. For the three-layer Algorithm \ref{Algorithm 2}, the computational cost of a single $\tau$-RVI iteration is $\mathcal{O}(|\mathcal{Z}||\mathcal{S}|^2|\mathcal{Y}|^2|\mathcal{A}|^3)$, and according to Theorem~\ref{theorem3:convergence}, the $\tau$-RVI requires at most $\mathcal{O}\left(\log(1/\delta)\right)$ iterations to reach accuracy $\delta$. Consequently, the time complexity of the three-layer Algorithm~\ref{Algorithm 2} is
				\begin{equation}
					\mathcal{O}\left(\log(\frac{1}{\epsilon_1})\log(\frac{1}{\epsilon_2})\log(\frac{1}{\delta})|\mathcal{Z}||\mathcal{S}|^2|\mathcal{Y}|^2|\mathcal{A}|^3\right),
				\end{equation}
				which is \textbf{polynomial} in all problem parameters. In terms of space complexity, the algorithm stores one value vector $\tilde{V}_{K}(\gamma;\lambda)$ of size $\mathcal{O}(|\mathcal{S}\times\mathcal{Y}\times\mathcal{A}|)$ together with the sparse transition structure of size $\mathcal{O}(|\mathcal{Z}||\mathcal{S}|^2|\mathcal{Y}|^2|\mathcal{A}|^3)$. However, the overall computational cost can be dominated by the outer layers: the three-layer algorithm performs a large number of nested loops to reach accuracies $\epsilon_1, \epsilon_2$ and $\delta$. When these tolerances are set small, the number of outer loops can be very high, and the resulting run time, though still polynomial, might become impractical.}
				
				\textcolor{black}{In contrast, the proposed \textsc{QuickBLP} addresses this challenge by eliminating the outer-loop search. In the first stage, \textsc{OnePDSI} incurs $\mathcal{O}(|\mathcal{Z}||\mathcal{S}|^2|\mathcal{Y}|^2|\mathcal{A}|^3)$ per iteration and converges in $\mathcal{O}\!\big(\log(1/\delta)\big)$
				iterations (Theorem~\ref{the6upperbound2}), \emph{i.e.}, the complexity of \textsc{OnePDSI} is polynomial:
				\begin{equation}
					\mathcal{O}\left(\log(\frac{1}{\delta})|\mathcal{Z}||\mathcal{S}|^2|\mathcal{Y}|^2|\mathcal{A}|^3\right).
				\end{equation}
				If the stopping condition of Theorem~\ref{the6} is satisfied, the algorithm terminates, and no LP needs to be solved. 
				Only when this condition is not met does the second stage become active, in which the LP Problem \ref{p7} with $n=|\mathcal{Z}|\times|\mathcal{S}|\times|\mathcal{Y}|\times|\mathcal{A}|^2$ variables and $m=|\mathcal{S}\times\mathcal{Y}\times\mathcal{A}|+2$ equality constraints must be solved. The LP is handled via the Column Generation (CG) method \cite{desaulniers2006column} to decompose the Master Problem (MP) into some manageable Restricted Master Problems (RMP). Each RMP is solved using a primal--dual interior-point method \cite[Chap. 11]{boyd2004convex}. Since the occupancy-measure LP admits optimal basic feasible solutions, the solution vector $\mathbf{x}$ can be chosen with at most $m$ nonzeros; hence, at CG iteration $t$ the RMP involves $n_t
				\le m$ active columns. A primal--dual interior-point method requires $\mathcal{O}(\sqrt{m}\log(1/\epsilon))$ iterations \cite[Theo. 3.1]{gondzio2012interior}, each dominated by solving a sparse Schur system with $n_t$ variables and $m$ constraints, whose factorization cost $T_{\text{fact}}(m,n_t)$ ranges from $\mathcal{O}(m^{1.5})$ to $\mathcal{O}(m^{2})$. As a result, the total cost for solving the LP via CG and primal--dual interior-point methods can be upper bounded by:
				\begin{equation}
					\mathcal{O}\left(\sum_{t=1}^q\sqrt{m}\log(1/\epsilon)\cdot T_{\text{fact}}(m,n_t)\right),
				\end{equation}
				where $q\le m$ is the number of CG iterations and $\sum_{t=1}^q\le m^2$. This yields a worst-case complexity between $\mathcal{O}(m^{3.5})$ to $\mathcal{O}(m^4)$. The space complexity to solve the LP is dominated by storing and factorizing the Schur complement system of order $m$. Since each RMP contains $n_t\le m$ active columns, the overall memory requirement is $\mathcal{O}(m)$ to $\mathcal{O}(m^2)$, depending on the sparsity structure of the constraints.}
				
				\textcolor{black}{\textbf{Discussion}: Both the algorithms achieve \textbf{polynomial} complexity. The three-layer algorithm admits a lower-degree polynomial bound in theory but suffers from heavy nested searches, which can make the runtime large in practice. 
				In contrast, \textsc{QuickBLP} adopts a two-stage design: in many instances the first stage alone suffices, and the second-stage LP is only occasionally invoked. 
				Although the LP stage has a higher \textit{worst-case} bound, the combination of CG and sparse interior-point methods makes the practical complexity often much lower than the theoretical worst case. As a result, \textsc{QuickBLP} is typically far more efficient in real problem instances.}
				
				\subsection{\textcolor{black}{Approximate LP: Scaling to Large Spaces}}
				\textcolor{black}{In large-scale systems, the growth of the state space $\mathcal{S}$, the action space $\mathcal{A}$, and the delay space $\mathcal{Y}$ renders optimally solving the problem intractable. A key advantage of \textsc{QuickBLP} is that its linear programming formulation naturally accommodates approximate linear programming (ALP), making it well suited for scaling to large state--action--delay spaces. Following \cite{malek2014linear}, we first adopt a linear approximation of the occupation measure \(x(\gamma, z, a)\) to reduce the number of variables:
				\begin{equation}
					\begin{aligned}
						x(\gamma,z,a)&\approx\hat{x}(\gamma,z,a;\boldsymbol{\theta})\\
						&\triangleq\mu_0(\gamma,z,a)+\sum_{i=1}^{d}\theta_i\psi_i(\gamma,z,a)\\
						&=\mu_0(\gamma,z,a)+\boldsymbol{\theta}\boldsymbol{\psi}^{T},
					\end{aligned}
				\end{equation}
				where \(\mu_0(\gamma, z, a)\) is a fixed baseline function, \(\{\psi_i(\gamma, z, a)\}_{i=1}^d\) are predefined basis functions, and \(\boldsymbol{\theta} \in \mathbb{R}^d\) is the parameter vector to be optimized. This reduces the number of variables from $n=|\mathcal{Z}|\times|\mathcal{S}|\times|\mathcal{Y}|\times|\mathcal{A}|^2$ to feature dimension $d$. Substituting this approximation into the LP Problem~\ref{p7} yields the following ALP:
									\begin{problem}[\textit{Linear Programming Approximation}]\label{p8-1}
					\begin{subequations}
						\begin{align}
							&Q^\star(f_{\max})\approx\min_{\boldsymbol{\theta}} \sum_{\gamma,z,a}q(\gamma,z,a)(\mu_0(\gamma,z,a)+\boldsymbol{\theta}\boldsymbol{\psi}^{T}) \\
							\mathrm{s.t.}\quad& \sum_{\gamma,z,a} f(z)\hat{x}(\gamma,z,a;\boldsymbol{\theta})=\frac{1}{f_{\max}},\\
							& \sum_{z,a}\hat{x}(\gamma',z,a;\boldsymbol{\theta})=\sum_{\gamma,z,a}p(\gamma'|\gamma,z,a)\hat{x}(\gamma,z,a;\boldsymbol{\theta}),\notag\\&\quad\quad\quad\quad\quad\quad\quad\quad\quad\quad\forall \gamma'\in\mathcal{S}\times\mathcal{Y}\times\mathcal{A},\\
							&\sum_{\gamma,z,a}\hat{x}(\gamma,z,a;\boldsymbol{\theta})=1, \\
							& \hat{x}(\gamma,z,a;\boldsymbol{\theta})\ge0, ~~\forall \gamma\in\mathcal{S}\times\mathcal{Y}\times\mathcal{A},z\in\mathcal{Z},a\in\mathcal{A},
						\end{align}
					\end{subequations}
				\end{problem}
				The above ALP reduces the number of decision variables from \(n = |\mathcal{Z}| \times |\mathcal{S}| \times |\mathcal{Y}| \times |\mathcal{A}|^2\) to a fixed feature dimension \(d\), but the number of equality constraints still scales with the sizes of $\mathcal{S}$, $\mathcal{Y}$, and $\mathcal{A}$. This residual dependence on the system size limits the scalability of the approximation unless the constraints are further relaxed.}
				
				\textcolor{black}{To achieve scalability, the hard constraints can be relaxed into soft penalties in the objective function~\cite{malek2014linear}. Specifically, the objective function is augmented with weighted $\ell_1$ violations of the constraints, yielding a surrogate loss. 
			This reformulation converts the constrained LP into an \textit{unconstrained stochastic convex optimization problem} over the low-dimensional parameter $\boldsymbol{\theta}\in\mathbb{R}^d$, which can be solved efficiently using \emph{stochastic subgradient descent}. 
			As a result, the overall complexity scales only with the feature dimension $d$, becoming independent of $\mathcal{S}$, $\mathcal{Y}$, and $\mathcal{A}$, and thus achieves true scalability.}
				
				\subsection{Sensitivity Analysis and Sampling Frequency Threshold}
				
				To better understand the relationship between $h^{\star}$ and the maximum rate constraint parameter $f_{\max}$, we conduct a \textit{sensitivity analysis} in this subsection to explore how $h^{\star}$ depends on the maximum rate constraint with parameter $f_{\max}$. The key result is stated in the following Lemma, whose proof can be found in Appendix \ref{apposa}.
				
				\begin{theorem}\label{Lemmasa}
					(Sensitivity Analysis). Define $f_{\max}^T$ as:
					$
						f_{\max}^T\triangleq\frac{1}{\mathcal{F}^{(\rho^{\star})^-}},
					$
					the following assertions hold true:
					
					($i$). If $f_{\max}\ge f_{\max}^T$, $h^{\star}$ is independent of $f_{\max}$; 
					
					($ii$). If $f_{\max}< f_{\max}^T$,  $h^{\star}=f_{\max}Q^{\star}(f_{\max})$ is monotonically non-increasing with $f_{\max}$, with the derivative given as:
					$
						\frac{\mathrm{d}h^{\star}}{\mathrm{d}f_{\max}}=U(\lambda^{\star}),
					$
					where 
					$
						\lambda^{\star} = \underset{\lambda}{\arg\max} \left\{ \lambda + f_{\max}U(\lambda) \right\}.
					$
				\end{theorem}
				
				Furthermore, the condition that distinguishes the two cases in Lemma \ref{Lemmasa} directly establishes the following corollary, which characterizes the phenomenon where additional sampling does not contribute to further decision-making performance.
				\begin{corollary} (Sampling frequency threshold)
					When the sampling frequency $f_{\max}$ exceeds or equals the threshold $f_{\max}^T$, \emph{i.e.}, $f_{\max}\ge f_{\max}^T$, any further increase in sampling frequency provides no improvement in decision performance. 
				\end{corollary}
				
				%
				\section{Simulation Results}\label{sectionVIII}
				\subsection{Comparing Benchmark}
				\subsubsection{Sampling Policies}
				In this paper, the sampling policy that is co-designed with the goal-oriented remote decision-making policy to minimize the \textit{long-term average cost} in Problem \ref{p1} is referred to as ``\textit{Goal-oriented sampling}''. We compare this sampling policy with the following benchmarks:
				\begin{itemize}
					\item \textit{Uniform sampling}: The sampling is activated periodically with a constant sampling interval $d\in\mathbb{N}^+$. In this case, the sampling process follows $S_{i+1}-S_i=d$ for $\forall i$ (see \cite[Section VI]{DBLP:conf/infocom/KaulYG12} for a detailed system description of uniform sampling with queuing and random service delay). Since the sampling interval $d$ is limited to integer values in the discrete-time MDP setting, the corresponding sampling frequency $1/d$ also takes discrete values. Hence, the simulation yields a performance curve comprising discrete points, as shown in Fig.~\ref{fig:compareconstraint}.
					
					\item \textit{Zero-wait sampling}: An update is transmitted once the previous update is delivered, \emph{i.e.}, $Z_i=0$ for $\forall i$. This policy achieves the maximum throughput. The sampling frequency for \textit{zero-wait} sampling is feasible only if $ f_{\operatorname{max}} \geq 1/\mathbb{E}[Y_i] $. 
					
					\item \textit{Constant-wait sampling}: Waiting before transmitting is reported to be a good alternative for information \textit{freshness} \cite{sun2017update}. Here we consider $Z_i=z$, where $z\in\mathbb{N}^+$ is a constant waiting time. The sampling frequency for \textit{constant-wait sampling} is $1/(z+\mathbb{E}[Y_i])$, which imposes a feasibility condition: $f_{\max}\ge1/(z+\mathbb{E}[Y_i])$.
					
					\item \textit{AoI-optimal sampling}:
					The \textit{AoI-optimal} policy determines $Z_i$ by \cite[Theorem 4]{DBLP:journals/tit/SunUYKS17}, which is a threshold-based policy $Z_i=\max(0,\beta-Y_i)$, where $\beta$ is the solution to the following equations:
					\begin{equation}
						\begin{aligned}
							\mathbb{E}\left[Y+z(Y)\right]=\operatorname*{max}\left(\frac{1}{f_{\operatorname*{max}}},\frac{\mathbb{E}[(Y+z(Y))^{2}]}{2\beta}\right)
						\end{aligned}
					\end{equation}
					with 
					$
						z(Y)=\max(0,\beta-Y).
					$
					\cite[Algorithm 2]{DBLP:journals/tit/SunUYKS17} presents a low-complexity algorithm to solve $\beta$.
				\end{itemize}
				
				\subsubsection{Remote Decision Making}
				We consider the following decision-making policies:
				\begin{itemize}
					\item \textit{Co-design with goal-oriented sampling}: This decision-making policy is obtained by solving Problem \ref{p1}. For conciseness, we refer to the co-design approach simply as ``\textit{goal-oriented sampling}''.
					\item \textit{Myopic decision policy} $\pi_{\text{myopic}}$: At each time step, the remote decision maker selects an action based on the most recent (possibly outdated) observation $ X_{t - \Delta(t)}$. The selected action aims to minimize the instantaneous cost and is formally given by \begin{equation} a_t = \pi_{\text{myopic}}(X_{t - \Delta(t)}) = \arg\min_a \mathcal{C}(X_{t - \Delta(t)}, a), \end{equation} where $\mathcal{C}(\cdot, \cdot)$ denotes the one-step cost function. This policy is termed \textit{myopic} because it excludes any prediction or planning mechanisms; instead, decisions are made greedily to minimize the one-step cost based on the most recently available observation.
					\item \textit{Long-term optimal decision policy $\pi^\star$}: The long-term optimal policy thinks \textit{ahead} to minimize the long-term average cost over time. We compute it by solving the ACOE of the primal MDP using the RVI algorithm. The resulting policy employed at the remote decision maker $\pi^\star:\mathcal{S}\rightarrow \mathcal{A}$ selects an action based on the latest available observation 
					$X_{t-\Delta(t)}$, according to $a_t=\pi^\star(X_{t-\Delta(t)})$. \textcolor{black}{\textit{Unless otherwise stated, this long-term optimal decision policy is employed as the default downstream decision-making policy under all benchmark sampling strategies. Throughout the paper, each baseline (e.g., "zero-wait", "AoI-optimal", "constant-wait") refers to a \textit{composite strategy} that combines the corresponding sampling policy with this long-term optimal decision policy.}} 
				\end{itemize}
				\begin{figure}[t]
					\centering
					\includegraphics[width=0.9\linewidth]{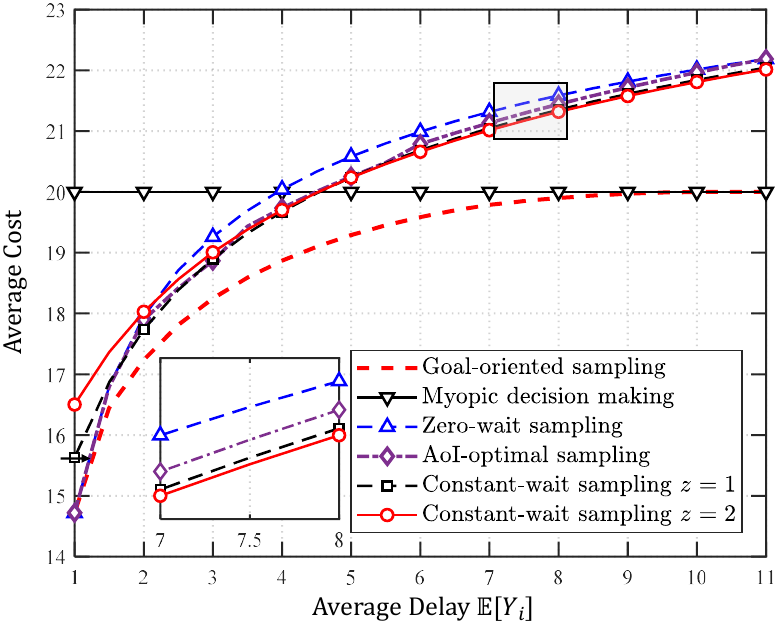}
					\caption{Average cost vs. Average delay $\mathbb{E}[Y_i]$ under \textcolor{black}{binary delay model ($Y_{\max} = 11$), with parameter $p$ controlling the average delay. All baseline sampling policies are paired with the corresponding \textit{long-term optimal decision policy}.}}
					\label{fig:comparewithoutconstraint}
				\end{figure}
				\begin{figure}[t]
					\centering
					\includegraphics[width=0.9\linewidth]{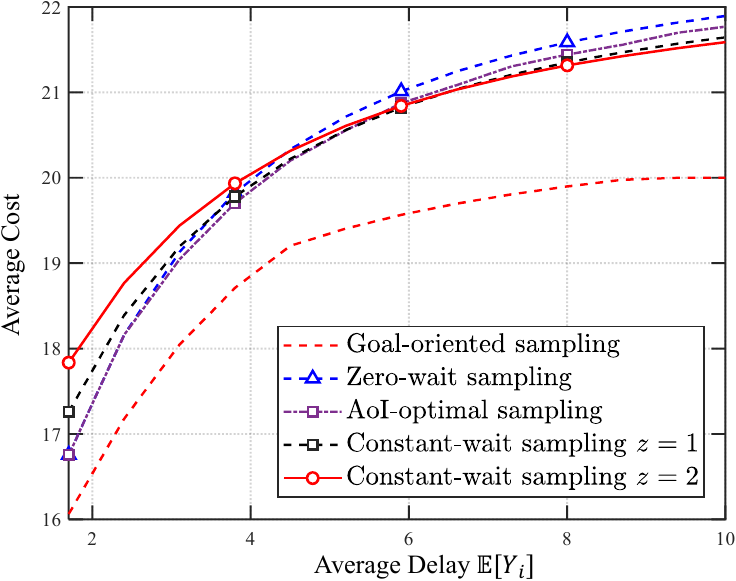}
					\caption{\textcolor{black}{Average cost vs. Average delay $\mathbb{E}[Y_i]$ under binary delay model ($p = 0.3$), with $Y_{\max}$ controlling the average delay. All baseline sampling policies are paired with the corresponding \textit{long-term optimal decision policy}.}}
					\label{fig:comparewithoutconstraint2}
				\end{figure}
					\begin{figure}[t]
					\centering
					\includegraphics[width=0.9\linewidth]{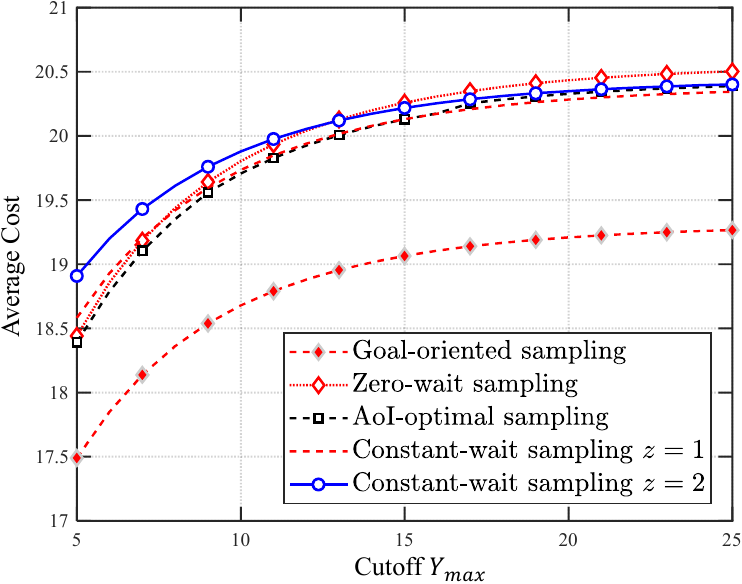}
					\caption{\textcolor{black}{Average cost vs. Cutoff $Y_{\max}$ under truncated geometric delay model with $q = 0.3$. All baseline sampling policies are paired with the corresponding \textit{long-term optimal decision policy}.}}
					\label{fig:geometric}
				\end{figure}
				\begin{figure}[t]
					\centering
					\includegraphics[width=0.9\linewidth]{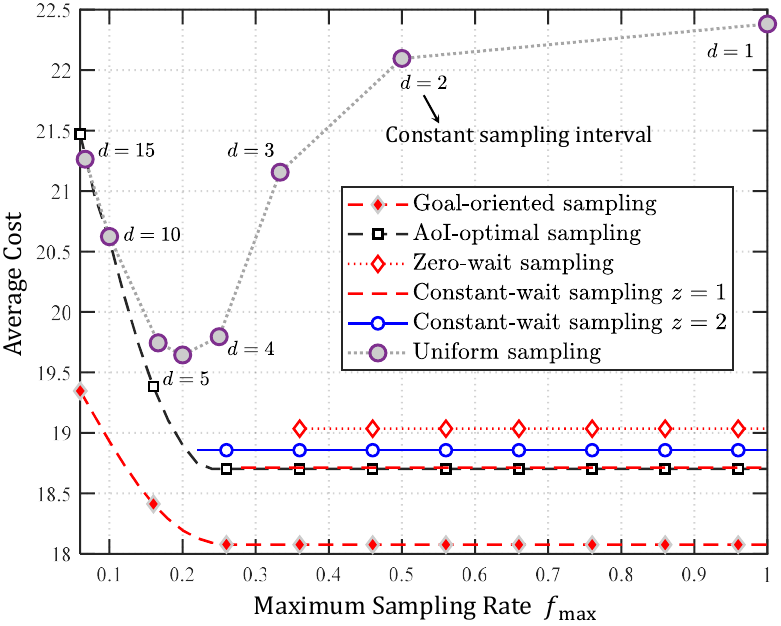}
					\caption{Average cost vs. Maximum sampling frequency $f_{\max}$. \textcolor{black}{All baseline sampling policies are paired with the corresponding \textit{long-term optimal decision policy}.}}
					\label{fig:compareconstraint}
				\end{figure}

				\subsection{Simulation Parameter Setup}
				\subsubsection{Primal MDP}
				
				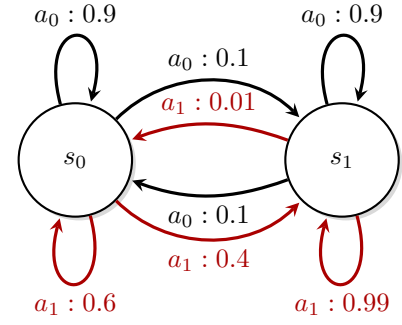
\begin{figure}[h]
					\centering
					\begin{tikzpicture}[
						> = stealth,
						shorten > = 1pt,
						auto,
						node distance = 2cm,
						thick
						]
						
						\definecolor{myblue}{RGB}{0, 0, 0}  
						\definecolor{myred}{RGB}{170, 0, 0} 
						
						\tikzstyle{state} = [circle, draw=black, thick, fill=white, drop shadow={shadow xshift=0.5mm, shadow yshift=-0.5mm, opacity=0.3}, minimum width=1.5cm, font=\normalsize]
						
						\node[state] (s0) {$s_0$};
						\node[state] (s1) [right=of s0] {$s_1$};
						
						\path[->, very thick]
						(s0) edge[bend left=45, color=myblue] node[above, text=myblue] {$a_0: 0.1$} (s1)
						(s1) edge[bend left=20, color=myblue] node[below, text=myblue] {$a_0: 0.1$} (s0)
						(s0) edge[loop above, color=myblue] node[text=myblue] {$a_0: 0.9$} (s0)
						(s1) edge[loop above, color=myblue] node[text=myblue] {$a_0: 0.9$} (s1);
						
						\path[->, very thick] 
						(s0) edge[bend right=45, color=myred] node[below, text=myred] {$a_1: 0.4$} (s1)
						(s1) edge[bend right=20, color=myred] node[above, text=myred] {$a_1: 0.01$} (s0)
						(s0) edge[loop below, color=myred] node[text=myred] {$a_1: 0.6$} (s0)
						(s1) edge[loop below, color=myred] node[text=myred] {$a_1: 0.99$} (s1);		
						
					\end{tikzpicture}
					\caption{Simulation Setup: Transition diagram of the primal MDP given in Appendix \ref{appendixh}. }\label{fig:9}
				\end{figure}

				As a case study, we specifically consider a benchmark setup for clarity and insight. In this setup, the sizes of the state space and the action space of the \textit{primal MDP} are both $2$, consisting of states $s_0,s_1$ and actions $a_0,a_1$\footnote{The binary state space is broadly representative of many real-world systems. For instance, it can model the \textbf{occurrence or absence} of critical events such as \textit{fires, industrial failures, security breaches, or anomalies in equipment status}. The binary action space can be interpreted as whether or not to apply a control intervention in response to the critical event, e.g. extinguishing a detected fire or repairing a malfunctioning piece of equipment.}. The \textit{primal MDP} tuple is detailed in Appendix \ref{appendixh}, and is visualized by transition diagram given in Fig. \ref{fig:9}. 
				
				In the primal MDP, the corresponding \textit{myopic} decision-making policy is
				\begin{equation}
					\pi_{\text{myopic}}(s_0)=a_0, \pi_{\text{myopic}}(s_1)=a_0,
				\end{equation}
				indicating a constant action regardless of state. Consequently, this decision-making policy does not utilize state information at the remote decision maker, and thus the value of information transmission is \textit{null} in this context.
				
				The \textit{long-term optimal decision} policy is solved by using RVI algorithm, and the solution is given as:
				\begin{equation}
					\pi^\star(s_0)=a_1, \pi^\star(s_1)=a_0,
				\end{equation}
				which clearly depends on the system state. In this case, the remote decision maker must rely on the potentially \textit{stale} state to select the ``\textit{right}'' action, highlighting the utility of information transmission under the decision-making policy.
				\subsubsection{\textcolor{black}{Finite-support Memory-less Delays}}
				\textcolor{black}{Throughout, we restrict attention to delay distributions with finite support. It is also practically justified, as many communication and control systems enforce a maximum admissible delay (via timeouts, buffer limits, or QoS deadlines). As a representative example, we consider a binary-valued random delay model where each delay $Y_i$ takes value $1$ with probability $p$ and $Y_{\max}$ with probability $1 - p$, i.e., $\Pr(Y_i = 1) = p,\Pr(Y_i = Y_{\max}) = 1 - p,\forall i\in\mathbb{Z}^+$} \textcolor{black}{(cf.\ \cite{DBLP:journals/tit/BedewySKS21})}. This binary-delay formulation enables tractable analysis while capturing key aspects of stochastic delay. 
				
				\textcolor{black}{Importantly, the proposed policy maps $(X_{S_i},Y_i,A_{i-1})$ to actions, so the solution complexity scales with $|\mathcal Y|$; a finite-support assumption keeps the policy solution computationally manageable. The analysis extends to \emph{any} finite-support distribution without memory, including truncated variants of heavy-tailed models\footnote{\textcolor{black}{Beyond truncation, handling \emph{unbounded-support} delays (e.g., geometric channels) directly would require additional conditions to ensure average-cost optimality; a systematic treatment is an interesting direction for future work.}}. Specifically, a geometric channel can be approximated by a \emph{truncated geometric} delay with parameter $q\in(0,1)$ and cutoff $Y_{\max}$: 
				\begin{equation}
					\Pr(Y=y)=\frac{q(1-q)^{y-1}}{1-(1-q)^{Y_{\max}}},\quad y=1,\dots,Y_{\max}.
				\end{equation}}
					
				\subsection{Discussions}
				
				\begin{table}[t]
					\centering
					\caption{\textcolor{black}{Relative cost reduction (\%) of goal-oriented sampling compared with baselines at selected average delays.}}
					\label{tabel6}
					\begin{tabular}{c|cccc}
						\toprule
						$\mathbb{E}[Y_i]$ &
						Zero-wait &
						AoI-optimal &
						Const-wait $z=2$ \\
						\midrule
						1.7 &  4.18\%  &  4.18\% &  9.98\% \\
						5.9 &  6.23\%  &  6.85\% &  6.09\% \\
						8.0 &  7.18\%  &  7.83\% &  6.66\% \\
						14.3 & 10.11\%& 9.87\% &  8.76\%\\
						\bottomrule
					\end{tabular}
				\end{table}

				Fig. \ref{fig:comparewithoutconstraint} and Fig. \ref{fig:comparewithoutconstraint2} illustrate the performance comparisons between various sampling and decision-making policies under different average delays $\mathbb{E}[Y_i]$. \textcolor{black}{In Fig. \ref{fig:comparewithoutconstraint}, the average delay is adjusted by adjusting $p$, with the maximum delay fixed at $Y_{\max} = 11$. In Fig.~\ref{fig:comparewithoutconstraint2}, the average delay is increased by enlarging $Y_{\max}$ while keeping $p = 0.3$ fixed. Across both scenarios, the proposed goal-oriented sampling strategy consistently outperforms all baseline methods in terms of minimizing the average cost. Table \ref{tabel6} reports the relative cost reduction $\eta_b = (C_b - \rho^{\star})/C_b \times 100\%$ of the proposed goal-oriented sampling compared with representative baselines at several average delay levels $\mathbb{E}[Y_i]$, where $C_b$ is the average cost of baseline $b$. The average delay is controlled by adjusting $Y_{\max}$. 
				For small delays, goal-oriented sampling achieves about 
					{4\%} reduction over zero-wait and AoI-optimal policies and nearly 
					{10\%} over the constant-wait policy. 
					As the delay increases to moderate levels 
					($\mathbb{E}[Y_i]=5.9$--$8$), the improvement grows to 
					{6--8\%}, and at large delays 
					($\mathbb{E}[Y_i]=14.3$) the gain further increases to about 
					{10\%} against zero-wait and AoI-optimal. The performance gain primarily arises from the fact that the proposed sampling policy is explicitly designed to optimize decision effectiveness rather than communication metrics alone. Unlike conventional AoI-optimal sampling policy, which treats information freshness and control decisions as separate optimization problems, goal-oriented sampling allocates sampling opportunities more judiciously, prioritizing updates that are most relevant to improving decision outcomes. In contrast, AoI-optimal sampling may transmit timely but less informative data, resulting in suboptimal decision performance despite lower AoI.}
				
				A key observation from Fig. \ref{fig:comparewithoutconstraint} is that the \textit{myopic decision-making policy}, which selects a fixed action regardless of state information, yields a constant cost independent of delay. \textcolor{black}{This myopic policy serves as a delay-agnostic baseline in the sense that it does not explicitly condition its decision rule on AoI beyond the last received state; the gap to the proposed goal-oriented sampling quantifies the benefit of explicitly leveraging AoI in decision-making.} This highlights an insight: in scenarios where the control policy is insensitive to informed state information, enhancing the communication channel (e.g., by reducing delay) \textcolor{black}{may} not lead to any improved goal-oriented decision-making performance. Moreover, Fig. \ref{fig:comparewithoutconstraint} demonstrates that the performance gap between the proposed \textit{goal-oriented sampling} policy and the \textit{myopic decision-making} policy gradually narrows as the average delay $\mathbb{E}[Y_i]$ increases. This convergence reflects a fundamental limitation in decision-making under communication constraints: when the delay becomes excessively large, \textcolor{black}{the received state information becomes so outdated that it no longer provides reliable guidance for current decisions. As a result, the optimal goal-oriented policy degenerates into state-independent behavior in high-delay regimes, where the influence of timely and fresh information on decision performance becomes negligible.}
				
				\textcolor{black}{Fig.~\ref{fig:geometric} illustrates the average cost achieved by various sampling strategies under a truncated geometric delay model with parameter $q = 0.3$. The $x$-axis represents the cutoff parameter $Y_{\max}$, which limits the maximum possible delivery delay. As $Y_{\max}$ increases, the truncated delay distribution gradually approaches the standard (non-truncated) geometric distribution, and the resulting average cost curves converge accordingly. Across the entire range of $Y_{\max}$, the proposed goal-oriented sampling strategy consistently outperforms all baselines, achieving significantly lower average cost---especially under large delay cutoffs. }

				Fig. \ref{fig:compareconstraint} illustrates the relationship between average cost and the maximum sampling frequency $f_{\max}$. The proposed goal-oriented sampling consistently achieves the lowest cost across all values of $f_{\max}$. As $f_{\max}$ increases, the performance of both \textit{goal-oriented sampling} and \textit{AoI-optimal} sampling improves gradually. Notably, the \textit{goal-oriented sampling} curve aligns with the \textit{sensitivity analysis} presented in Theorem \ref{Lemmasa}, which predicts that the average cost initially decreases with $f_{\max}$, then saturates to a constant value.

				In contrast, the uniform sampling policy in Fig. \ref{fig:compareconstraint} exhibits a \textit{U-shaped} cost curve. As $f_{\max}$ decreases, its performance converges to that of AoI-optimal sampling policy. This convergence highlights a key insight: under sparse sampling constraints, the transmitted information becomes too outdated to support effective state-dependent decision-making, resulting in uniformly poor performance. Conversely, when the sampling interval $d=1/f_{\max}$ becomes small, the system begins to accumulate a \textit{queuing backlog}, resulting in increased delivery delays. This queuing-induced staleness reduces the timeliness of information at the decision-maker, thereby degrading goal-oriented decision-making performance.

				\section{Conclusion \textcolor{black}{and Future Directions}}\label{sectionVI}
				In this paper, we have proposed a new remote MDP problem, namely AR-MDP in the time-lag MDP framework. Specifically, AoI, typically an optimization indicator for information \textit{freshness}, is incorporated into AR-MDP both as a controllable random processing delay and as critical side information to support remote decision-making. To investigate the fundamental trade-off between communication constraint and the decision-making effectiveness, we considered both \textit{sampling frequency constraint} and \textit{random delay} in this problem and developed low-complexity \textit{one-layer} algorithms, \textsc{OnePDSI} and \textsc{QuickBLP} to solve this problem efficiently. \textcolor{black}{Through theoretical analysis and experiments, we reveal how communication-induced information \textit{staleness} can negatively impact remote decision-making performance.}
				
				\textbf{Future directions.} Several promising directions remain open for future exploration. First, while this work focuses on finite-state, finite-action, and finite bounded-delay systems, extending the framework to handle \textit{continuous state and action spaces}, as well as delay distributions with infinite support, remains a challenging yet important direction. 
				\textcolor{black}{In particular, finiteness underlies several key steps of the present analysis, including the existence of stationary average-cost optimal policies and the associated optimality arguments on the lifted MDP, the finite-dimensional occupation-measure LP formulation in the constrained case, and the convergence and error guarantees of the proposed algorithms. Extending to infinite spaces will therefore require replacing these finite-dimensional tools, for example via discretization/quantization, or by invoking average-cost MDP/CMDP theory on Borel spaces under additional regularity conditions.}
				Second, the current model assumes that a sample, once sent, cannot be interrupted. Enabling \textit{preemptive or adaptive transmission mechanisms} may enhance the system responsiveness. \textcolor{black}{Third, relaxing the action-holding assumption and incorporating \textit{AoI-triggered decision-making mechanisms}, where sampling or control actions are explicitly adapted based on real-time AoI values, may allow for more efficient utilization of limited communication resources.} Moreover, while the proposed algorithms offer a favorable balance between optimality and complexity under the considered setting, future work could explore \textit{scalable reinforcement learning techniques} that generalize to unknown systems.

				\appendices
				\normalsize
				\textcolor{black}{\section{Proof of Lemma \ref{l1}}\label{l1proof}
				\subsection{Proof of Sufficient Statistics}
				We consider the finite-horizon value function
				\begin{equation}\label{eq:val-marginalize}
					\begin{aligned}
						J_t(\mathcal{I}_t)&\triangleq\max_{a_{t:T}}\mathbb{E}^{\,a_{t:T}}\left[\sum_{k=t}^T \mathcal{C}(X_k, a_k)|\mathcal{I}_t\right]\\
						&=\max_{a_{t:T}}\sum_{k=t}^T\sum_{s'\in\mathcal S} \mathcal{C}(s', a_{k})\cdot\Pr^{\,a_{t:T}}(X_k=s'\mid\mathcal{I}_t),
					\end{aligned}
				\end{equation}
				where $ \Pr^{\,a_{t:T}}$ (and $\mathbb E^{\,a_{t:T}}$) denotes the probability measure (and expectation) \emph{induced by} the chosen action sequence $a_{t:T}$.	
				Because $X_k$ is a controlled Markov process, the filtering distribution at time $k$ given $\mathcal I_t$ depends on $\big(X_{t-\Delta(t)},\Delta(t),a_{t-\Delta(t):t-1}\big)$ and not on earlier history, hence  
				\begin{equation}\label{eq:filter1}
					\begin{aligned}
						&\Pr^{\,a_{t:T}}\!\big(X_k\in\cdot \,\big|\,\mathcal I_t\big)
						\\ &=\Pr^{\,a_{t:T}}\!\big(X_k\in\cdot \,\big|\,X_{t-\Delta(t)},\Delta(t),a_{t-\Delta(t):t-1}\big).
					\end{aligned}
				\end{equation}
				For $t\in[D_i,D_{i+1})$ we have $t-\Delta(t)=S_i$ by \eqref{eq1} and $Y_i=D_i-S_i$, so
				\begin{equation}
					a_{S_i:t-1}
					=\big(a_{S_i:D_i-1},\,a_{D_i:t-1}\big)
					=(A_{i-1},\,a_{D_i:t-1}),
				\end{equation}
				and
				\begin{equation}\label{eq:filter2}
					\Pr^{\,a_{t:T}}\!\big(X_k\in\cdot \,\big|\,\mathcal I_t\big)
					= \Pr^{\,a_{t:T}}\!\left(X_k\in\cdot\mid X_{S_i},Y_i,A_{i-1},a_{D_i:t-1}\right).
				\end{equation}
			As assumed by \eqref{eq8}, the control is held constant on $[D_i,D_{i+1})$,
			 \begin{equation}
					a_{D_i:t-1}=(a_t,\ldots,a_t), \text{for } t\in [D_i,D_{i+1}),
				\end{equation}
				so the term $a_{D_i:t-1}$ is completely determined by $a_t$. We can therefore \emph{absorb} its effect into the induced measure and establish,
				\begin{equation}\label{eq27}
					\begin{aligned}
						J_t(\mathcal I_t)
						&=\max_{a_{t:T}}\sum_{k=t}^{T}\sum_{s'\in\mathcal S}
						\mathcal C(s',a_k)\,
						\Pr^{\,a_{t:T}}\!\big(X_k=s'\,\big|\,\mathcal G_i\big)\\
						&=\max_{a_{t:T}}\;
						\mathbb E^{\,a_{t:T}}\!\Bigg[\sum_{k=t}^{T}\mathcal C(X_k,a_k)\,\Big|\,\mathcal G_i\Bigg]
						\;=\; J_t(\mathcal G_i).
					\end{aligned}
				\end{equation}
				From Definition \ref{definition1} and we know that $\mathcal G_i=(X_{S_i},Y_i,A_{i-1})$ is a sufficient statistics of $\mathcal{I}_t$ for $t\in [D_i,D_{i+1})$.
				\subsection{Determining $u_t$ is equivalent to Determining $Z_i$}
				By \eqref{eq2} and \eqref{eq4}, within $[D_i,D_{i+1})$ the sampling control $u_t$ determines the next sampling epoch via
				\begin{equation}
					S_{i+1} \;=\; \inf\{\tau\ge D_i: a^S_\tau=1\} \;=\; D_i+Z_i,
				\end{equation}
				with $Z_i\in\mathbb N$ the waiting time. Conversely, any choice of $S_{i+1}$ (equivalently $Z_i$) induces a unique sequence $\{a^S_\tau\}_{\tau\in[D_i,D_{i+1})}$ by setting
				$a^S_\tau=0$ for $\tau\in[D_i,S_{i+1})$ and $a^S_{S_{i+1}}=1$ (and then proceeding to the next interval).
				Thus the mapping between $\{a^S_\tau\}_{\tau\in[D_i,D_{i+1})}$ and $S_{i+1}$ (or $Z_i$) is one-to-one, and optimizing over one is equivalent to optimizing over the other.}
				
				\textcolor{black}{\section{Proof of Lemma \ref{l3}}\label{proofofl2}
				\subsection{Proof of Part (i)}
				\subsubsection{$\rho^\star\le \lambda \iff U(\lambda)\le0$} 
				If $\rho^\star\le \lambda$, there exists a policy $\pi=(Z_0,A_0,Z_1,A_1,\cdots)$ such that \begin{equation}\label{eq22}
					\lim _{\mathrm{n} \rightarrow \infty} \frac{\sum_{i=0}^{n-1} \mathbb{E}_\pi\left[\sum_{t=\it{D}_{i}}^{\it{D}_{i+1}-1} \mathcal{C}(X_t, A_i)\right]}{\sum_{i=0}^{n-1} \mathbb{E}_\pi\left[Y_{i+1}+Z_i\right]}\le\lambda,
				\end{equation}
				which is equivalent to 
				\begin{equation}\label{fracless0}
					\resizebox{1\hsize}{!}{$\begin{aligned}
							\lim _{\mathrm{n} \rightarrow \infty} \frac{\frac{1}{n}\sum_{i=0}^{n-1} \mathbb{E}_\pi\left[\sum_{t=\it{D}_{i}}^{\it{D}_{i+1}-1} \mathcal{C}(X_t, A_i)\right]-\lambda\mathbb{E}_\pi\left[Y_{i+1}+Z_i\right]}{\frac{1}{n}\sum_{i=0}^{n-1} \mathbb{E}_\pi\left[Y_{i+1}+Z_i\right]}\le0.
						\end{aligned}$}
				\end{equation}
				Since  $Y_i>0$ and $0\le Z_i<\infty$, we have that $\lim_{\mathrm{n} \rightarrow \infty}\frac{1}{n}\sum_{i=0}^{n-1} \mathbb{E}_{
				\pi}\left[Y_{i+1}+Z_i\right]$ always exists, satisfying \begin{equation}0<\label{rege}\lim_{\mathrm{n} \rightarrow \infty}\frac{1}{n}\sum_{i=0}^{n-1} \mathbb{E}_{
				\pi}\left[Y_{i+1}+Z_i\right]<\infty.\end{equation} Thus, we have that the numerator of (\ref{fracless0}) satisfies\begin{equation}\label{25eq}
					\lim _{\mathrm{n} \rightarrow \infty}\frac{1}{n}\sum_{i=0}^{n-1} \mathbb{E}_\pi\left[\sum_{t=\it{D}_{i}}^{\it{D}_{i+1}-1} \mathcal{C}(X_t, A_i)\right]-\lambda\mathbb{E}_\pi[Z_i+Y_{i+1}]\le 0.
				\end{equation}
				This implies that the infimum of the left hand side of (\ref{25eq}) is also at most $0$, \emph{i.e}, $U(\lambda)\le 0$.}
				
				\textcolor{black}{On the contrary, when $U(\lambda)\le 0$, we can know that there exists a policy $\pi=(Z_0,A_0,Z_1,A_1,\cdots)$ that satisfies (\ref{25eq}). As (\ref{rege}) always holds, we can easily obtain that (\ref{eq22}) holds. Note that $\rho^\star$ is the infimum of the left hand side of (\ref{eq22}), we have \begin{equation}
					\rho^\star\le\lim _{\mathrm{n} \rightarrow \infty} \frac{\sum_{i=0}^{n-1} \mathbb{E}_\pi\left[\sum_{t=\it{D}_{i}}^{\it{D}_{i+1}-1} \mathcal{C}(X_t, a_i)\right]}{\sum_{i=0}^{n-1} \mathbb{E}_\pi\left[Y_{i+1}+Z_i\right]}\le\lambda.
				\end{equation}
				\subsubsection{$\rho^\star>\lambda \iff U(\lambda)>0$}
				If $\rho^\star>\lambda$, we have that for any policy $(Z_0,A_0,Z_1,A_1,\cdots)$, the following inequality always holds
				\begin{equation}\label{27}
					\lim _{\mathrm{n} \rightarrow \infty} \frac{\sum_{i=0}^{n-1} \mathbb{E}\left[\sum_{t=\it{D}_{i}}^{\it{D}_{i+1}-1} \mathcal{C}(X_t, A_i)\right]}{\sum_{i=0}^{n-1} \mathbb{E}\left[Y_{i+1}+Z_i\right]}>\lambda.
				\end{equation}		
				Since (\ref{rege}) holds, we have that for any policy $\pi$, 
				\begin{equation}\label{aaa}
					\lim _{\mathrm{n} \rightarrow \infty}\frac{1}{n}\sum_{i=0}^{n-1} \mathbb{E}_\pi\left[\sum_{t=\it{D}_{i}}^{\it{D}_{i+1}-1} \mathcal{C}(X_t, A_i)\right]-\lambda\mathbb{E}_\pi[Z_i+Y_{i+1}]> 0.
				\end{equation}
				Since (\ref{aaa}) holds for any policies, it follows that the infimum value of the left-hand side (LHS) of (\ref{aaa}) is also greater than $0$, implying that $\rho^\star>\lambda$.}
				
				\textcolor{black}{When $U(\lambda) > 0$, it is established that condition (\ref{aaa}) is satisfied for any policy sequence $(Z_0, A_0, Z_1, A_1, \cdots)$. Given that (\ref{rege}) always holds, it follows directly that (\ref{27}) holds for any policies, implying that the infimum of the LHS of (\ref{27}) is also greater than $\lambda$, \emph{i.e.}, $\rho^\star>\lambda$.
				\subsection{Proof of Part (ii)}
				By Part (i), $U(\lambda)=0$ iff $\lambda=\rho^\star$. Let $\lambda=\rho^\star$. If $\pi^\star$ is optimal for Problem~\ref{p4}, then 
				\begin{equation}\label{29}
					\lim _{\mathrm{n} \rightarrow \infty}\frac{1}{n}\sum_{i=0}^{n-1} \mathbb{E}\left[\sum_{t=\it{D}_{i}}^{\it{D}_{i+1}-1} \mathcal{C}(X_t, A_i)\right]-\lambda(Z_i+Y_{i+1})= 0,
				\end{equation}
				which implies that 
				\begin{equation}\label{eq23}
					\lim _{\mathrm{n} \rightarrow \infty} \frac{\sum_{i=0}^{n-1} \mathbb{E}_{\pi^*}\left[\sum_{t=\it{D}_{i}}^{\it{D}_{i+1}-1} \mathcal{C}(X_t, A_i)\right]}{\sum_{i=0}^{n-1} \mathbb{E}_{\pi^*}\left[Y_{i+1}+Z_i\right]}=\lambda.
				\end{equation}
				Note that $\rho^\star=\lambda$, we have that for the policy $\pi^*$, \begin{equation}
					\rho^\star=\lim _{\mathrm{n} \rightarrow \infty} \frac{\sum_{i=0}^{n-1} \mathbb{E}\left[\sum_{t=\it{D}_{i}}^{\it{D}_{i+1}-1} \mathcal{C}(X_t, A_i)\right]}{\sum_{i=0}^{n-1} \mathbb{E}\left[Y_{i+1}+Z_i\right]},
				\end{equation}
				which infers that policy $\pi^*$ is also the optimal policy of Problem \ref{p3}.
				\subsection{Proof of Part (iii)}
				From Part (i), we know that proving Part (iii) is equivalent to prove that $U(\lambda)$ is monotonically non-increasing in terms of $\lambda$, \emph{i.e.}, for any $\Delta\lambda>0$, $U(\lambda+\Delta\lambda)\le U(\lambda)$. This is verified by the following inequalities:
				\begin{equation}
					\resizebox{1\hsize}{!}{$
						\begin{aligned}	
							& U(\lambda+\Delta\lambda)=\\
							&\inf_{\phi_{0:\infty}} \lim _{\mathrm{n} \rightarrow \infty}\frac{1}{n} {\sum_{i=0}^{n-1}\left\{ \mathbb{E}\left[\sum_{t=\it{D}_{i}}^{\it{D}_{i+1}-1} \mathcal{C}(X_t, A_i)\right]-(\lambda+\Delta\lambda)\mathbb{E} \left[Z_i+Y_{i+1}\right]\right\}}\\
							&=\inf _{\phi_{0:\infty}}\left\{\lim _{\mathrm{n} \rightarrow \infty}\frac{1}{n} {\sum_{i=0}^{n-1}\left\{ \mathbb{E}\left[\sum_{t=\it{D}_{i}}^{\it{D}_{i+1}-1} \mathcal{C}(X_t, A_i)\right]-\lambda\mathbb{E} \left[Z_i+Y_{i+1}\right]\right\}}\right.\\
							&\left.-\lim _{\mathrm{n} \rightarrow \infty}\frac{1}{n}\sum_{i=0}^{n-1} \Delta\lambda\mathbb{E} [Z_i+Y_{i+1}]       \right\}\\
							&\le \inf _{\phi_{0:\infty}}\lim _{\mathrm{n} \rightarrow \infty}\frac{1}{n} {\sum_{i=0}^{n-1}\left\{ \mathbb{E}\left[\sum_{t=\it{D}_{i}}^{\it{D}_{i+1}-1} \mathcal{C}(X_t, A_i)\right]-\lambda\mathbb{E} \left[Z_i+Y_{i+1}\right]\right\}}\\
							&=U(\lambda).
						\end{aligned}$}
				\end{equation}
				Thus, we have that $\lambda=\rho^\star$ is the unique root of $U(\lambda)=0$.}
				
				\textcolor{black}{\section{Transition Probability of $\mathscr{P}_{\mathrm{MDP}}(\lambda)$}\label{appendixc}
				Recall that $\mathcal G_i=(X_{S_i},Y_i,A_{i-1})$ and $\mathcal G_{i+1}=(X_{S_{i+1}},Y_{i+1},A_i)$.
				We assume $\{Y_i\}$ are i.i.d. and independent of the source process $\{X_t\}$ and of the actions. 
				We derive $\Pr(\mathcal G_{i+1}\mid \mathcal G_i,Z_i,A_i)$ by computing the three marginals $\Pr(X_{S_{i+1}}\mid\mathcal G_i,Z_i,A_i)$,
				$\Pr(Y_{i+1}\mid\mathcal G_i,Z_i,A_i)$, and
				$\Pr(A_i\mid\mathcal G_i,Z_i,A_i)$.
				\subsection{$\Pr\left(X_{S_{i+1}}|\mathcal{G}_i,Z_i,A_i\right)$}\label{CA}
				Condition on $\mathcal G_i=(s,\delta,a)$, where $\delta=Y_i$ and $a=A_{i-1}$.
				From $t=S_i$ to $t=D_i-1$ the action is constant and equal to $A_{i-1}=a$, so the source evolves for $\delta$ steps under transition matrix $\mathbf P_{a}$, yielding a $\delta$-step kernel $\mathbf P_{a}^{\delta}$.
				From $t=D_i$ to $t=S_{i+1}-1$ the action is $A_i$, so the source further evolves for $Z_i$ steps under $\mathbf P_{A_i}$, i.e., kernel $\mathbf P_{A_i}^{Z_i}$ (with the convention $\mathbf P_{A_i}^{0}=I$).
				By the semigroup property, we have
				\begin{equation}
					\Pr\!\left(X_{S_{i+1}}=s'\mid \mathcal G_i=(s,\delta,a),Z_i,A_i\right)
					=\big[\mathbf P_{a}^{\delta}\,\mathbf P_{A_i}^{Z_i}\big]_{s\times s'}.
				\end{equation}				
				\subsection{$\Pr(Y_{i+1}|\mathcal{G}_i,Z_i,A_i)$}
				Since $\{Y_i\}$ are i.i.d. and independent of $(\mathcal G_i,Z_i,A_i)$,
				\begin{equation}
								\Pr\!\left(Y_{i+1}=\delta'\mid \mathcal G_i=(s,\delta,a),Z_i,A_i\right)
					=\Pr(Y_{i+1}=\delta').
				\end{equation}
				\subsection{$\Pr(A_i|\mathcal{G}_i,Z_i,A_i)$}
				The third component of $\mathcal G_{i+1}$ is the record of the action taken during $[D_i,S_{i+1})$, namely $A_i$. Hence,
				\begin{equation}
					\Pr\left(A_i=a'|\mathcal{G}_i=(s,\delta,a),Z_i,A_i\right)=\mathbbm{1}\{a'=A_i\}.
				\end{equation}
				\subsection{Product form and the transition kernel}
				By the stated independence, $Y_{i+1}$ is independent of $X_{S_{i+1}}$ given $(\mathcal G_i,Z_i,A_i)$, and the value of $A_i$ is deterministic under the conditioning. Therefore,
				\begin{equation}
				\begin{aligned}
					&\Pr\!\left(\mathcal G_{i+1}=(s',\delta',a')\mid \mathcal G_i=(s,\delta,a),Z_i,A_i\right)\\
					&\qquad= \Pr(Y_{i+1}=\delta')\cdot
					\big[\mathbf P_{a}^{\delta}\,\mathbf P_{A_i}^{Z_i}\big]_{s\times s'}
					\cdot \mathbbm{1}\{a'=A_i\},
				\end{aligned}
			\end{equation}
				which is exactly the transition probability stated in~\eqref{tran}.} 
				
				\textcolor{black}{\section{Proof of Lemma \ref{l4}}\label{appendixd}
				Recall that $\mathcal G_i=(X_{S_i},Y_i,A_{i-1})$, $Z_i=S_{i+1}-D_i$, $D_{i+1}=S_{i+1}+Y_{i+1}$, and $\{Y_i\}$ are i.i.d. and independent of the source process $\{X_t\}$ and of the actions. With the per-epoch cost in \eqref{cfun}, we have
				\begin{equation}\label{36}
					\begin{aligned}
						&\inf _{\phi_{0: \infty}} \limsup _{n \rightarrow \infty} \frac{1}{n} \mathbb{E}\left[\sum_{i=1}^n g(\mathcal{G}_i, A_i,Z_i;\lambda)\right]\\
						&=\inf _{\phi_{0: \infty}} \limsup _{T \rightarrow \infty} \frac{1}{n} \mathbb{E}\left[\sum_{i=1}^n q(\mathcal{G}_i, A_i,Z_i)-\lambda\mathbb{E}[Z_i+Y_{i+1}]\right].
					\end{aligned}
				\end{equation}
				Hence it suffices to show
				\begin{equation}\label{eq:goal}
					\mathbb{E}\!\left[\sum_{t=D_i}^{D_{i+1}-1}\mathcal C(X_t,A_i)\right]
					=\mathbb{E}\!\left[q(\mathcal G_i,Z_i,A_i)\right].
				\end{equation}
				For each epoch $i$, the expectation $\mathbb{E}\!\left[\sum_{t=D_i}^{D_{i+1}-1}\mathcal C(X_t,A_i)\right]$ can be decomposed as \eqref{decompose}.
				\begin{figure*}
					\begin{equation}\label{decompose}
						\color{black}
						\begin{aligned}
							\mathbb{E}\!\left[\sum_{t=D_i}^{D_{i+1}-1}\mathcal C(X_t,A_i)\right]
							&=\mathbb{E}\!\left[\;
							\mathbb{E}\!\left[\left.\sum_{t=D_i}^{D_i+Z_i+Y_{i+1}-1}\mathcal C(X_t,A_i)\,\right|\,\mathcal G_i,Z_i,A_i\right]\right] \\
							&=\mathbb{E}\!\left[\;
							\mathbb{E}_{Y_{i+1}}\!\left[\sum_{t=D_i}^{D_i+Z_i+Y_{i+1}-1}\sum_{s'\in\mathcal S}\mathcal C(s',A_i)\,
							\Pr\!\left(X_t=s'\mid \mathcal G_i,Z_i,A_i\right)\right]\right].
						\end{aligned}
					\end{equation}
					\hrulefill
				\end{figure*} 
				From $t=S_i$ to $t=D_i-1$ the action is $A_{i-1}=a$, i.e., $\delta$ steps under transition matrix $\mathbf P_a$, yielding $\mathbf P_a^{\delta}$. 
				From $t=D_i$ to any $t\in\{D_i,\ldots,D_{i+1}-1\}$, the action is $A_i$, i.e., $(t-D_i)$ further steps under $\mathbf P_{A_i}$, yielding $\mathbf P_{A_i}^{\,t-D_i}$ (with $\mathbf P_{u}^{0}=I$).
				By the semigroup property,
				\begin{equation}\label{eq100}
					\Pr\!\left(X_t=s'\mid \mathcal G_i=(s,\delta,a),Z_i,A_i\right)
					=\big[\mathbf P_{a}^{\delta}\,\mathbf P_{A_i}^{\,t-D_i}\big]_{s\times s'}.
				\end{equation}
				Substitute \eqref{eq100} into \eqref{decompose} and reindex $\tau=t-D_i\in\{0,\ldots,Z_i+Y_{i+1}-1\}$, then exchange the sums:
				\begin{equation}
					\begin{aligned}
						&\mathbb{E}\!\left[\sum_{t=D_i}^{D_{i+1}-1}\mathcal C(X_t,A_i)\right]\\
						&=\mathbb{E}\!\left[\sum_{s'\in\mathcal S}\mathcal C(s',A_i)\;
						\mathbb{E}_{Y_{i+1}}\!\left[\sum_{\tau=0}^{Z_i+Y_{i+1}-1}
						\big[\mathbf P_{A_{i-1}}^{Y_i}\,\mathbf P_{A_i}^{\,\tau}\big]_{X_{S_i}\times s'}\right]\right].
					\end{aligned}
				\end{equation}
				By the definition in \eqref{15}, the right-hand side equals $\mathbb{E}[q(\mathcal G_i,Z_i,A_i)]$, which proves \eqref{eq:goal}. Therefore, we establish:
				\begin{equation}
					\begin{aligned}
						&\limsup_{n\to\infty}\frac{1}{n}\,
						\mathbb{E}\!\left[\sum_{i=0}^{n-1}\big(q(\mathcal G_i,Z_i,A_i)-\lambda f(Z_i)\big)\right]
						=
						\\
						&\limsup_{n\to\infty}\frac{1}{n}\,
						\mathbb{E}\!\left[\sum_{i=0}^{n-1}\left(\sum_{t=D_i}^{D_{i+1}-1}\mathcal C(X_t,A_i)
						-\lambda\,(Z_i+\mathbb E[Y_{i+1}])\right)\right],
					\end{aligned}
				\end{equation}
				which is exactly the objective of Problem~\ref{p4}. Hence $\mathscr{P}_{\mathrm{MDP}}(\lambda)$ is equivalent to Problem~\ref{p4}.}
				
				\textcolor{black}{\section{Proof of Theorem \ref{thm:lifted-unichain}}\label{appendixe}
					Fix an arbitrary stationary deterministic policy 
					$\pi:\mathcal S\times\mathcal Y\times\mathcal A\to\mathcal A\times\mathcal Z$
					and an arbitrary initial lifted state $\mathcal G_0=(s_0,\delta_0,a_0)\in\mathcal S\times\mathcal Y\times\mathcal A$. Denote by $K^\pi$ the one-step transition kernel of the lifted chain under $\pi$, and by $(K^\pi)^m$ its $m$-step kernel.
					Define the \emph{common small set}
					\[
					\mathcal V \;\triangleq\; \{s^\star\}\times \mathcal Y \times \mathcal A \;\subseteq\; \mathsf S\times\mathcal Y\times\mathcal A.
					\]
					We claim that there is a policy-independent one-block minorization:
					\begin{equation}\label{eq:block-minorization}
						\left[(K^\pi)^m\right]_{(s,\delta,a)\times\mathcal{V}}\ \ge\ \epsilon
						\qquad\forall(s,\delta,a)\in\mathcal S\times\mathcal Y\times\mathcal A.
					\end{equation}
					To see this, fix $(s,\delta,a)$ and unfold $m$ steps. Along any admissible length-$m$ sequence 
					$\{(A_t,Z_t,\delta_t)\}_{t=0}^{m-1}$ generated by the policy $\pi$ with linkage constraints $a_{t+1}=A_t$, the $s$-component transition from $s$ to $s^\star$ equals the matrix product on the left of \eqref{eq:entrance}. 
					By the hypothesis \eqref{eq:entrance}, this probability is at least $\epsilon$, uniformly over all such sequences. 
					Since $\mathcal C$ does not restrict the $(\delta,a)$ components at time $m$, we obtain
					\begin{equation}
						\left[(K^\pi)^m\right]_{(s,\delta,a)\times\mathcal{V}}\ge \epsilon>0,
					\end{equation}
					which proves \eqref{eq:block-minorization}.}
					
					\textcolor{black}{Now apply the Markov property on block times $\{0,m,2m,\ldots\}$. For every $n\in\mathbb N$, and the policy $\pi$
					\begin{equation}
						\Pr\!\left\{\mathcal G_{km}\notin\mathcal C\ \text{for }k=1,\ldots,n\ \Big|\ \mathcal G_0=(s,\delta,a)\right\}
						\;\le\; (1-\epsilon)^n.
					\end{equation}
					Letting $n\to\infty$ shows that the hitting time 
					$\tau_{\mathcal C}\triangleq \inf\{k\ge1:\mathcal G_{km}\in\mathcal C\}$ is almost surely finite.
					Restarting from any $y\in\mathcal C$ and repeating the same argument yields that the chain visits $\mathcal C$ infinitely often with probability one, uniformly over the initial state and the policy $\pi$.}
					
					\textcolor{black}{We now establish uniqueness of the recurrent class. Suppose, by contradiction, that there exist two disjoint recurrent classes $\mathcal R_1,\mathcal R_2\subseteq \mathsf S\times\mathcal Y\times\mathcal A$ under $\pi$. 
					Recurrent classes are closed, and starting from any $x\in\mathcal R_j$ the chain returns to $\mathcal R_j$ infinitely often almost surely. 
					But from any starting point the chain also visits $\mathcal C$ infinitely often almost surely, hence $\mathcal C\subseteq \mathcal R_j$ for $j=1,2$, which implies $\mathcal R_1\cap\mathcal R_2\supseteq\mathcal C\neq\varnothing$, which is a contradiction. 
					Therefore, there is at most one recurrent class. Since the state space is finite, at least one recurrent class exists, and uniqueness follows. 
					This proves that the lifted MDP is unichain in the sense stated.}

				\textcolor{black}{\section{Proof of Lemma \ref{l5}}\label{appendixf}
				\subsection{Proof of $\rho^{\star}\ge\min_{s,a}\mathcal{C}(s,a)$}
				For all $t$ and any policy,
				$\mathcal C(X_t,a_t)\ge \min_{s,a}\mathcal C(s,a)$ almost surely. Hence
				\begin{equation}
				\begin{aligned}
					\rho^\star
					&= \inf_{\phi_{0:\infty}}\limsup_{T\to\infty}\frac{1}{T}\,
					\mathbb{E}\!\left[\sum_{t=1}^{T}\mathcal C(X_t,a_t)\right] \\
					&\ge \inf_{\phi_{0:\infty}}\limsup_{T\to\infty}\frac{1}{T}\,
					\mathbb{E}\!\left[\sum_{t=1}^{T}\min_{s,a}\mathcal C(s,a)\right]
					= \min_{s,a}\mathcal C(s,a).
				\end{aligned}
				\end{equation}
				\subsection{Proof of $\rho^{\star}\le\min_a\sum_{s\in\mathcal{S}}{\pi}_a(s)\cdot \mathcal{C}(s,a)$}
				\begin{figure*}[b!]
					\hrulefill
					\begin{equation}\label{70}
						\color{black}
						\begin{aligned}
							0
							&=\min _{A_i, Z_i}\big\{g(\gamma^{\mathrm{ref}},A_i,Z_i;\rho^{\star})+\mathbb{E}[W^*\left(\gamma'\right)|\gamma^{\mathrm{ref}},Z_i,A_i]\big\}
							=\min _{A_i, Z_i}\big\{q(\gamma^{\mathrm{ref}},A_i,Z_i)-\rho^{\star}\cdot f(Z_i)+\mathbb{E}[W^*\left(\gamma'\right)|\gamma^{\mathrm{ref}},Z_i,A_i]\big\}\\
							&=\min _{A_i, Z_i}\left\{f(Z_i)\cdot\left(\frac{q(\gamma^{\mathrm{ref}},A_i,Z_i)+\mathbb{E}[W^*\left(\gamma'\right)|\gamma^{\mathrm{ref}},Z_i,A_i]}{f(Z_i)}-\rho^{\star}\right) \right\},
						\end{aligned}\tag{113}
					\end{equation}
				\end{figure*}
				Restrict to the subclass of policies that fix the action to a constant $a\in\mathcal A$. Then $\{X_t\}$ evolves as a time-homogeneous Markov chain with transition matrix $\mathbf P_a$.
				Assume that $\mathbf P_a$ admits a stationary distribution $\pi_a$ and that the chain is ergodic so that time averages converge to stationary expectations. By the Markov chain ergodic theorem,
				\begin{equation}
					\lim_{T\to\infty}\frac{1}{T}\,\mathbb{E}\!\left[\sum_{t=1}^{T}\mathcal C(X_t,a)\right]
					= \sum_{s\in\mathcal S}\pi_a(s)\,\mathcal C(s,a).
				\end{equation}
				Therefore, for each $a$,
				\begin{equation}\label{42}
					\begin{aligned}
						\rho^\star \le
						\limsup_{T\to\infty}\frac{1}{T}\,\mathbb{E}\!\left[\sum_{t=1}^{T}\mathcal C(X_t,a)\right]
						= \sum_{s}\pi_a(s)\,\mathcal C(s,a),
					\end{aligned}
				\end{equation}
				and minimizing over $a$ yields the stated upper bound.}
				
				\textcolor{black}{\section{Proof of Theorem \ref{the2}}\label{appendixg}
				From \cite[Proposition 7.4.1]{bertsekas2012dynamic}, we know that for any $\lambda$, the optimal value of Problem \ref{p4}, which is $U(\lambda)$, is the same for all initial states and some values $V^*(\gamma;\lambda), \gamma\in\mathcal{S}\times\mathcal{Y}\times\mathcal{A}$ and satisfies the following Bellman equation:
				\begin{align}\label{64}
					&V^*(\gamma;\lambda)+U(\lambda)=\nonumber\\&\min _{A_i, Z_i}\big\{g(\gamma,A_i,Z_i;\lambda)+ \mathbb{E}[V^*\left(\gamma';\lambda\right)|\gamma,A_i,Z_i]\big\},
				\end{align}
				Substituting $\lambda=\rho^{\star}$ and $U(\rho^{\star})=0$ into the Bellman equation, 
				\begin{align}\label{65}
					&V^*(\gamma;\rho^{\star})=\nonumber\\&\min _{A_i, Z_i}\big\{g(\gamma,A_i,Z_i;\rho^{\star})+ \mathbb{E}[V^*\left(\gamma';\rho^{\star}\right)|\gamma,A_i,Z_i]\big\},
				\end{align}
				Similar to the RVI algorithm, we introduce the \textit{relative value function} defined as 
				\begin{equation}\label{66}
					W^*(\gamma)\triangleq V^*(\gamma;\rho^{\star})-V^*(\gamma^{\mathrm{ref}};\rho^{\star}),
				\end{equation}
				where $\gamma^{\mathrm{ref}}$ is called \textit{reference state} and can be arbitrarily chosen from space $\mathcal{S}\times\mathcal{Y}\times\mathcal{A}$. Then, substituting (\ref{66}) into (\ref{65}) yields
				\begin{equation}\label{67}
					\begin{aligned}
						W^*(\gamma)=\min _{A_i, Z_i}\big\{g(\gamma,A_i,Z_i;\rho^{\star})+\mathbb{E}[W^*\left(\gamma'\right)|\gamma,Z_i,A_i]\big\},
					\end{aligned}
				\end{equation}
				Applying $\gamma=\gamma^{\mathrm{ref}}$ in (\ref{66}) and (\ref{67}) leads to
				\begin{equation}\label{68}
					W^*(\gamma^{\mathrm{ref}})=0,\forall \gamma\in\mathcal{S}\times\mathcal{Y}\times\mathcal{A},
				\end{equation}
				and 
				\begin{equation}\label{69}
					\begin{aligned}
						&W^*(\gamma^{\mathrm{ref}})=\\&\min _{A_i, Z_i}\big\{g(\gamma^{\mathrm{ref}},A_i,Z_i;\rho^{\star})+\mathbb{E}[W^*\left(\gamma'\right)|\gamma^{\mathrm{ref}},Z_i,A_i]\big\},
					\end{aligned}
				\end{equation}
				Then, substituting (\ref{68}) and $g(\gamma^{\mathrm{ref}},A_i,Z_i;\rho^{\star})=q(\gamma^{\mathrm{ref}},A_i,Z_i)-\rho^{\star}\cdot f(Z_i)$ into (\ref{69}) yields (\ref{70}).	
				\stepcounter{equation}}
				
				\textcolor{black}{Because $f(Z_i)>0$, we have that (\ref{70}) holds only if
				\begin{equation}\label{71}
					\min_{A_i,Z_i}\left\{\frac{q(\gamma^{\mathrm{ref}},A_i,Z_i)+\mathbb{E}[W^*\left(\gamma'\right)|\gamma^{\mathrm{ref}},Z_i,A_i]}{f(Z_i)}-\rho^{\star}\right\}=0.
				\end{equation}
				Moving $\rho^{\star}$ to the RHS of (\ref{71}) yields
				\begin{equation}
					\rho^{\star}=\min_{A_i,Z_i}\left\{\frac{q(\gamma^{\mathrm{ref}},A_i,Z_i)+\mathbb{E}[W^*\left(\gamma'\right)|\gamma^{\mathrm{ref}},Z_i,A_i]}{f(Z_i)}\right\}.
				\end{equation}
				We thus accomplish the proof.}
				
				\section{The Primal MDP}\label{appendixh}
				We consider the following parameter setup: 
				\begin{itemize}
					\item The \textit{state space} is a binary space: $\mathcal{S}=\{s_0,s_1\}$.
					\item The \textit{action space} is a binary space: $\mathcal{A}=\{a_0,a_1\}$.
					\item The \textit{transition probability matrix} of $X_t$ is \begin{equation}
						\mathbf{P}_{a_0}=\begin{bmatrix}
							0.9 &0.1\\0.1 &0.9
						\end{bmatrix}, \mathbf{P}_{a_1}=\begin{bmatrix}
							0.6 &0.4\\0.01 &0.99
						\end{bmatrix}.
					\end{equation}
					\item The cost function $\mathcal{C}(X_t,a_t)$ is given as \begin{equation}
						\begin{aligned}
							&\mathcal{C}(s_0,a_0)=40,\mathcal{C}(s_0,a_1)=60,\\
							&\mathcal{C}(s_1,a_0)=0,\mathcal{C}(s_1,a_1)=20.
						\end{aligned}
					\end{equation}
				\end{itemize}
				\section{Proof of Theorem \ref{convergence1} \textcolor{black}{and Theorem \ref{theorem3:convergence}}}\label{proof:convergence:theorem2}
				The proof is divided into two parts. First, we prove in Section \ref{subGa} that the limits $\lim_{K\to\infty}\tilde{U}_K(\lambda)$ and $\lim_{K\to\infty}\tilde{V}_K(\gamma;\lambda),\forall \gamma\in\mathcal{S}\times\mathcal{Y}\times\mathcal{A}$ are both finite; Second, we explicitly establish that $\lim_{K \to \infty}{\tilde{U}_K(\lambda)}=U(\lambda)$ and $\lim_{K \to \infty}\tilde{V}_K(\gamma;\lambda)={V}^\star(\gamma;\lambda)/\tau$ in Section \ref{subGb}.
				\subsection{Convergence of \eqref{MRVI}}\label{subGa}
				Denote $Z^{(K)}(\gamma)$, $A^{(K)}(\gamma)$ as the waiting time and controlled action that achieves the minimum in the $K$-th relation:
				\begin{equation}\label{eq82}
					\begin{aligned}					&(A^{(K)}(\gamma),Z^{(K)}(\gamma))=\\&\mathop{\arg\min}_{A_i,Z_i}\Big\{g(\gamma,Z_i,A_i;\lambda)+\tau \mathbb{E}[\tilde{V}_{K}\left(\gamma';\lambda\right)|\gamma,Z_i,A_i]\Big\}.
					\end{aligned}
				\end{equation}
				Define $\widetilde{\mathbf{V}}_K(\lambda)\in\mathbb{R}^{|\mathcal{S}\times\mathcal{Y}\times\mathcal{A}|}$ as the column vector formed by stacking the values $\tilde{V}_{K}(\gamma;\lambda)$ for all $\gamma\in\mathcal{S}\times\mathcal{Y}\times\mathcal{A}$, and let $\mathbf{e}$ an all-one vector of the same dimension. Similarly, define $\mathbf{g}_K(\lambda)$ as the column vector composed of the immediate costs $g(\gamma,Z^{(K)}(\gamma),A^{(K)}(\gamma);\lambda)$ arranged under the same indexing scheme. Let $\mathbf{P}(K)\in\mathbb{R}^{|\mathcal{S}\times\mathcal{Y}\times\mathcal{A}|\times|\mathcal{S}\times\mathcal{Y}\times\mathcal{A}|}$ denote the transition probability matrix where the $(i,j)$-th entry corresponds to the transition probability from $\gamma_i$ to $\gamma_j$ under the control $(Z^{(K)}(\gamma_i),A^{(K)}(\gamma_i))$, with $\gamma_i,\gamma_j$ indexed according to the same fixed ordering of $\mathcal{S}\times\mathcal{Y}\times\mathcal{A}$. Similarly, let $\widetilde{\mathbf{P}(K)}\in\mathbb{R}^{|\mathcal{S}\times\mathcal{Y}\times\mathcal{A}|\times|\mathcal{S}\times\mathcal{Y}\times\mathcal{A}|}$ be defined using the modified transition probabilities $\widetilde{p_{\gamma_i\gamma_j}}(Z^{(K)}(\gamma),A_K(\gamma))$. Under this notation, $\tau$-RVI in \eqref{MRVI} can be equivalently written in vector form as:
				\begin{equation}\label{eq83vw}
					\begin{aligned}
						\tilde{\mathbf{V}}_{K+1}(\lambda)&=(1-\tau)\tilde{\mathbf{V}}_{K}(\lambda)+\mathbf{g}_K(\lambda)\\&+\tau\mathbf{P}(K)\tilde{\mathbf{V}}_{K}(\lambda)-\tilde{U}_{K+1}(\lambda)\mathbf{e}.
					\end{aligned}
				\end{equation}
				From \eqref{eq82}, the pair $A^{(K)}(\gamma),Z^{(K)}(\gamma)$ is selected to minimize the following objective
				\begin{equation}
					g(\gamma,Z_i,A_i;\lambda)+\tau \mathbb{E}[\tilde{V}_{K}\left(\gamma';\lambda\right)|\gamma,Z_i,A_i].
				\end{equation}
				This implies that the chosen action $A^{(K)}(\gamma),Z^{(K)}(\gamma)$ at the $K$-th iteration yields an objective value no greater than that obtained by any other actions. In vector form, this yields the following inequality: 
				\begin{equation}
					\mathbf{g}_K(\lambda)+\tau\mathbf{P}(K)\tilde{\mathbf{V}}_{K}(\lambda)\le\mathbf{g}_t(\lambda)+\tau\mathbf{P}(t)\tilde{\mathbf{V}}_{K}(\lambda), \forall K,t\ge0,
				\end{equation}
				which can be combined with the update rule in \eqref{eq83vw} to derive the following upper bound:
				\begin{equation}\label{eq83}
					\begin{aligned}
						\tilde{\mathbf{V}}_{K+1}(\lambda)\le
						(1-\tau)\tilde{\mathbf{V}}_{K}(\lambda)+\mathbf{g}_{K-1}(\lambda)\\+\tau\mathbf{P}(K-1)\tilde{\mathbf{V}}_{K}(\lambda)-\tilde{U}_{K+1}(\lambda)\mathbf{e}.
					\end{aligned}
				\end{equation}
				Similarly, applying the update equation \eqref{eq83vw} at iteration $K$ and letting $t=K$ yields:
				\begin{equation}\label{eq84}
					\begin{aligned}
						\tilde{\mathbf{V}}_{K}(\lambda)&\le
						(1-\tau)\tilde{\mathbf{V}}_{K-1}(\lambda)+\mathbf{g}_{K}(\lambda)+\\&\quad\tau\mathbf{P}(K)\tilde{\mathbf{V}}_{K-1}(\lambda)-\tilde{U}_{K}(\lambda)\mathbf{e}.
					\end{aligned}
				\end{equation}
				Let us define the difference between successive relative value function iterates as \begin{equation}
					\widetilde{\varDelta_K\mathbf{V}(\lambda)}\triangleq\tilde{\mathbf{V}}_{K+1}(\lambda)-\tilde{\mathbf{V}}_{K}(\lambda).
				\end{equation}
				Subtracting \eqref{eq84} from \eqref{eq83} yields recursive inequalities that characterize the evolution of the value difference:
				\begin{subequations}\label{ineq85}
					\begin{align}
						\widetilde{\varDelta_K\mathbf{V}(\lambda)}&\le(1-\tau)\widetilde{\varDelta_{K-1}\mathbf{V}(\lambda)}\notag+\\&\tau\mathbf{P}(K-1)\widetilde{\varDelta_{K-1}\mathbf{V}(\lambda)}+(\tilde{U}_{K}(\lambda)-\tilde{U}_{K+1}(\lambda))\mathbf{e},\\
						\widetilde{\varDelta_K\mathbf{V}(\lambda)}&\ge(1-\tau)\widetilde{\varDelta_{K-1}\mathbf{V}(\lambda)}+\notag\\&\tau\mathbf{P}(K)\widetilde{\varDelta_{K-1}\mathbf{V}(\lambda)}+(\tilde{U}_{K}(\lambda)-\tilde{U}_{K+1}(\lambda))\mathbf{e}.
					\end{align}
				\end{subequations}
				From \eqref{ptransform}, we know that the matrix $\widetilde{\mathbf{P}(K)}$ satisfies:
				\begin{equation}\label{transpmatrix}
					\widetilde{\mathbf{P}(K)}=(1-\tau)\mathbf{I}+\tau\mathbf{P}(K), \text{ for } K\ge1.
				\end{equation}
				Substituting \eqref{transpmatrix} into \eqref{ineq85} yields:
				\begin{subequations}\label{ineq88}
					\begin{align}				&\widetilde{\varDelta_K\mathbf{V}(\lambda)}\ge\widetilde{\mathbf{P}(K)}\widetilde{\varDelta_{K-1}\mathbf{V}(\lambda)}+(\tilde{U}_{K}(\lambda)-\tilde{U}_{K+1}(\lambda))\mathbf{e},\\
						&\widetilde{\varDelta_K\mathbf{V}(\lambda)}\le\widetilde{\mathbf{P}(K-1)}\widetilde{\varDelta_{K-1}\mathbf{V}(\lambda)}+(\tilde{U}_{K}(\lambda)-\tilde{U}_{K+1}(\lambda))\mathbf{e}.
					\end{align}
				\end{subequations}
				By recursively applying these inequalities over $L$ iterations, we derive the following bounds:
				\begin{subequations}\label{itKtimes}
					\begin{align}
						\widetilde{\varDelta_K\mathbf{V}(\lambda)}&\ge\prod_{t=K}^{K-L+1}\widetilde{\mathbf{P}(t)}\widetilde{\varDelta_{K-L}\mathbf{V}(\lambda)}+\notag\\&(\tilde{U}_{K+1-L}(\lambda)-\tilde{U}_{K+1}(\lambda))\mathbf{e},\label{111a}\\
						\widetilde{\varDelta_K\mathbf{V}(\lambda)}&\le\prod_{t=K-1}^{K-L}\widetilde{\mathbf{P}(t)}\widetilde{\varDelta_{K-L}\mathbf{V}(\lambda)}+\notag\\
						&(\tilde{U}_{K+1-L}(\lambda)-\tilde{U}_{K+1}(\lambda))\mathbf{e}.\label{111b}
					\end{align}
				\end{subequations}
				Since the transformed MDP is a \textit{unichain} and the transition probability matrix $\widetilde{\mathbf{P}(K)}$ holds \textit{aperiodic} for $\forall K\ge1$ (as verified in the proof sketch), we have that there exists a positive integer $L$, a constant $\epsilon>0$, and a state $\gamma^\star$ such that:
				\begin{subequations}\label{Ldefinition}
					\begin{align}
						\left[\prod_{t=K}^{K-L+1}\widetilde{\mathbf{P}(t)}\right]_{\gamma\times\gamma^\star}\ge\epsilon,\quad\forall\gamma\in\mathcal{S}\times\mathcal{Y}\times\mathcal{A},\\
						\left[\prod_{t=K-1}^{K-L}\widetilde{\mathbf{P}(t)}\right]_{\gamma\times\gamma^\star}\ge\epsilon,\quad\forall\gamma\in\mathcal{S}\times\mathcal{Y}\times\mathcal{A}\label{112b}
					\end{align}
				\end{subequations}
				In the following, we establish the existence of finite limits by analyzing two distinct cases.
				\subsubsection{Case 1: $\gamma^\star=\gamma^{\text{r}}$}
				If $\gamma^\star=\gamma^{\text{r}}$, we can derive the inequality \eqref{ieq92} from \eqref{111b}, as shown at the top of this page,
				\begin{figure*}[t!]
					\begin{subequations}\label{ieq92}
						\begin{align}
							&\tilde{{V}}_{K+1}(\gamma;\lambda)-\tilde{{V}}_{K}(\gamma;\lambda)\notag\\
							&{\le}\tilde{U}_{K+1-L}(\lambda)-\tilde{U}_{K+1}(\lambda)+\sum_{\gamma'}	\left[\prod_{t=K-1}^{K-L}\widetilde{\mathbf{P}(t)}\right]_{\gamma\times\gamma'}\times\left(\tilde{{V}}_{K-L+1}(\gamma';\lambda)-\tilde{{V}}_{K-L}(\gamma';\lambda)\right)\label{113a}\\
							&{=}\tilde{U}_{K+1-L}(\lambda)-\tilde{U}_{K+1}(\lambda)+\sum_{\gamma'\ne\gamma^{\text{r}}}	\left[\prod_{t=K-1}^{K-L}\widetilde{\mathbf{P}(t)}\right]_{\gamma\times\gamma'}\times\left(\tilde{{V}}_{K-L+1}(\gamma';\lambda)-\tilde{{V}}_{K-L}(\gamma';\lambda)\right)\label{113b}\\
							&{\le}\tilde{U}_{K+1-L}(\lambda)-\tilde{U}_{K+1}(\lambda)+\max_{\gamma}\left\{\tilde{{V}}_{K-L+1}(\gamma;\lambda)-\tilde{{V}}_{K-L}(\gamma;\lambda)\right\}\times\sum_{\gamma'\ne\gamma^\star}	\left[\prod_{t=K-1}^{K-L}\widetilde{\mathbf{P}(t)}\right]_{\gamma\times\gamma'}\label{113c}
							\\&{\le}\tilde{U}_{K+1-L}(\lambda)-\tilde{U}_{K+1}(\lambda)+(1-\epsilon)\max_{\gamma}\left\{\tilde{{V}}_{K-L+1}(\gamma;\lambda)-\tilde{{V}}_{K-L}(\gamma;\lambda)\right\}, \quad\forall\gamma\in\mathcal{S}\times\mathcal{Y}\times\mathcal{A}.\label{113d}
						\end{align}
					\end{subequations}
					\hrulefill
				\end{figure*}
				where the transition from \eqref{113a} to \eqref{113b} uses the fact that for all $K$, the following holds:
				\begin{equation}\label{115}
					\tilde{{V}}_{K}(\gamma^{\text{r}};\lambda)=(1-\tau)^K\tilde{{V}}_{0}(\gamma^{\text{r}};\lambda)=0,\quad\forall K\in\mathbb{N};
				\end{equation}
				given the initialization $\tilde{{V}}_{K-1}(\gamma^{\text{r}};\lambda)=0$. Inequality \eqref{113c} uses the uniform bound:
				\begin{equation}
					\begin{aligned}
						&\tilde{V}_{K-L+1}(\gamma;\lambda)-\tilde{V}_{K-L}(\gamma;\lambda)\le\\&\max_{\gamma}\{\tilde{V}_{K-L+1}(\gamma;\lambda)-\tilde{V}_{K-L}(\gamma;\lambda)\},
					\end{aligned}
				\end{equation}
				and inequality \eqref{113d} follows from the definition in \eqref{112b}:
				\begin{equation}
					\sum_{\gamma'\ne\gamma^\star}	\left[\prod_{t=K-1}^{K-L}\widetilde{\mathbf{P}(t)}\right]_{\gamma\times\gamma'}=1-\left[\prod_{t=K-1}^{K-L}\widetilde{\mathbf{P}(t)}\right]_{\gamma\times\gamma^\star}{\le}1-\epsilon,
				\end{equation}
				and the fact that the product term is \textit{non-negative}:
				\begin{equation}
					\begin{aligned}
						&\max_{\gamma}\left\{\tilde{{V}}_{K-L+1}(\gamma;\lambda)-\tilde{{V}}_{K-L}(\gamma;\lambda)\right\}\\&\ge\tilde{{V}}_{K-L+1}(\gamma^{\text{r}};\lambda)-\tilde{{V}}_{K-L}(\gamma^{\text{r}};\lambda)=0.
					\end{aligned}
				\end{equation}
				Since \eqref{ieq92} holds for $\forall \gamma\in\mathcal{S}\times\mathcal{Y}\times\mathcal{A}$, we can bound the maximum increment:
				\begin{equation}\label{ineq97}
					\begin{aligned}
						&\max_{\gamma}\left\{\tilde{{V}}_{K+1}(\gamma;\lambda)-\tilde{{V}}_{K}(\gamma;\lambda)\right\}\\&\le(1-\epsilon)\max_{\gamma}\left\{\tilde{{V}}_{K-L+1}(\gamma;\lambda)-\tilde{{V}}_{K-L}(\gamma;\lambda)\right\}\\&\quad+\tilde{U}_{K+1-L}(\lambda)-\tilde{U}_{K+1}(\lambda).
					\end{aligned}
				\end{equation}
				\begin{figure*}[t!]
					\begin{subequations}\label{ieq98}
						\begin{align}
							&\tilde{{V}}_{K+1}(\gamma;\lambda)-\tilde{{V}}_{K}(\gamma;\lambda)\notag\\
							&{\ge}\tilde{U}_{K+1-L}(\lambda)-\tilde{U}_{K+1}(\lambda)+\sum_{\gamma'}	\left[\prod_{t=K}^{K-L+1}\widetilde{\mathbf{P}(t)}\right]_{\gamma\times\gamma'}\times\left(\tilde{{V}}_{K-L+1}(\gamma';\lambda)-\tilde{{V}}_{K-L}(\gamma';\lambda)\right)\label{120a}\\
							&{=}\tilde{U}_{K+1-L}(\lambda)-\tilde{U}_{K+1}(\lambda)+\sum_{\gamma'\ne\gamma^{\text{r}}}	\left[\prod_{t=K}^{K-L+1}\widetilde{\mathbf{P}(t)}\right]_{\gamma\times\gamma'}\times\left(\tilde{{V}}_{K-L+1}(\gamma';\lambda)-\tilde{{V}}_{K-L}(\gamma';\lambda)\right)\label{120b}\\
							&{\ge}\tilde{U}_{K+1-L}(\lambda)-\tilde{U}_{K+1}(\lambda)+\min_{\gamma}\left\{\tilde{{V}}_{K-L+1}(\gamma;\lambda)-\tilde{{V}}_{K-L}(\gamma;\lambda)\right\}\times\sum_{\gamma'\ne\gamma^\star}	\left[\prod_{t=K}^{K-L+1}\widetilde{\mathbf{P}(t)}\right]_{\gamma\times\gamma'}\label{120cv2}
							\\&{\ge}\tilde{U}_{K+1-L}(\lambda)-\tilde{U}_{K+1}(\lambda)+(1-\epsilon)\min_{\gamma}\left\{\tilde{{V}}_{K-L+1}(\gamma;\lambda)-\tilde{{V}}_{K-L}(\gamma;\lambda)\right\}, \quad\forall\gamma\in\mathcal{S}\times\mathcal{Y}\times\mathcal{A}.\label{120d}
						\end{align}
					\end{subequations}
					\hrulefill
				\end{figure*}	
				In a similar manner, using \eqref{111a}, we obtain the lower bound shown in \eqref{ieq98} at the top of the next page, where inequality \eqref{120a} establishes by  $\gamma^\star=\gamma^{\text{r}}$ and 
				\begin{align}
					&\tilde{{V}}_{K-L+1}(\gamma';\lambda)-\tilde{{V}}_{K-L}(\gamma';\lambda)\ge\notag\\&\min_{\gamma}\left\{\tilde{{V}}_{K-L+1}(\gamma;\lambda)-\tilde{{V}}_{K-L}(\gamma;\lambda)\right\};
				\end{align}
				and \eqref{120b} establishes because of \eqref{112b}: 
				\begin{equation}
					\sum_{\gamma'\ne\gamma^\star}	\left[\prod_{t=K}^{K-L+1}\widetilde{\mathbf{P}(t)}\right]_{\gamma\times\gamma'}=1-\left[\prod_{t=K}^{K-L+1}\widetilde{\mathbf{P}(t)}\right]_{\gamma\times\gamma^\star}\overset{(a)}{\le}1-\epsilon,
				\end{equation}
				and the fact that the product term is \textit{non-positive}:
				\begin{equation}
					\begin{aligned}
						&\min_{\gamma}\left\{\tilde{{V}}_{K-L+1}(\gamma;\lambda)-\tilde{{V}}_{K-L}(\gamma;\lambda)\right\}\\&\le\tilde{{V}}_{K-L+1}(\gamma^{\text{r}};\lambda)-\tilde{{V}}_{K-L}(\gamma^{\text{r}};\lambda)=0.
					\end{aligned}
				\end{equation}
				Since inequality \eqref{ieq98} holds for $\forall \gamma$, we can establish:
				\begin{equation}\label{ineq99}
					\begin{aligned}
						&\min_{\gamma}\left\{\tilde{{V}}_{K+1}(\gamma;\lambda)-\tilde{{V}}_{K}(\gamma;\lambda)\right\}\\&\ge(1-\epsilon)\min_{\gamma}\left\{\tilde{{V}}_{K-L+1}(\gamma;\lambda)-\tilde{{V}}_{K-L}(\gamma;\lambda)\right\}\\&\quad+\tilde{U}_{K+1-L}(\lambda)-\tilde{U}_{K+1}(\lambda).
					\end{aligned}
				\end{equation}
				Subtracting \eqref{ineq99} from \eqref{ineq97} directly yields \eqref{relation100} at the top of the next page.
				\begin{figure*}[]
					\begin{equation}\label{relation100}
						\begin{aligned}
							&\max_{\gamma}\left\{\tilde{{V}}_{K+1}(\gamma;\lambda)-\tilde{{V}}_{K}(\gamma;\lambda)\right\}-\min_{\gamma}\left\{\tilde{{V}}_{K+1}(\gamma;\lambda)-\tilde{{V}}_{K}(\gamma;\lambda)\right\}\\&\le(1-\epsilon)\left(\max_{\gamma}\left\{\tilde{{V}}_{K-L+1}(\gamma;\lambda)-\tilde{{V}}_{K-L}(\gamma;\lambda)\right\}-\min_{\gamma}\left\{\tilde{{V}}_{K-L+1}(\gamma;\lambda)-\tilde{{V}}_{K-L}(\gamma;\lambda)\right\}\right).
						\end{aligned}
					\end{equation}
					\hrulefill
				\end{figure*}
				Iterating \eqref{relation100} yields that for some $M>0$ and all $K\ge1$, we have
				\begin{equation}\label{ieq96}
					\begin{aligned}				&\max_{\gamma}\left\{\tilde{{V}}_{K+1}(\gamma;\lambda)-\tilde{{V}}_{K}(\gamma;\lambda)\right\}\\&-\min_{\gamma}\left\{\tilde{{V}}_{K+1}(\gamma;\lambda)-\tilde{{V}}_{K}(\gamma';\lambda)\right\}\\&\le M(1-\epsilon)^{K/L}.
					\end{aligned}
				\end{equation}
				Therefore, the relative difference between $\tilde{{V}}_{K+1}(\gamma;\lambda)$ and $\tilde{{V}}_{K}(\gamma;\lambda)$ is upper bounded by
				\begin{align}\label{eq97}				&\left|\tilde{{V}}_{K+1}(\gamma;\lambda)-\tilde{{V}}_{K}(\gamma;\lambda)\right|\notag\\&\le\max_{\gamma}\left\{\tilde{{V}}_{K+1}(\gamma;\lambda)-\tilde{{V}}_{K}(\gamma;\lambda)\right\}-\notag\\&\quad\min_{\gamma}\left\{\tilde{{V}}_{K+1}(\gamma';\lambda)-\tilde{{V}}_{K}(\gamma;\lambda)\right\}\notag\\&\le M(1-\epsilon)^{K/L},\forall \gamma\in\mathcal{S}\times\mathcal{Y}\times\mathcal{A}.
				\end{align}
				This indicates that $\{\tilde{{V}}_{K}(\gamma;\lambda)\}_{K\in\mathbb{N}^+}$ forms a \textit{Cauchy sequence}. Specifically, for $\forall T>1$, the following holds:
				\begin{subequations}
					\begin{align}				&|\tilde{{V}}_{K+T}(\gamma;\lambda)-\tilde{{V}}_{K}(\gamma;\lambda)|\notag\\&\le\sum_{t=0}^{T-1}|\tilde{{V}}_{K+t+1}(\gamma;\lambda)-\tilde{{V}}_{K+t}(\gamma;\lambda)|\label{128a}\\
						&\le M\sum_{t=0}^{T-1}(1-\epsilon)^{\frac{K+t}{L}}=\frac{M(1-\epsilon)^{K/L}(1-(1-\epsilon)^{T/L})}{1-(1-\epsilon)^{1/L}},\label{128b}
					\end{align}
				\end{subequations}
				where inequality \eqref{128a} follows from the triangle inequality. 
				
				Letting $T\rightarrow\infty$, \eqref{128b} yields
				\begin{equation}\label{129}
					|\tilde{{V}}_{K}(\gamma;\lambda)-\tilde{{V}}_{\infty}(\gamma;\lambda)|\le \frac{M(1-\epsilon)^{K/L}}{1-(1-\epsilon)^{1/L}},
				\end{equation}
				which confirms that 
				$\tilde{{V}}_{K}(\gamma;\lambda)$ converges to a bounded value $\tilde{{V}}_{\infty}(\gamma;\lambda)$ as $K\to\infty$, given any $\gamma$.
				
				We next prove that $\tilde{U}_K(\lambda)$ also converges to a bounded value as $K\to\infty$. The update rule in \eqref{23a} can be rewritten into a vector form as:
				\begin{equation}\label{eq130v2}
					\begin{aligned}
						&\tilde{U}_{K+1}(\lambda)\\&=\min _{A_i, Z_i}\Big\{g(\gamma^{\text{r}},Z_i,A_i;\lambda)+\tau[\mathbf{P}_{A_i,Z_i}]_{N(\gamma^{\text{r}}),:}\times\tilde{\mathbf{V}}_K(\lambda) \Big\},
					\end{aligned}		
				\end{equation}
				where $[\mathbf{P}_{a,z}]_{N(\gamma^{\text{r}}),:}$ denotes the row vector formed by stacking the transition probabilities $\{p_{\gamma^{\text{r}}\gamma'}(a,z)\}$ for all ${\gamma'\in\mathcal{S}\times\mathcal{Y}\times\mathcal{A}}$, arranged according to the same index as that of $\tilde{\mathbf{V}}_K(\lambda)$. Here, $N(\gamma^{\text{r}})$ denotes the index of the reference state $\gamma^{\text{r}}$. By substituting the optimal control pair $(A^{(K)}(\gamma^{\text{r}}),Z^{(K)}(\gamma^{\text{r}}))$, as defined in \eqref{eq82}, into the right-hand side of \eqref{eq130v2}, we obtain the upper bound of $\tilde{U}_{K+1}$ given in \eqref{131}, which is at the top of the next page.
				\begin{figure*}
					\begin{subequations}\label{131}
						\begin{align}				\tilde{U}_{K+1}(\lambda)&=g(\gamma^{\text{r}},Z^{(K)}(\gamma^{\text{r}}),A^{(K)}(\gamma^{\text{r}});\lambda)+\tau\left[\mathbf{P}_{A^{(K)}(\gamma^{\text{r}}),Z^{(K)}(\gamma^{\text{r}})}\right]_{N(\gamma^{\text{r}}),:}\times\tilde{\mathbf{V}}_K(\lambda)\label{eq131a}\\
							&\le g(\gamma^{\text{r}},Z^{(K-1)}(\gamma^{\text{r}}),A^{(K-1)}(\gamma^{\text{r}});\lambda)+\tau\left[\mathbf{P}_{A^{(K-1)}(\gamma^{\text{r}}),Z^{(K-1)}(\gamma^{\text{r}})}\right]_{N(\gamma^{\text{r}}),:}\times\tilde{\mathbf{V}}_K(\lambda).\label{131b}
						\end{align}
					\end{subequations}
					\hrulefill
				\end{figure*}
				Similarly, the $\tilde{U}_K(\lambda)$ can be upper bounded as shown in \eqref{eq133}, which is at the top of the next page.
				\begin{figure*}[t]
					\begin{subequations}\label{eq133}
						\begin{align}
							\tilde{U}_{K}(\lambda)&=g(\gamma^{\text{r}},Z^{(K-1)}(\gamma^{\text{r}}),A^{(K-1)}(\gamma^{\text{r}});\lambda)+\tau\left[\mathbf{P}_{A^{(K-1)}(\gamma^{\text{r}}),Z^{(K-1)}(\gamma^{\text{r}})}\right]_{N(\gamma^{\text{r}}),:}\times\tilde{\mathbf{V}}_{K-1}(\lambda)\label{133a}\\
							&\le g(\gamma^{\text{r}},Z^{(K)}(\gamma^{\text{r}}),A^{(K)}(\gamma^{\text{r}});\lambda)+\tau\left[\mathbf{P}_{A^{(K)}(\gamma^{\text{r}}),Z^{(K)}(\gamma^{\text{r}})}\right]_{N(\gamma^{\text{r}}),:}\times\tilde{\mathbf{V}}_{K}(\lambda).\label{133b}
						\end{align}
						\hrulefill		
					\end{subequations}
				\end{figure*}
				Subtracting \eqref{133a} from \eqref{131b} yields:
				\begin{equation}\label{eq134v2}
					\begin{aligned}				&\tilde{U}_{K+1}(\lambda)-\tilde{U}_{K}(\lambda)\\&\le\tau\left[\mathbf{P}_{A^{(K-1)}(\gamma^{\text{r}}),Z^{(K-1)}(\gamma^{\text{r}})}\right]_{N(\gamma^{\text{r}}),:}\times\left(\tilde{\mathbf{V}}_{K}(\lambda)-\tilde{\mathbf{V}}_{K-1}(\lambda)\right).
					\end{aligned}
				\end{equation}
				Subtracting \eqref{eq131a} from \eqref{133b} yields:
				\begin{equation}\label{eq135v2}
					\begin{aligned}
						&\tilde{U}_{K+1}(\lambda)-\tilde{U}_{K}(\lambda)\\&\ge\tau\left[\mathbf{P}_{A^{(K)}(\gamma^{\text{r}}),Z^{(K)}(\gamma^{\text{r}})}\right]_{N(\gamma^{\text{r}}),:}\times\left(\tilde{\mathbf{V}}_{K}(\lambda)-\tilde{\mathbf{V}}_{K-1}(\lambda)\right).
					\end{aligned}
				\end{equation}
				Combining the upper bound in \eqref{eq134v2} and the lower bound in \eqref{eq135v2}, and applying the inequality $y\le x\le z\Rightarrow|x|\le\max\left\{|y|,|z|\right\}$, we obtain \eqref{eq136} at the top of the next page.
				\begin{figure*}
					\begin{equation}\label{eq136}
						\begin{aligned}
							&\left|\tilde{U}_{K+1}(\lambda)-\tilde{U}_{K}(\lambda)\right|\le\\&\max\Bigg\{{\left|\tau\left[\mathbf{P}_{A^{(K-1)}(\gamma^{\text{r}}),Z^{(K-1)}(\gamma^{\text{r}})}\right]_{N(\gamma^{\text{r}}),:}\times\left(\tilde{\mathbf{V}}_{K}(\lambda)-\tilde{\mathbf{V}}_{K-1}(\lambda)\right)\right|},\left|\tau\left[\mathbf{P}_{A^{(K)}(\gamma^{\text{r}}),Z^{(K)}(\gamma^{\text{r}})}\right]_{N(\gamma^{\text{r}}),:}\times\left(\tilde{\mathbf{V}}_{K}(\lambda)-\tilde{\mathbf{V}}_{K-1}(\lambda)\right)\right|\Bigg\}.
						\end{aligned}
					\end{equation}
					\hrulefill
				\end{figure*}
				The two terms inside the maximum operator in \eqref{eq136} can each be bounded as follows:
				\begin{subequations}\label{135}
					\begin{align}
						&\left|\tau\left[\mathbf{P}_{A^{(K)}(\gamma^{\text{r}}),Z^{(K)}(\gamma^{\text{r}})}\right]_{N(\gamma^{\text{r}}),:}\times\left(\tilde{\mathbf{V}}_{K}(\lambda)-\tilde{\mathbf{V}}_{K-1}(\lambda)\right)\right|\notag\\&\le\tau\left[\mathbf{P}_{A^{(K)}(\gamma^{\text{r}}),Z^{(K)}(\gamma^{\text{r}})}\right]_{N(\gamma^{\text{r}}),:}\times\left|\tilde{\mathbf{V}}_{K}(\lambda)-\tilde{\mathbf{V}}_{K-1}(\lambda)\right|\label{eq137a}\\&\le\tau M(1-\epsilon)^{\frac{K-1}{L}},\label{eq137b}
					\end{align}
				\end{subequations}
				\begin{subequations}\label{eq138v2}
					\begin{align}
						&\left|\tau\left[\mathbf{P}_{A^{(K-1)}(\gamma^{\text{r}}),Z^{(K-1)}(\gamma^{\text{r}})}\right]_{N(\gamma^{\text{r}}),:}\times\left(\tilde{\mathbf{V}}_{K}(\lambda)-\tilde{\mathbf{V}}_{K-1}(\lambda)\right)\right|\notag\\&\le\tau\left[\mathbf{P}_{A^{(K)}(\gamma^{\text{r}}),Z^{(K)}(\gamma^{\text{r}})}\right]_{N(\gamma^{\text{r}}),:}\times\left|\tilde{\mathbf{V}}_{K}(\lambda)-\tilde{\mathbf{V}}_{K-1}(\lambda)\right|\label{eq138a}\\&\le\tau M(1-\epsilon)^{\frac{K-1}{L}},\label{eq138b}
					\end{align}
				\end{subequations}
				where \eqref{eq137a} and \eqref{eq138a} follow from the \textit{triangle inequality}; \eqref{eq137b} and \eqref{eq138b} use the contraction bound established in \eqref{eq97}, along with the fact that sum of the vector $\tau\left[\mathbf{P}_{A^{(K-1)}(\gamma^{\text{r}}),Z^{(K-1)}(\gamma^{\text{r}})}\right]_{N(\gamma^{\text{r}}),:}$ is $1$. Substituting \eqref{135} and \eqref{eq138v2} into \eqref{eq136}, we conclude that
				\begin{equation}\label{139}
					\left|\tilde{U}_{K+1}(\lambda)-\tilde{U}_{K}(\lambda)\right|\le\tau M(1-\epsilon)^{\frac{K-1}{L}}.
				\end{equation}
				This demonstrates that $\{\tilde{{U}}_{K}(\lambda)\}_{K\in\mathbb{N}^+}$ forms a \textit{Cauchy sequence}. In particular, for $\forall T>1$, we have:
				\begin{subequations}
					\begin{align}			&|\tilde{{U}}_{K+T}(\lambda)-\tilde{{U}}_{K}(\lambda)|\notag\\&\label{140a}\le\sum_{t=0}^{T-1}|\tilde{{U}}_{K+t+1}(\lambda)-\tilde{{U}}_{K+t}(\lambda)|\\
						&\le\tau M\sum_{t=0}^{T-1}(1-\epsilon)^{\frac{K+t-1}{L}}\label{140b}\\&=\frac{\tau M(1-\epsilon)^{(K-1)/L}\left(1-(1-\epsilon)^{T/L}\right)}{1-(1-\epsilon)^{1/L}},\label{140c}
					\end{align}
				\end{subequations}
				where inequality \eqref{140a} follows from the \textit{triangle inequality}, and \eqref{140b} follows directly from the bound in \eqref{139}.
				Taking the limit as $T\rightarrow\infty$, we have
				\begin{equation}\label{141}
					\begin{aligned}
						&\left|\tilde{{U}}_{K}(\lambda)-\tilde{{U}}_{\infty}(\lambda)\right|\\&\le\lim_{T \to \infty}\frac{\tau M(1-\epsilon)^{(K-1)/L}\left(1-(1-\epsilon)^{T/L}\right)}{1-(1-\epsilon)^{1/L}}\\&=\frac{\tau M(1-\epsilon)^{(K-1)/L}}{1-(1-\epsilon)^{1/L}}.
					\end{aligned}
				\end{equation}
				This confirms that $\tilde{{U}}_{K}(\lambda)$ converges to a bounded limiting value, denoted by $\tilde{{U}}_{\infty}(\lambda)$.
				\subsubsection{Case 2:  $\gamma^\star\ne\gamma^{\text{r}}$}
				In \textit{Case 1}, we established that when $\gamma^\star=\gamma^{\text{r}}$, both sequences $\{\tilde{{V}}_{K}(\gamma;\lambda)\}_{K\in\mathbb{N}^+}$ and $\{\tilde{{U}}_{K}(\lambda)\}_{K\in\mathbb{N}^+}$ constitute \textit{Cauchy sequences}, and thus converge to bounded values. Here we extend the result to the more general case where the strict condition $\gamma^\star=\gamma^{\text{r}}$ is relaxed. This generalization introduces significant analytical challenges, as key inequalities, specifically \eqref{113c} and \eqref{120cv2}, no longer hold under the relaxed assumption. To overcome this challenge, we introduce an \textit{auxiliary iteration sequence} in the following.
				\begin{iteration}(Auxiliary Iteration Sequence). For a given $\lambda$ and a parameter  $0<\tau\le1$, the auxiliary iteration sequence iteratively generate sequences $\{\bar{U}_{K}(\lambda)\}^{K\in\mathbb{N}^+}$ and $\{\bar{V}_{K}(\gamma;\lambda)\}_{\gamma\in\mathcal{S}\times\mathcal{Y}\times\mathcal{A}}^{K\in\mathbb{N}^+}$ with a starting initial value $\{\bar{V}_0({\gamma;\lambda})\}_{\gamma\in\mathcal{S}\times\mathcal{Y}\times\mathcal{A}}$.
					\begin{subequations}\label{MRVIv2}
						\begin{align}
							\bar{U}_{K+1}(\lambda)&=\min _{A_i, Z_i}\Big\{g(\gamma^{\star},Z_i,A_i;\lambda)+\notag\\&\quad\quad\quad\quad\quad\tau \mathbb{E}[\bar{V}_{K}\left(\gamma';\lambda\right)|\gamma^{\star},Z_i,A_i]\Big\},\\
							\bar{V}_{K+1}(\gamma;\lambda)&=(1-\tau)\bar{V}_{K}(\gamma;\lambda)+\min _{A_i, Z_i}\Big\{g(\gamma,Z_i,A_i;\lambda)\notag\\&+\tau \mathbb{E}[\bar{V}_{K}\left(\gamma';\lambda\right)|\gamma,Z_i,A_i]\Big\}-\bar{U}_{K+1}(\lambda),\notag\\& \quad\quad\quad\quad\quad\quad\quad\forall \gamma\in\mathcal{S}\times\mathcal{Y}\times\mathcal{A},
						\end{align}
					\end{subequations}
					where the initial condition satisfies that $\bar{V}_0(\gamma;\lambda)=\tilde{V}_0(\gamma;\lambda)$ for $\forall \gamma\in\mathcal{S}\times\mathcal{Y}\times\mathcal{A}$.
				\end{iteration}
				A key property of the generated \textit{auxiliary iteration sequence} is that there exists an $M>0$ such that for all $K\ge1$,
				\begin{equation}\label{ieq96v2}
					\begin{aligned}				&\max_{\gamma}\left\{\bar{{V}}_{K+1}(\gamma;\lambda)-\bar{{V}}_{K}(\gamma;\lambda)\right\}\\&-\min_{\gamma}\left\{\bar{{V}}_{K+1}(\gamma;\lambda)-\bar{{V}}_{K}(\gamma';\lambda)\right\}\\&\le M(1-\epsilon)^{K/L},
					\end{aligned}
				\end{equation}
				where the proof follows an analogous approach to that of inequality \eqref{ieq96} in \textit{Case 1} and thus we omit the details here. This demonstrates that the sequences $\{\bar{{V}}_{K}(\gamma;\lambda)\}_{K\in\mathbb{N}^+}$ also forms a \textit{Cauchy sequence}. We next describe the relationship between $\{\tilde{{V}}_{K}(\gamma;\lambda)\}_{K\in\mathbb{N}^+}$ and $\{\bar{{V}}_{K}(\gamma;\lambda)\}_{K\in\mathbb{N}^+}$. Compare \eqref{MRVIv2} with \eqref{MRVI}, the relationship between $\bar{V}_{K}{(\gamma;\lambda)}$ and $\tilde{V}_{K}{(\gamma;\lambda)}$ can be established by:
				\begin{equation}\label{eq112}
					\begin{aligned}
						\tilde{V}_{K}{(\gamma;\lambda)}=\bar{V}_{K}{(\gamma;\lambda)}+\Psi\left(\bar{\mathbf{V}}_{K-1}(\lambda)\right),
					\end{aligned}
				\end{equation}
				where $\bar{\mathbf{V}}_{K}(\lambda)$ is a row vector consisted of $\bar{{V}}_{K}(\gamma;\lambda)$ for all ${\gamma\in\mathcal{S}\times\mathcal{Y}\times\mathcal{A}}$, arranged by the same index scheme as that of $\tilde{V}_K(\gamma)$. The function $\Psi\left(\bar{\mathbf{V}}_{K-1}(\lambda)\right)$ is defined as
				\begin{align}
					&\Psi\left(\bar{\mathbf{V}}_{K-1}(\lambda)\right)\triangleq\notag\\&\min _{A_i, Z_i}\big\{g(\gamma^{\star},Z_i,A_i;\lambda)+\tau \mathbb{E}[\bar{V}_{K-1}\left(\gamma';\lambda\right)|\gamma^{\star},Z_i,A_i]\big\}\nonumber\\&-\min _{A_i, Z_i}\big\{g(\gamma^{\text{r}},Z_i,A_i;\lambda)+\tau \mathbb{E}[\bar{V}_{K-1}\left(\gamma';\lambda\right)|{\gamma^{\text{r}}},Z_i,A_i]\big\}.
				\end{align}
				Similarly, the relationship between $\bar{V}_{K+1}{(\gamma;\lambda)}$ and $\tilde{V}_{K+1}{(\gamma;\lambda)}$ can be established by
				\begin{equation}\label{eq114}
					\begin{aligned}
						\tilde{V}_{K+1}{(\gamma;\lambda)}=\bar{V}_{K+1}{(\gamma;\lambda)}+\Psi(\bar{\mathbf{V}}_{K}(\lambda))
					\end{aligned}
				\end{equation}
				Subtracting \eqref{eq112} from \eqref{eq114} yields:
				\begin{equation}\label{eq147}
					\begin{aligned}
						&\tilde{V}_{K+1}{(\gamma;\lambda)}-\tilde{V}_{K}{(\gamma;\lambda)}\\&=\bar{V}_{K+1}{(\gamma;\lambda)}-\bar{V}_{K}{(\gamma;\lambda)}+\Psi(\bar{\mathbf{V}}_{K}(\lambda))-\Psi(\bar{\mathbf{V}}_{K-1}(\lambda)),
					\end{aligned}		
				\end{equation}
				Thus, by applying the max and min operators to both sides of \eqref{eq147}, we obtain the following:
				\begin{subequations}
					\begin{align}
						&\max_{\gamma}\{\tilde{V}_{K+1}{(\gamma;\lambda)}-\tilde{V}_{K}{(\gamma;\lambda)}\}=\notag\\&\max_{\gamma}\{\bar{V}_{K+1}{(\gamma;\lambda)}-\bar{V}_{K}{(\gamma;\lambda)}\}+\Psi(\bar{\mathbf{V}}_{K};\lambda)-\Psi(\bar{\mathbf{V}}_{K-1}(\lambda)),\label{eq115a}\\
						&\min_{\gamma}\{\tilde{V}_{K+1}{(\gamma;\lambda)}-\tilde{V}_{K}{(\gamma;\lambda)}\}=\notag\\&\min_{\gamma}\{\bar{V}_{K+1}{(\gamma;\lambda)}-\bar{V}_{K}{(\gamma;\lambda)}\}+\Psi(\bar{\mathbf{V}}_{K}(\lambda))-\Psi(\bar{\mathbf{V}}_{K-1}(\lambda)).\label{eq115b}
					\end{align}
				\end{subequations}
				By subtracting equation \eqref{eq115b} from equation \eqref{eq115a}, we obtain the key identity presented in \eqref{149}, which is at the top of this page.
				\begin{figure*}
					\begin{equation}\label{149}
						\begin{aligned}				&\max_{\gamma}\left\{\tilde{V}_{K+1}{(\gamma;\lambda)}-\tilde{V}_{K}{(\gamma;\lambda)}\right\}-\min_{\gamma}\left\{\tilde{V}_{K+1}{(\gamma;\lambda)}-\tilde{V}_{K}{(\gamma;\lambda)}\right\}\\&=\max_{\gamma}\left\{\bar{V}_{K+1}{(\gamma;\lambda)}-\bar{V}_{K}{(\gamma;\lambda)}\right\}-\min_{\gamma}\left\{\bar{V}_{K+1}{(\gamma;\lambda)}-\bar{V}_{K}{(\gamma;\lambda)}\right\}.
						\end{aligned}
					\end{equation}
					\hrulefill
				\end{figure*}
				Combining \eqref{149} with \eqref{ieq96v2}, it follows that the sequence ${\tilde{V}_K(\gamma;\lambda)}$ satisfies:
				\begin{equation}\label{150}
					\begin{aligned}				&\max_{\gamma}\left\{\tilde{{V}}_{K+1}(\gamma;\lambda)-\tilde{{V}}_{K}(\gamma;\lambda)\right\}\\&-\min_{\gamma}\left\{\tilde{{V}}_{K+1}(\gamma;\lambda)-\bar{{V}}_{K}(\gamma';\lambda)\right\}\le M(1-\epsilon)^{K/L}.
					\end{aligned}
				\end{equation}
				Given \eqref{150}, and by applying the same reasoning used in Case 1 (\eqref{eq97}--\eqref{141}), we can show that both ${\tilde{V}_K(\gamma;\lambda)}$ and ${\tilde{U}_K(\lambda)}$ form \textit{Cauchy sequences}. Thus, they also converge to bounded values.
				\hfill$\blacksquare$
				
				\subsection{Convergence Direction}\label{subGb}
				In this subsection, we prove that the convergent bounded values $\tilde{U}_{\infty}(\lambda)$ and $\tilde{V}_{\infty}(\gamma;\lambda)$ constitute a solution to the ACOE \eqref{eq21}. Given the convergence of $\tilde{U}_{K}(\lambda)$ and $\tilde{V}_{K}(\gamma;\lambda)$, we can take $K\to\infty$ into the $\tau$-RVI \eqref{MRVI}, and establish that:
				\begin{equation}\label{eq111}
					\begin{aligned}
						\tilde{V}_{\infty}(\gamma;\lambda)+\tilde{U}_{\infty}(\lambda)=\min _{A_i, Z_i}\big\{g(\gamma,Z_i,A_i;\lambda)+\\ \sum_{\gamma'}\widetilde{p_{\gamma \gamma'}}(Z_i,A_i)\tilde{V}_{\infty}(\gamma';\lambda)\big\}, \forall \gamma\in\mathcal{S}\times\mathcal{Y}\times\mathcal{A}.
					\end{aligned}
				\end{equation}	 
				Compare \eqref{eq111} with \eqref{eq26} and we have that 
				\begin{equation}\label{eq113}
					\begin{aligned}
						\tilde{U}_{\infty}(\lambda)&=\tilde{U}(\lambda),\\\tilde{V}_{\infty}(\gamma;\lambda)&=\tilde{V}^\star(\gamma;\lambda),\forall {\gamma\in\mathcal{S}\times\mathcal{Y}\times\mathcal{A}}.
					\end{aligned}
				\end{equation}	
				Having established the equivalence between $\tilde{U}_{\infty}(\lambda)$ and $\tilde{U}(\lambda)$, as well as between $\tilde{V}_{\infty}(\gamma;\lambda)$ and $\tilde{V}^\star(\gamma;\lambda)$, we now proceed to derive the relationships between $\tilde{U}(\lambda)$ and ${U}(\lambda)$, as well as between $\tilde{V}_{\infty}(\gamma;\lambda)$ and ${V}^\star(\gamma;\lambda)$.
				By substituting \eqref{ptransform} into \eqref{eq26}, we obtain:
				\begin{equation}\begin{aligned}
						\tilde{V}^\star(\gamma;\lambda)+\tilde{U}(\lambda)=\min _{A_i, Z_i}\Big\{g(\gamma,Z_i,A_i;\lambda)+\\ \sum_{\gamma'}\tau{p_{\gamma \gamma'}}(Z_i,A_i)\tilde{V}^\star(\gamma';\lambda)+(1-\tau)\tilde{V}^\star(\gamma;\lambda)\Big\},\\ \forall \gamma\in\mathcal{S}\times\mathcal{Y}\times\mathcal{A}.
					\end{aligned}		
				\end{equation}
				This expression can be reformulated into a more concise form:
				\begin{equation}\label{newform}
					\begin{aligned}
						&\tau\tilde{V}^\star(\gamma;\lambda)+\tilde{U}(\lambda)=\min _{A_i, Z_i}\Big\{g(\gamma,Z_i,A_i;\lambda)+\\& \mathbb{E}\left[\tau\tilde{V}^\star(\gamma';\lambda)|Z_i,A_i\right]\Big\},\text{for } \forall \gamma\in\mathcal{S}\times\mathcal{Y}\times\mathcal{A}.	
					\end{aligned}	
				\end{equation}
				Comparing \eqref{newform} with the ACOE \eqref{eq21}, we observe that $\tilde{U}(\lambda)$ and $\tau\tilde{V}^\star(\gamma;\lambda),\gamma\in\mathcal{S}\times\mathcal{Y}\times\mathcal{A}$ are solutions to \eqref{eq21}. This leads to the following relationship:
				\begin{equation}\label{eq116}
					\begin{aligned}
						\tilde{U}(\lambda)&={U}(\lambda),\\\tau\tilde{V}^\star(\gamma;\lambda)&={V}^\star(\gamma;\lambda),{\gamma\in\mathcal{S}\times\mathcal{Y}\times\mathcal{A}}.
					\end{aligned}
				\end{equation}
				\textcolor{black}{By substituting \eqref{eq116} into \eqref{eq113} and invoking \eqref{141}, we obtain}
				\begin{equation}\color{black}
					e_{\mathrm{U}}^{(K)}(\lambda)\le\frac{\tau M(1-\epsilon)^{(K-1)/L}}{1-(1-\epsilon)^{1/L}},
				\end{equation}
				\textcolor{black}{which completes the proof of Theorem~\ref{theorem3:convergence}.}
				
				\textcolor{black}{Furthermore, letting $K\to\infty$ on both sides of \eqref{129} and \eqref{141}, we have}
				\begin{equation}
					\begin{aligned}
						0\le\lim_{K \to \infty}e_{\mathrm{U}}^{(K)}(\lambda)\le\lim_{K \to \infty}\frac{\tau M(1-\epsilon)^{(K-1)/L}}{1-(1-\epsilon)^{1/L}}=0,\\
						0\le\lim_{K \to \infty}e_{\mathrm{V}}^{(K)}(\gamma;\tau,\lambda)\le\lim_{K \to \infty}\frac{M(1-\epsilon)^{K/L}}{1-(1-\epsilon)^{1/L}}=0.
					\end{aligned}
				\end{equation}
				\textcolor{black}{Therefore, by the squeeze theorem, both error terms converge to zero, which completes the proof of Theorem~\ref{convergence1}.}
				\hfill$\blacksquare$

				\section{Proof of Theorem \ref{Lemma6}}\label{appendixJ}
				The proof proceeds in two distinct parts to establish the desired result. First, in \cref{JA}, we demonstrate that both $\lim_{K\to\infty}\rho_K$ and $\lim_{K\to\infty}\widetilde{W}_K(\gamma),\forall \gamma\in\mathcal{S}\times\mathcal{Y}\times\mathcal{A}$ exist and are finite for all $\gamma \in \mathcal{S} \times \mathcal{Y} \times \mathcal{A}$. Then, we explicitly establish that $\lim_{K \to \infty}{\rho_K}=\rho^\star$ and $\lim_{K \to \infty}\widetilde{W}_K(\gamma)=\frac{W^\star(\gamma)}{\kappa\cdot\mathbb{E}[Y_i]}$ in Section \ref{J-B}.
				\subsection{Convergence of \eqref{prop2}}\label{JA}
				
				Denote $Z^{(K)}(\gamma)$, $A^{(K)}(\gamma)$ as the waiting time and controlled action that achieves the minimum in the $K$-th relation given in \eqref{eq126} at the top of this page.
				\begin{figure*}
					\begin{equation}\label{eq126}
						\begin{aligned}
							(A^{(K)}(\gamma),Z^{(K)}(\gamma))\triangleq\mathop{\arg\min}_{A_i,Z_i}\left\{\frac{q(\gamma,Z_i,A_i)-\kappa\widetilde{W}_K(\gamma)\cdotp\mathbb{E}[Y_i]+\kappa\mathbb{E}\left[\widetilde{W}_K(\gamma')|\gamma,Z_i,A_i\right]\cdotp\mathbb{E}[Y_i]}{f(Z_i)}\right\}.
						\end{aligned}
					\end{equation}
					\hrulefill
				\end{figure*}
				Let $\widetilde{\mathbf{W}}_K\in\mathbb{R}^{|\mathcal{S}\times\mathcal{Y}\times\mathcal{A}|}$ as the column vector formed by stacking the values $\widetilde{W}_{K}(\gamma)$ for all ${\gamma\in\mathcal{S}\times\mathcal{Y}\times\mathcal{A}}$, and define $\mathbf{q}_K\in\mathbb{R}^{|\mathcal{S}\times\mathcal{Y}\times\mathcal{A}|}$ as the column vector composed $q(\gamma,Z^{(K)}(\gamma),A^{(K)}(\gamma))$ arranged under the same indexing scheme. Similarly, let $\mathbf{f}_K\in\mathbb{R}^{|\mathcal{S}\times\mathcal{Y}\times\mathcal{A}|}$ be a column vector consisting of  $f(Z^{(K)}(\gamma))$. Let $\mathbf{P}(K)\in\mathbb{R}^{|\mathcal{S}\times\mathcal{Y}\times\mathcal{A}|\times|\mathcal{S}\times\mathcal{Y}\times\mathcal{A}|}$ denote a stochastic matrix where the $(i,j)$-th entry is given by $p_{\gamma_i\gamma_j}(Z^{(K)}(\gamma_i),A^{(K)}(\gamma_i))$. Finally, we denote by $\oslash$ the element-wise \textit{Hadamard division} operator between vectors, defined as:
				\begin{equation}
					{\mathbf{b}}\oslash{\mathbf{a}}\triangleq\left[\frac{b_1}{a_1},\frac{b_2}{a_2},\cdots,\frac{b_n}{a_n}\right]^{\text{T}},
				\end{equation}
				where $\mathbf{b}=[b_1,\cdots,b_n]$ and $\mathbf{a}=[a_1,\cdots,a_n]$. 
				With the above notations, the recursive relation for $\widetilde{W}_K(\gamma)$ in \eqref{prop2} can be compactly expressed in vector form as:
				\begin{equation}\label{eq128}
					\begin{aligned}
						\widetilde{\mathbf{W}}_{K+1}&=\widetilde{\mathbf{W}}_{K}+\mathbf{q}_K\oslash\mathbf{f}_K-\kappa\mathbb{E}[Y_i]\cdotp\widetilde{\mathbf{W}}_K\oslash\mathbf{f}_K\\&+\kappa\mathbb{E}[Y_i]\cdotp\mathbf{P}(K)\widetilde{\mathbf{W}}_K\oslash\mathbf{f}_K-h_{K+1}\mathbf{e}.
					\end{aligned}
				\end{equation}
				Since $A^{(K)}(\gamma)$ and $Z^{(K)}(\gamma)$ are chosen to minimize the right-hand side of \eqref{eq126}, it follows that for all $K, t \ge 0$,
				\begin{equation}\label{eq129}
					\begin{aligned}
						&\mathbf{q}_K\oslash\mathbf{f}_K-\kappa\mathbb{E}[Y_i]\cdotp\widetilde{\mathbf{W}}_K\oslash\mathbf{f}_K+\kappa\mathbb{E}[Y_i]\cdotp\mathbf{P}(K)\widetilde{\mathbf{W}}_K\oslash\mathbf{f}_K\\&\le\mathbf{q}_t\oslash\mathbf{f}_t-\kappa\mathbb{E}[Y_i]\cdotp\widetilde{\mathbf{W}}_K\oslash\mathbf{f}_t+\kappa\mathbb{E}[Y_i]\cdotp\mathbf{P}(t)\widetilde{\mathbf{W}}_K\oslash\mathbf{f}_t.
					\end{aligned}
				\end{equation}
				Applying this inequality in \eqref{eq128} yields the upper bound:
				\begin{equation}\label{eq130}
					\begin{aligned}
						\widetilde{\mathbf{W}}_{K+1}&\le\widetilde{\mathbf{W}}_{K}+\mathbf{q}_{K-1}\oslash\mathbf{f}_{K-1}-\kappa\mathbb{E}[Y_i]\cdotp\widetilde{\mathbf{W}}_K\oslash\mathbf{f}_{K-1}\\&+\kappa\mathbb{E}[Y_i]\cdotp\mathbf{P}(K-1)\widetilde{\mathbf{W}}_{K}\oslash\mathbf{f}_{K-1}-\rho_{K+1}\mathbf{e}.
					\end{aligned}
				\end{equation}
				Similarly, the recursive relationship between $\widetilde{\mathbf{W}}_K$ and $\widetilde{\mathbf{W}}_{K-1}$ can be established leveraging \eqref{eq129}:
				\begin{subequations}\label{127}
					\begin{align}
						\widetilde{\mathbf{W}}_{K}&=\widetilde{\mathbf{W}}_{K-1}+\mathbf{q}_{K-1}\oslash\mathbf{f}_{K-1}-\kappa\mathbb{E}[Y_i]\cdotp\widetilde{\mathbf{W}}_{K-1}\oslash\mathbf{f}_{K-1}\notag\\&\quad+\kappa\mathbb{E}[Y_i]\cdotp\mathbf{P}(K-1)\widetilde{\mathbf{W}}_{K-1}\oslash\mathbf{f}_{K-1}-\rho_{K}\mathbf{e}\\
						&\le\widetilde{\mathbf{W}}_{K-1}+\mathbf{q}_{K}\oslash\mathbf{f}_{K}-\kappa\mathbb{E}[Y_i]\cdotp\widetilde{\mathbf{W}}_{K-1}\oslash\mathbf{f}_{K}\notag\\&\label{eq131}\quad+\kappa\mathbb{E}[Y_i]\cdotp\mathbf{P}(K)\widetilde{\mathbf{W}}_{K-1}\oslash\mathbf{f}_K-\rho_{K}\mathbf{e}.
					\end{align}
				\end{subequations}
				By subtracting \eqref{127} from \eqref{eq130} and defining:
				\begin{equation}
					\widetilde{\varDelta_K\mathbf{W}}\triangleq\widetilde{\mathbf{W}}_{K+1}-\widetilde{\mathbf{W}}_{K},
				\end{equation}
				we obtain the recursive relations in \eqref{178} at the top of the next page.
				\begin{figure*}
					\begin{subequations}\label{178}
						\begin{align}
							\widetilde{\varDelta_K\mathbf{W}}&\le\widetilde{\varDelta_{K-1}\mathbf{W}}-\kappa\mathbb{E}[Y_i]\cdotp\widetilde{\varDelta_{K-1}\mathbf{W}}\oslash\mathbf{f}_{K-1}+\kappa\mathbb{E}[Y_i]\cdotp\mathbf{P}(K-1)\widetilde{\varDelta_{K-1}\mathbf{W}}\oslash\mathbf{f}_{K-1}+(\rho_K-\rho_{K+1})\mathbf{e}\nonumber\\
							&=\left(\mathbf{I}-\text{diag}\left(\frac{\kappa\mathbb{E}[Y_i]}{\mathbf{f}_{K-1}}\right)+\mathbf{P}(K-1)\text{diag}\left(\frac{\kappa\mathbb{E}[Y_i]}{\mathbf{f}_{K-1}}\right)\right)\widetilde{\varDelta_{K-1}\mathbf{W}}+(\rho_K-\rho_{K+1})\mathbf{e},\label{131a}\\
							\widetilde{\varDelta_K\mathbf{W}}&\ge\widetilde{\varDelta_{K-1}\mathbf{W}}-\kappa\mathbb{E}[Y_i]\cdotp\widetilde{\varDelta_{K-1}\mathbf{W}}\oslash\mathbf{f}_K+\kappa\mathbb{E}[Y_i]\cdotp\mathbf{P}(K)\widetilde{\varDelta_{K-1}\mathbf{W}}\oslash\mathbf{f}_K+(\rho_K-\rho_{K+1})\mathbf{e}\nonumber\\
							&=\left(\mathbf{I}-\text{diag}\left(\frac{\kappa\mathbb{E}[Y_i]}{\mathbf{f}_K}\right)+\mathbf{P}(K)\text{diag}\left(\frac{\kappa\mathbb{E}[Y_i]}{\mathbf{f}_K}\right)\right)\widetilde{\varDelta_{K-1}\mathbf{W}}+(\rho_K-\rho_{K+1})\mathbf{e}.\label{eq131b}
						\end{align}
					\end{subequations}
					\hrulefill
				\end{figure*}
				For short-hand notations, we define ${\widetilde{\mathbf{P}_{\kappa}(K)}}$ as: 
				\begin{equation}\label{definition2}
					{\widetilde{\mathbf{P}_{\kappa}(K)}}\triangleq\mathbf{I}-\text{diag}\left(\frac{\kappa\mathbb{E}[Y_i]}{\mathbf{f}_K}\right)+\mathbf{P}(K)\text{diag}\left(\frac{\kappa\mathbb{E}[Y_i]}{\mathbf{f}_K}\right).
				\end{equation}
				In the following lemma, we show that the introduced matrix $\widetilde{\mathbf{P}}_{\kappa}(K)$ is \textit{aperiodic} and \textit{stochastic}.
				\begin{Lemma}\label{Lemma9}
					If $0<\kappa<1$ and $\mathbf{P}(K)$ forms a unichain, the matrix ${\widetilde{\mathbf{P}_{\kappa}(K)}}$ is an aperiodic unichain stochastic matrix, \emph{i.e.}, for any $0<\epsilon<1$, there exists a positive integer $L$ and a state $\gamma^\star$ satisfying:
					\begin{subequations}\label{181}
						\begin{align}
							&\left[\prod_{t=K}^{K-L+1}\widetilde{\mathbf{P}_{\kappa}(t)}\right]_{\gamma\times\gamma^\star}\ge\epsilon,\quad\forall\gamma\in\mathcal{S}\times\mathcal{Y}\times\mathcal{A},\label{180a}\\
							&\left[\prod_{t=K-1}^{K-L}\widetilde{\mathbf{P}_{\kappa}(t)}\right]_{\gamma\times\gamma^\star}\ge\epsilon,\quad\forall\gamma\in\mathcal{S}\times\mathcal{Y}\times\mathcal{A}.\label{180b}
						\end{align}
					\end{subequations}  		
				\end{Lemma}
				\begin{proof}
					See Appendix \ref{appendixK}.
				\end{proof}
				By substituting \eqref{definition2} into \eqref{131a} and \eqref{eq131b}, we obtain
				\begin{subequations}
					\begin{align}
						\widetilde{\varDelta_K\mathbf{W}}\ge{\widetilde{\mathbf{P}_{\kappa}(K-1)}}\widetilde{\varDelta_{K-1}\mathbf{W}}+(\rho_K-\rho_{K+1})\mathbf{e},\label{181a}\\
						\widetilde{\varDelta_K\mathbf{W}}\le{\widetilde{\mathbf{P}_{\kappa}(K)}}\widetilde{\varDelta_{K-1}\mathbf{W}}+(\rho_K-\rho_{K+1})\mathbf{e}.\label{181b}
					\end{align}
				\end{subequations}
				Upon iterating \eqref{181a} and \eqref{181b} for $L$ successive steps, we obtain the following lower and upper bounds:
				\begin{subequations}\label{eq135}
					\begin{align}
						&\widetilde{\varDelta_K\mathbf{W}}\ge\prod_{t=K}^{K-L+1}\widetilde{\mathbf{P}_{\kappa}(t)}\widetilde{\varDelta_{K-L}\mathbf{W}}+(\rho_{K+1-L}-\rho_{K+1})\mathbf{e},\label{182a}\\&\widetilde{\varDelta_K\mathbf{W}}\le\prod_{t=K-1}^{K-L}\widetilde{\mathbf{P}_{\kappa}(t)}\widetilde{\varDelta_{K-L}\mathbf{W}}+(\rho_{K+1-L}-\rho_{K+1})\mathbf{e}\label{182b}.
					\end{align}
				\end{subequations}
				
				In the following, we establish the existence of finite limits by analyzing two distinct cases.
				\subsubsection{Case 1: $\gamma^\star=\gamma^{\text{r}}$}
				If $\gamma^\star=\gamma^{\text{r}}$, by reformulating \eqref{182b} in its scalar form, we can obtain the inequality \eqref{183}.
				\begin{figure*}[t!]
					\begin{subequations}\label{183}
						\begin{align}
							&\widetilde{{W}}_{K+1}(\gamma)-\widetilde{{W}}_{K}(\gamma)\notag\\
							&{\le}\rho_{K+1-L}-\rho_{K+1}+\sum_{\gamma'}	\left[\prod_{t=K-1}^{K-L}\widetilde{\mathbf{P}(t)}\right]_{\gamma\times\gamma'}\times\left(\widetilde{{W}}_{K-L+1}(\gamma')-\widetilde{{W}}_{K-L}(\gamma')\right)\label{183a}\\
							&{=}{h}_{K+1-L}-{h}_{K+1}+\sum_{\gamma'\ne\gamma^{\text{r}}}	\left[\prod_{t=K-1}^{K-L}\widetilde{\mathbf{P}(t)}\right]_{\gamma\times\gamma'}\times\left(\widetilde{{W}}_{K-L+1}(\gamma')-\widetilde{{W}}_{K-L}(\gamma')\right)\label{183b}\\
							&{\le}{h}_{K+1-L}-{h}_{K+1}+\max_{\gamma}\left\{\widetilde{{W}}_{K-L+1}(\gamma)-\widetilde{{W}}_{K-L}(\gamma)\right\}\times\sum_{\gamma'\ne\gamma^\star}	\left[\prod_{t=K-1}^{K-L}\widetilde{\mathbf{P}(t)}\right]_{\gamma\times\gamma'}\label{183c}
							\\&{\le}{h}_{K+1-L}-{h}_{K+1}+(1-\epsilon)\max_{\gamma}\left\{\widetilde{{W}}_{K-L+1}(\gamma)-\widetilde{{W}}_{K-L}(\gamma)\right\}, \quad\forall\gamma\in\mathcal{S}\times\mathcal{Y}\times\mathcal{A}.\label{183d}
						\end{align}
					\end{subequations}
					\hrulefill
				\end{figure*} 
				where \eqref{183b} follows from substituting $\gamma=\gamma^{\text{r}}$ into \eqref{38b}, which yields: 
				\begin{equation}\label{184}
					\widetilde{W}_{K+1}(\gamma^{\text{r}})=\widetilde{W}_{K}(\gamma^{\text{r}}).
				\end{equation}	
				Iterating \eqref{184} yields
				\begin{equation}\label{185}
					\widetilde{W}_{K}(\gamma^{\text{r}})=\widetilde{W}_{0}(\gamma^{\text{r}})=0,\forall K\ge1,
				\end{equation}
				where \eqref{185} holds due to the initialization $\widetilde{W}_0(\gamma^{\text{r}}) = 0$ in \textsc{OnePDSI}. Inequality \eqref{183c} is derived by noting that $\gamma^\star = \gamma^{\text{r}}$ and applying the following bound: \begin{equation}
					\begin{aligned}
						&\widetilde{{W}}_{K-L+1}(\gamma)-\widetilde{{W}}_{K-L}(\gamma)\le\\&\max_{\gamma}\left\{\widetilde{{W}}_{K-L+1}(\gamma)-\tilde{{W}}_{K-L}(\gamma)\right\};
					\end{aligned}
				\end{equation} 
				Inequality \eqref{183d} is established by observing that: 
				\begin{equation}
					\begin{aligned}
						&\max_{\gamma}\left\{\widetilde{{W}}_{K-L+1}(\gamma)-\tilde{{W}}_{K-L}(\gamma)\right\}\\&\ge\widetilde{{W}}_{K-L+1}(\gamma^\text{r})-\tilde{{W}}_{K-L}(\gamma^{\text{r}})=0,
					\end{aligned}
				\end{equation}
				and a reformulation of \eqref{180b}:
				\begin{equation}
					\begin{aligned}
						\sum_{\gamma'\ne\gamma^\star}	\left[\prod_{t=K-1}^{K-L}\widetilde{\mathbf{P}_\kappa(t)}\right]_{\gamma\times\gamma'}&=1-\left[\prod_{t=K-1}^{K-L}\widetilde{\mathbf{P}_\kappa(t)}\right]_{\gamma\times\gamma^\star}\\&{\le}1-\epsilon,
					\end{aligned}
				\end{equation}
				Since inequality \eqref{183} holds for $\forall \gamma\in\mathcal{S}\times\mathcal{Y}\times\mathcal{A}$, we can rewrite it as
				\begin{equation}\label{189}
					\begin{aligned}
						&\max_{\gamma}\left\{\widetilde{{W}}_{K+1}(\gamma)-\widetilde{{W}}_{K}(\gamma)\right\}\\&\le(1-\epsilon)\max_{\gamma}\left\{\widetilde{{W}}_{K-L+1}(\gamma)-\widetilde{{W}}_{K-L}(\gamma)\right\}\\&\quad+\rho_{K+1-L}-\rho_{K+1}.
					\end{aligned}
				\end{equation}
				\begin{figure*}[t!]
					\begin{subequations}\label{190}
						\begin{align}
							&\widetilde{{W}}_{K+1}(\gamma)-\widetilde{{W}}_{K}(\gamma)\notag\\
							&{\ge}\rho_{K+1-L}-\rho_{K+1}+\sum_{\gamma'}	\left[\prod_{t=K}^{K-L+1}\widetilde{\mathbf{P}_\kappa(t)}\right]_{\gamma\times\gamma'}\times\left(\widetilde{{W}}_{K-L+1}(\gamma')-\widetilde{{W}}_{K-L}(\gamma')\right)\label{190a}\\
							&{=}\rho_{K+1-L}-\rho_{K+1}+\sum_{\gamma'\ne\gamma^{\text{r}}}	\left[\prod_{t=K}^{K-L+1}\widetilde{\mathbf{P}_\kappa(t)}\right]_{\gamma\times\gamma'}\times\left(\widetilde{{W}}_{K-L+1}(\gamma')-\widetilde{{W}}_{K-L}(\gamma')\right)\label{190b}\\
							&{\ge}\rho_{K+1-L}-\rho_{K+1}+\min_{\gamma}\left\{{{}}_{K-L+1}(\gamma;\lambda)-\widetilde{{W}}_{K-L}(\gamma;\lambda)\right\}\times\sum_{\gamma'\ne\gamma^\star}	\left[\prod_{t=K}^{K-L+1}\widetilde{\mathbf{P}(t)}\right]_{\gamma\times\gamma'}\label{190c}
							\\&{\ge}\rho_{K+1-L}-\rho_{K+1}+(1-\epsilon)\min_{\gamma}\left\{\widetilde{{W}}_{K-L+1}(\gamma)-\rho_{K-L}(\gamma;\lambda)\right\}, \quad\forall\gamma\in\mathcal{S}\times\mathcal{Y}\times\mathcal{A}.\label{190d}
						\end{align}
					\end{subequations}
					\hrulefill
				\end{figure*}	
				Similar to the derivation of \eqref{183}, we can obtain inequality \eqref{190} from \eqref{181a}, as shown at the top of the next page. The derivation proceeds through several key steps: First, inequality \eqref{190a} is established by expanding the vector form of \eqref{181a} into its scalar representation. Then, \eqref{190b} follows directly from \eqref{185}. To establish \eqref{190c}, we use the condition $\gamma^\star = \gamma^{\text{r}}$ along with the inequality: \begin{equation}
					\begin{aligned}
						&\widetilde{{W}}_{K-L+1}(\gamma')-\widetilde{{W}}_{K-L}(\gamma')\ge\\&\min_{\gamma}\left\{\widetilde{{W}}_{K-L+1}(\gamma)-\widetilde{{W}}_{K-L}(\gamma)\right\};
					\end{aligned}
				\end{equation} Inequality \eqref{190d} follows from two key observations. First,
				\begin{equation}
					\begin{aligned}
						&\min_{\gamma}\left\{\widetilde{{W}}_{K-L+1}(\gamma)-\tilde{{W}}_{K-L}(\gamma)\right\}\\&\le\widetilde{{W}}_{K-L+1}(\gamma^\text{r})-\tilde{{W}}_{K-L}(\gamma^{\text{r}})=0.
					\end{aligned}
				\end{equation}
				And second:
				\begin{equation}
					\sum_{\gamma'\ne\gamma^\star}	\left[\prod_{t=K}^{K-L+1}\widetilde{\mathbf{P}(t)}\right]_{\gamma\times\gamma'}=1-\left[\prod_{t=K}^{K-L+1}\widetilde{\mathbf{P}(t)}\right]_{\gamma\times\gamma^\star}\overset{(a)}{\le}1-\epsilon,
				\end{equation}
				where the inequality follows from \eqref{180a}. Given that inequality \eqref{190} holds for all $\gamma\in\mathcal{S}\times\mathcal{Y}\times\mathcal{A}$, we can rewrite it in terms of the minimal difference across states:
				\begin{equation}\label{194}
					\begin{aligned}
						&\min_{\gamma}\left\{\widetilde{{W}}_{K+1}(\gamma)-\widetilde{{W}}_{K}(\gamma)\right\}\\&\ge(1-\epsilon)\min_{\gamma}\left\{\widetilde{{W}}_{K-L+1}(\gamma)-\widetilde{{W}}_{K-L}(\gamma)\right\}\\&\quad+\rho_{K+1-L}-\rho_{K+1}.
					\end{aligned}
				\end{equation}
				Finally, by subtracting \eqref{194} from \eqref{189}, we have \eqref{195} at the top of the next page. 
				\begin{figure*}
					\begin{equation}\label{195}
						\begin{aligned}
							&\max_{\gamma}\left\{\widetilde{{W}}_{K+1}(\gamma)-\widetilde{{W}}_{K}(\gamma)\right\}-\min_{\gamma}\left\{\widetilde{{W}}_{K+1}(\gamma)-\widetilde{{W}}_{K}(\gamma)\right\}\\\le&(1-\epsilon)\left(\max_{\gamma}\left\{\widetilde{{W}}_{K-L+1}(\gamma)-\widetilde{{W}}_{K-L}(\gamma)\right\}-\min_{\gamma}\left\{\widetilde{{W}}_{K-L+1}(\gamma)-\widetilde{{W}}_{K-L}(\gamma)\right\}\right).
						\end{aligned}
					\end{equation}
					\hrulefill
				\end{figure*}
				By iterating \eqref{195}, we can conclude that for some constant $M$ and any $K\ge1$:
				\begin{equation}\label{196}
					\begin{aligned}
						&\max_{\gamma}\Big\{\widetilde{{W}}_{K+1}(\gamma)-\widetilde{{W}}_{K}(\gamma)\Big\}-\min_{\gamma'}\Big\{\widetilde{{W}}_{K+1}(\gamma)-\widetilde{{W}}_{K}(\gamma)\Big\}\\&\le M(1-\epsilon)^{K/L}.
					\end{aligned}
				\end{equation}
				This result allows us to establish an upper bound on the relative difference between  $\widetilde{{W}}_{K+1}(\gamma)$ and $\widetilde{{W}}_{K}(\gamma)$:
				\begin{subequations}\label{eq138}
					\begin{align}					&\left|\widetilde{{W}}_{K+1}(\gamma)-\widetilde{{W}}_{K}(\gamma)\right|\le\notag\\\label{197a}&\max_{\gamma}\Big\{\widetilde{{W}}_{K+1}(\gamma)-\widetilde{{W}}_{K}(\gamma)\Big\}-\min_{\gamma}\Big\{\widetilde{W}_{K+1}(\gamma)-\widetilde{{W}}_{K}(\gamma)\Big\}\\&\label{197b}\le M(1-\epsilon)^{K/L},\quad \forall \gamma\in\mathcal{S}\times\mathcal{Y}\times\mathcal{A}.
					\end{align}
				\end{subequations}
				This bound demonstrates that the sequence $\{\widetilde{{W}}_{K}(\gamma)\}_{K\in\mathbb{N}^+}$ constitutes a \textit{Cauchy sequence}. Specifically, $\forall T>1$, the following inequality holds:
				\begin{subequations}\label{eq146}
					\begin{align}				&\left|\widetilde{{W}}_{K+T}(\gamma)-\widetilde{{W}}_{K}(\gamma)\right|\notag\\&\label{198a}\le\sum_{t=0}^{T-1}\left|\widetilde{{W}}_{K+t+1}(\gamma)-\widetilde{{W}}_{K+t}(\gamma)\right|\\&\le M\sum_{t=0}^{T-1}(1-\epsilon)^{\frac{K+t}{L}}=\frac{M(1-\epsilon)^{K/L}(1-(1-\epsilon)^{T/L})}{1-(1-\epsilon)^{1/L}},\label{198b}
					\end{align}
				\end{subequations}
				The inequality \eqref{198a} is derived from the \textit{triangle inequality}, while \eqref{198b} follows directly from \eqref{197b}.
				Taking the limit $T\to\infty$ yields:
				\begin{subequations}\label{eq201}
					\begin{align}
						&\left|\widetilde{{W}}_{K}(\gamma)-\widetilde{{W}}_{\infty}(\gamma)\right|\notag\\&\le\lim_{T \to \infty}\frac{M(1-\epsilon)^{K/L}(1-(1-\epsilon)^{T/L})}{1-(1-\epsilon)^{1/L}}\\&= \frac{M(1-\epsilon)^{K/L}}{1-(1-\epsilon)^{1/L}}.\label{199b}
					\end{align}	
				\end{subequations}
				This final result demonstrates that $\widetilde{{W}}_{K}(\gamma)$ will converge to a bounded value $\widetilde{{W}}_{\infty}(\gamma)$.
				
				We now establish that $\rho_k$ is also a \textit{Cauchy sequence} and consequently converges to a bounded value $\rho_\infty$. Through equation \eqref{38b} and the formal definition of $(A^{(K)}(\gamma),Z^{(K)}(\gamma))$ provided in \eqref{eq126}, we can derive the bound presented in \eqref{ieq141}, which is shown at the top of the next page. The establishment of \eqref{200a} follows from recasting relation \eqref{38b} into its vector form, where $[\mathbf{P}_{a,z}]_{N(\gamma^{\text{r}}),:}$ denotes the row vector consisted of $\{p_{\gamma^{\text{r}}\gamma'}(a,z)\}_{\gamma'\in\mathcal{S}\times\mathcal{Y}\times\mathcal{A}}$ and $N(\gamma^{\text{r}})$ denotes the index of the reference state $\gamma^{\text{r}}$. Subsequently, \eqref{200b} is established through direct application of the definition of $(A^{(K)}(\gamma),Z^{(K)}(\gamma))$ as given in \eqref{eq126}.
				\begin{figure*}
					\begin{subequations}\label{ieq141}
						\begin{align}
							&\rho_{K+1}=\widetilde{W}_K(\gamma^{\text{r}})+\min _{A_i, Z_i}\left\{\frac{q(\gamma^{\text{r}},Z_i,A_i)-\kappa\widetilde{W}_K(\gamma^{\text{r}})\cdotp\mathbb{E}[Y_i]+\kappa\mathbb{E}\left[\widetilde{W}_K(\gamma')|\gamma^{\text{r}},Z_i,A_i\right]\cdotp\mathbb{E}[Y_i]}{f(Z_i)}\right\}\notag\\
							&\label{200a}=\widetilde{W}_K(\gamma^{\text{r}})+\frac{q(\gamma^{\text{r}},Z^{(K)}(\gamma^{\text{r}}),A^{(K)}(\gamma^{\text{r}}))-\kappa\widetilde{W}_K(\gamma^{\text{r}})\cdotp\mathbb{E}[Y_i]+\kappa\mathbb{E}[Y_i]\cdotp\left[\mathbf{P}_{A^{(K)}(\gamma^{\text{r}}),Z^{(K)}(\gamma^{\text{r}})}\right]_{N(\gamma^{\text{r}}),:}\widetilde{\mathbf{W}}_K}{f(Z^{(K)}(\gamma^{\text{r}}))}\\
							&\label{200b}\le \widetilde{W}_K(\gamma^{\text{r}})+\frac{q(\gamma^{\text{r}},Z^{(K-1)}(\gamma^{\text{r}}),A^{(K-1)}(\gamma^{\text{r}}))-\kappa\widetilde{W}_K(\gamma^{\text{r}})\cdotp\mathbb{E}[Y_i]+\kappa\mathbb{E}[Y_i]\cdotp\left[\mathbf{P}_{A^{(K-1)}(\gamma^{\text{r}}),Z^{(K-1)}(\gamma^{\text{r}})}\right]_{N(\gamma^{\text{r}}),:}\widetilde{\mathbf{W}}_K}{f(Z^{(K-1)}(\gamma^{\text{r}}))}.
						\end{align}		
					\end{subequations} 
					\hrulefill
				\end{figure*}
				\begin{figure*}
					\begin{subequations}\label{ieq142}
						\begin{align}
							&\rho_{K}=\widetilde{W}_{K-1}(\gamma^{\text{r}})+\min _{A_i, Z_i}\left\{\frac{q(\gamma^{\text{r}},Z_i,A_i)-\kappa\widetilde{W}_{K-1}(\gamma^{\text{r}})\cdotp\mathbb{E}[Y_i]+\kappa\mathbb{E}\left[\widetilde{W}_{K-1}(\gamma')|\gamma^{\text{r}},Z_i,A_i\right]\cdotp\mathbb{E}[Y_i]}{f(Z_i)}\right\}\notag\\
							&\label{201a}=\widetilde{W}_{K-1}(\gamma^{\text{r}})+\frac{q(\gamma^{\text{r}},Z^{(K-1)}(\gamma^{\text{r}}),A^{(K-1)}(\gamma^{\text{r}}))-\kappa\widetilde{W}_{K-1}(\gamma^{\text{r}})\cdotp\mathbb{E}[Y_i]+\kappa\mathbb{E}[Y_i]\cdotp\left[\mathbf{P}_{A^{(K-1)}(\gamma^{\text{r}}),Z^{(K-1)}(\gamma^{\text{r}})}\right]_{N(\gamma^{\text{r}}),:}\widetilde{\mathbf{W}}_{K-1}}{f(Z^{(K-1)}(\gamma^{\text{r}}))}\\
							&\label{201b}\le \widetilde{W}_{K-1}(\gamma^{\text{r}})+\frac{q(\gamma^{\text{r}},Z^{(K)}(\gamma^{\text{r}}),A^{(K)}(\gamma^{\text{r}}))-\kappa\widetilde{W}_{K-1}(\gamma^{\text{r}})\cdotp\mathbb{E}[Y_i]+\kappa\mathbb{E}[Y_i]\cdotp\left[\mathbf{P}_{A^{(K)}(\gamma^{\text{r}}),Z^{(K)}(\gamma^{\text{r}})}\right]_{N(\gamma^{\text{r}}),:}\widetilde{\mathbf{W}}_{K-1}}{f(Z^{(K-1)}(\gamma^{\text{r}}))}.
						\end{align}	
					\end{subequations}
					\hrulefill
				\end{figure*}
				Iteratively, we can also establish the bound on $\rho_k$ as shown in \eqref{ieq142}, presented at the top of the next page. Subtracting \eqref{201a} from \eqref{200b} yields the upper bound of $\rho_{K+1}-\rho_K$ in \eqref{202}, where \eqref{202b} is established by re-writing the righ-hand side of \eqref{202a} into a vector form and \eqref{202c} holds directly from \eqref{definition2}.
				\begin{figure*}
					\begin{subequations}\label{202}
						\begin{align}
							\rho_{K+1}-\rho_K&\le\frac{\left(f(Z^{K-1}(\gamma^{\text{r}}))-\kappa\cdot\mathbb{E}[Y_i]\right)\cdot\left(\widetilde{{W}}_{K}(\gamma^{\text{r}})-\widetilde{{W}}_{K-1}(\gamma^{\text{r}})\right)+\kappa\mathbb{E}[Y_i]\cdot\left[\mathbf{P}_{A^{(K-1)}(\gamma^{\text{r}}),Z^{(K-1)}(\gamma^{\text{r}})}\right]_{N(\gamma^{\text{r}}),:}\left(\widetilde{\mathbf{W}_K}-\widetilde{\mathbf{W}_{K-1}}\right)}{f(Z^{K-1}(\gamma^{\text{r}}))}\label{202a}\\
							&=\left[\mathbf{I}-\text{diag}\left(\frac{\kappa\mathbb{E}[Y_i]}{\mathbf{f}_{K-1}}\right)+\mathbf{P}(K-1)\text{diag}\left(\frac{\kappa\mathbb{E}[Y_i]}{\mathbf{f}_{K-1}}\right)\right]_{N(\gamma^{\text{r}}),:}\times\left(\widetilde{\mathbf{W}_K}-\widetilde{\mathbf{W}_{K-1}}\right)\label{202b}\\
							&=\left[{\widetilde{\mathbf{P}_{\kappa}(K-1)}}\right]_{N(\gamma^{\text{r}}),:}\left(\widetilde{\mathbf{W}}_{K}-\widetilde{\mathbf{W}}_{K-1}\right).\label{202c}
						\end{align}
					\end{subequations}
					\hrulefill
				\end{figure*}
				Similarly, by subtracting \eqref{201b} from \eqref{200a}, we can establish the lower bound of $\rho_{K+1}-\rho_K$ in \eqref{203} in next page.
				\begin{figure*}
					\begin{subequations}\label{203}
						\begin{align}
							\rho_{K+1}-\rho_K&\ge\frac{\left(f(Z^{K}(\gamma^{\text{r}}))-\kappa\cdot\mathbb{E}[Y_i]\right)\cdot\left(\widetilde{{W}}_{K}(\gamma^{\text{r}})-\widetilde{{W}}_{K}(\gamma^{\text{r}})\right)+\kappa\mathbb{E}[Y_i]\cdot\left[\mathbf{P}_{A^{(K)}(\gamma^{\text{r}}),Z^{(K)}(\gamma^{\text{r}})}\right]_{N(\gamma^{\text{r}}),:}\left(\widetilde{\mathbf{W}_K}-\widetilde{\mathbf{W}_{K-1}}\right)}{f(Z^{K}(\gamma^{\text{r}}))}\\
							&=\left[\mathbf{I}-\text{diag}\left(\frac{\kappa\mathbb{E}[Y_i]}{\mathbf{f}_{K}}\right)+\mathbf{P}(K)\text{diag}\left(\frac{\kappa\mathbb{E}[Y_i]}{\mathbf{f}_{K}}\right)\right]_{N(\gamma^{\text{r}}),:}\times\left(\widetilde{\mathbf{W}_K}-\widetilde{\mathbf{W}_{K-1}}\right)\\
							&=\left[{\widetilde{\mathbf{P}_{\kappa}(K)}}\right]_{N(\gamma^{\text{r}}),:}\left(\widetilde{\mathbf{W}}_{K}-\widetilde{\mathbf{W}}_{K-1}\right).
						\end{align}
					\end{subequations}
					\hrulefill
				\end{figure*}
				Because $a\le b\le c \Rightarrow |b|\le\max\{|a|,|c|\}$, combing \eqref{202} and \eqref{203} yields \eqref{eq144v2} in the next page,
				\begin{figure*}
					\begin{equation}\label{eq144v2}
						\begin{aligned}	
							\left|\rho_{K+1}-\rho_{K}\right|&\le\max\left\{\left|\left[{\widetilde{\mathbf{P}_{\kappa}(K-1)}}\right]_{N(\gamma^{\text{r}}),:}\left(\widetilde{\mathbf{W}}_{K}-\widetilde{\mathbf{W}}_{K-1}\right)\right|,\left|\left[{\widetilde{\mathbf{P}_{\kappa}(K)}}\right]_{N(\gamma^{\text{r}}),:}\left(\widetilde{\mathbf{W}}_{K}-\widetilde{\mathbf{W}}_{K-1}\right)\right|\right\} \\
							&\le\max\left\{\left[{\widetilde{\mathbf{P}_{\kappa}(K-1)}}\right]_{N(\gamma^{\text{r}}),:}\left|\widetilde{\mathbf{W}}_{K}-\widetilde{\mathbf{W}}_{K-1}\right|,\left[{\widetilde{\mathbf{P}_{\kappa}(K)}}\right]_{N(\gamma^{\text{r}}),:}\left|\widetilde{\mathbf{W}}_{K}-\widetilde{\mathbf{W}}_{K-1}\right|\right\}.
						\end{aligned}
					\end{equation}
					\hrulefill
				\end{figure*}
				where the first and second terms in the maximum operator satisfy the following inequalities:	 
				\begin{subequations}\label{205}
					\begin{align}
						&\left[{\widetilde{\mathbf{P}_{\kappa}(K-1)}}\right]_{N(\gamma^{\text{r}}),:}\left|\widetilde{\mathbf{W}}_{K}-\widetilde{\mathbf{W}}_{K-1}\right|\notag\\&\label{205a}\le\left[{\widetilde{\mathbf{P}_{\kappa}(K-1)}}\right]_{N(\gamma^{\text{r}}),:}M(1-\epsilon)^{\frac{K-1}{L}}\mathbf{e}\\&=M(1-\epsilon)^{\frac{K-1}{L}},\label{205b}
					\end{align}
				\end{subequations}
				\begin{subequations}\label{206}
					\begin{align}
						&\left[{\widetilde{\mathbf{P}_{\kappa}(K)}}\right]_{N(\gamma^{\text{r}}),:}\left|\widetilde{\mathbf{W}}_{K}-\widetilde{\mathbf{W}}_{K}\right|\notag\\&\label{206a}\le\left[{\widetilde{\mathbf{P}_{\kappa}(K-1)}}\right]_{N(\gamma^{\text{r}}),:}M(1-\epsilon)^{\frac{K-1}{L}}\mathbf{e}\\&=M(1-\epsilon)^{\frac{K-1}{L}},\label{206b}
					\end{align}
				\end{subequations}
				where \eqref{205a} and \eqref{206a} are established from the vector form of \eqref{eq138}. \eqref{205b} and \eqref{206b} are established as the matrix $\widetilde{\mathbf{P}_{\kappa}(k)}$ is a stochastic matrix for $\forall K\ge1$:
				\begin{align}\label{sum1}
					&\sum_{\gamma'}\left[\widetilde{\mathbf{P}_{\kappa}(k)}\right]_{\gamma\times\gamma'}\notag\\&=\sum_{\gamma'}\left[\mathbf{I}-\text{diag}\left(\frac{\kappa\mathbb{E}[Y_i]}{\mathbf{f}_k}\right)+\mathbf{P}(k)\text{diag}\left(\frac{\kappa\mathbb{E}[Y_i]}{\mathbf{f}_k}\right)\right]_{\gamma\times\gamma'}\notag\\&=1-\frac{\kappa\mathbb{E}[Y_i]}{f\left(Z^{(k)}(\gamma)\right)}+\frac{\kappa\mathbb{E}[Y_i]}{f\left(Z^{(k)}(\gamma)\right)}\sum_{\gamma'}\left[{\mathbf{P}(k)}\right]_{\gamma\times\gamma'}=1.
				\end{align}
				Substituting \eqref{205} and \eqref{206} into \eqref{eq144v2} yields:
				\begin{equation}\label{207}
					\left|\rho_{K+1}-\rho_{K}\right|\le M(1-\epsilon)^{\frac{K-1}{L}},
				\end{equation}
				indicating that the sequence $\{\rho_{K}\}_{K\in\mathbb{N}^+}$ forms a \textit{Cauchy sequence}. Specifically, for $\forall T>1$:
				\begin{subequations}
					\begin{align}			|\rho_{K+T}-\rho_{K}|&\le\sum_{t=0}^{T-1}|\rho_{K+t+1}-\rho_{K+t}|\label{208a}\\&\label{208b}\le\tau M\sum_{t=0}^{T-1}(1-\epsilon)^{\frac{K+t-1}{L}}\\&=\frac{\tau M(1-\epsilon)^{(K-1)/L}(1-(1-\epsilon)^{T/L})}{1-(1-\epsilon)^{1/L}},
					\end{align}
				\end{subequations}
				where \eqref{208a} is established from the \textit{triangle inequality} and inequality \eqref{208b} holds from \eqref{207}.
				Taking the limit $T\to\infty$:
				\begin{equation}\label{eq155v2}\begin{aligned}
						|\rho_{K}-\rho_{\infty}|&\le\lim_{T \to \infty}\frac{\tau M(1-\epsilon)^{(K-1)/L}(1-(1-\epsilon)^{T/L})}{1-(1-\epsilon)^{1/L}}\\&= \frac{M(1-\epsilon)^{K/L}}{1-(1-\epsilon)^{1/L}}.
					\end{aligned}
				\end{equation}
				This indicates that $\rho_K$ will converge to a bounded value $\rho_{\infty}$.
				\subsubsection{Case 2: $\gamma^\star\ne\gamma^{\text{r}}$}
				
				The primary objective of this subsection is to demonstrate that the conclusion under $\gamma^\star=\gamma^{\text{r}}$ remains valid even when if $\gamma^\star\ne\gamma^{\text{r}}$. To prove this, we introduce an auxiliary iteration sequence in the following:
				\begin{iteration} (Auxiliary Iteration Sequence): For a given $0<\kappa<1$, the auxiliary iteration generates sequence $\{\overline{\rho_{K}}\}_{K\in\mathbb{N}^+}$ and $\{\overline{W}_{K}(\gamma)\}_{\gamma\in\mathcal{S}\times\mathcal{Y}\times\mathcal{A}}^{K\in\mathbb{N}^+}$ with a starting initial value $\{\overline{W}_{0}(\gamma)\}_{\gamma\in\mathcal{S}\times\mathcal{Y}\times\mathcal{A}}$ by \eqref{iteration156}, where $\gamma^{\text{r}}{\in\mathcal{S}\times\mathcal{Y}\times\mathcal{A}}$ is a fixed \textit{reference state} with an initial condition $\overline{W}_{0}(\gamma^{\text{r}})=0$ and $\overline{W}_0(\gamma)=\widetilde{W}_0(\gamma)$ for $\forall \gamma\in\mathcal{S}\times\mathcal{Y}\times\mathcal{A}$.
				\end{iteration}
				\begin{figure*}
					\begin{subequations}\label{iteration156}
						\begin{align}
							&\overline{\rho_{K+1}}=\overline{W}_K(\gamma^{\star})+\min _{A_i, Z_i}\left\{\frac{q(\gamma^{\star},Z_i,A_i)-\kappa\overline{W}_K(\gamma^{\star})\cdotp\mathbb{E}[Y_i]+\kappa\mathbb{E}\left[\overline{W}_K(\gamma')|\gamma^{\star},Z_i,A_i\right]\cdotp\mathbb{E}[Y_i]}{f(Z_i)}\right\},\\
							&\overline{W}_{K+1}(\gamma)=\overline{W}_K(\gamma)+\min _{A_i, Z_i}\left\{\frac{q(\gamma,Z_i,A_i)-\kappa\overline{W}_K(\gamma)\cdotp\mathbb{E}[Y_i]+\kappa\mathbb{E}\left[\overline{W}_K(\gamma')|\gamma,Z_i,A_i\right]\cdotp\mathbb{E}[Y_i]}{f(Z_i)}\right\}-\overline{\rho_{K+1}},\notag\\& \gamma\in\mathcal{S}\times\mathcal{Y}\times\mathcal{A},
						\end{align}
					\end{subequations}
					\hrulefill
				\end{figure*}
				A key property of the auxiliary iteration sequence is that there exists an $M>0$ such that for all $K\ge1$,
				\begin{equation}\label{ieq157}
					\begin{aligned}			&\max_{\gamma}\left\{\overline{{W}}_{K+1}(\gamma)-\overline{{W}}_{K}(\gamma)\right\}-\\&\min_{\gamma}\left\{\overline{{W}}_{K+1}(\gamma)-\overline{{W}}_{K}(\gamma)\right\}\le M(1-\epsilon)^{K/L}.
					\end{aligned}
				\end{equation}
				Compare \eqref{iteration156} with \eqref{prop2}, the relationship between $\overline{W}_{K}{(\gamma)}$ and $\widetilde{W}_{K}{(\gamma)}$ is established by:
				\begin{equation}\label{eq159}
					\begin{aligned}
						\widetilde{W}_{K}{(\gamma)}=\overline{W}_{K}{(\gamma)}+\Phi(\overline{\mathbf{W}}_{K-1}),
					\end{aligned}
				\end{equation}
				where $\overline{\mathbf{W}}_{K}$ is a column vector consisted of $\overline{{W}}_{K}(\gamma)$ for all ${\gamma\in\mathcal{S}\times\mathcal{Y}\times\mathcal{A}}$ and $\Phi(\overline{\mathbf{W}}_{K-1})$ is given in \eqref{213} at the top of the next page.
				\begin{figure*}
					\begin{align}\label{213}
						\Phi(\overline{\mathbf{W}}_{K-1})=&\overline{W}_{K-1}(\gamma^{\star})+\min _{A_i, Z_i}\left\{\frac{q(\gamma^{\star},Z_i,A_i)-\kappa\overline{W}_{K-1}(\gamma^{\star})\cdotp\mathbb{E}[Y_i]+\kappa\mathbb{E}\left[\overline{W}_{K-1}(\gamma')|\gamma^{\star},Z_i,A_i\right]\cdotp\mathbb{E}[Y_i]}{f(Z_i)}\right\}\nonumber\\-&\overline{W}_{K-1}(\gamma^{\text{r}})-\min _{A_i, Z_i}\left\{\frac{q(\gamma^{\text{r}},Z_i,A_i)-\kappa\overline{W}_{K-1}(\gamma^{\text{r}})\cdotp\mathbb{E}[Y_i]+\kappa\mathbb{E}\left[\overline{W}_{K-1}(\gamma')|\gamma^{\text{r}},Z_i,A_i\right]\cdotp\mathbb{E}[Y_i]}{f(Z_i)}\right\}.
					\end{align}
					\hrulefill
				\end{figure*}
				Meanwhile, we can establish
				\begin{equation}\label{eq160}
					\begin{aligned}
						\widetilde{W}_{K+1}{(\gamma)}=\overline{W}_{K+1}{(\gamma)}+\Phi(\overline{\mathbf{W}}_{K}).
					\end{aligned}
				\end{equation}
				Subtracting \eqref{eq159} from \eqref{eq160} yields:
				\begin{equation}\label{215}
					\begin{aligned}
						&\widetilde{W}_{K+1}{(\gamma)}-\widetilde{W}_{K}{(\gamma)}=\\&\overline{W}_{K+1}{(\gamma)}-\overline{W}_{K}{(\gamma)}+\Phi(\overline{\mathbf{W}}_{K})-\Phi(\overline{\mathbf{W}}_{K-1}),
					\end{aligned}
				\end{equation}
				Thus, by applying the max and min operators to both sides of \eqref{215}, we obtain the following:
				\begin{subequations}
					\begin{align}
						\max_{\gamma}\Big\{\widetilde{W}_{K+1}{(\gamma)}-\widetilde{W}_{K}{(\gamma)}\Big\}\notag&=\max_{\gamma}\Big\{\overline{W}_{K+1}{(\gamma)}-\overline{W}_{K}{(\gamma)}\Big\}\\&+\Phi(\overline{\mathbf{W}}_{K})-\Phi(\widetilde{\mathbf{W}}_{K-1}),\label{eq162a}\\
						\min_{\gamma}\Big\{\widetilde{W}_{K+1}{(\gamma)}-\widetilde{W}_{K}{(\gamma)}\Big\}\notag&=\min_{\gamma}\Big\{\overline{W}_{K+1}{(\gamma)}-\overline{W}_{K}{(\gamma)}\Big\}\\&+\Phi(\overline{\mathbf{W}}_{K})-\Phi(\widetilde{\mathbf{W}}_{K-1}),\label{eq162b}
					\end{align}
				\end{subequations}
				Subtracting \eqref{eq162b} from \eqref{eq162a} yields the key equality in \eqref{217}, which is shown at the top of the next page.
				\begin{figure*}
					\begin{equation}\label{217}
						\begin{aligned}			&\max_{\gamma}\Big\{\widetilde{W}_{K+1}{(\gamma)}-\widetilde{W}_{K}{(\gamma)}\Big\}-\min_{\gamma'}\Big\{\widetilde{W}_{K+1}{(\gamma)}-\widetilde{W}_{K}{(\gamma)}\Big\}\\&=\max_{\gamma}\Big\{\overline{W}_{K+1}{(\gamma)}-\overline{W}_{K}{(\gamma)}\Big\}-\min_{\gamma'}\{\overline{W}_{K+1}{(\gamma)}-\overline{W}_{K}{(\gamma)}.
						\end{aligned}
					\end{equation}
					\hrulefill
				\end{figure*}
				Substituting \eqref{ieq157} into \eqref{217}, we have that the original sequence $\{\widetilde{{W}}_{K}(\gamma)\}_{K\in\mathbb{N}^+}$ is upper bounded by:
				\begin{equation}\label{218}
					\begin{aligned}				&\max_{\gamma}\left\{\widetilde{{W}}_{K+1}(\gamma)-\widetilde{{W}}_{K}(\gamma)\right\}\\&-\min_{\gamma}\left\{\widetilde{{W}}_{K+1}(\gamma)-\widetilde{{W}}_{K}(\gamma)\right\}\\&\le M(1-\epsilon)^{K/L}.
					\end{aligned}
				\end{equation}
				With \eqref{218} in hand, it follows that $\{\widetilde{{W}}_{K}(\gamma)\}_{K\in\mathbb{N}^+}$ and $\{\rho_{K}(\lambda)\}_{K\in\mathbb{N}^+}$ are both \textit{Cauchy sequences}, following a reasoning similar to that in \textit{Case 1}: \eqref{184}-\eqref{eq155v2}. 
				\subsection{Convergence Direction}\label{J-B}
				In this subsection, we establish the relationship between the convergent values and the solution to \eqref{eqfunction}. As $\widetilde{W}_{K}(\gamma)$ and $\rho_{K}$ are convergent, we have
				\begin{subequations}
					\begin{align}
						\lim_{K \to \infty}\rho_{K}(\lambda)&=\lim_{K \to \infty}\rho_{K+1}(\lambda)=\rho_{\infty}(\lambda),\label{219a}\\\lim_{K \to \infty}\widetilde{{W}}_{K+1}(\gamma)&=\lim_{K \to \infty}\widetilde{{W}}_{K}(\gamma)=\widetilde{{W}}_{\infty}(\gamma),\forall~\gamma,\label{219b}
					\end{align}
				\end{subequations}
				Substituting \eqref{219a} and \eqref{219b} into \eqref{prop2} yields \eqref{220}, presented at the top of this page.
				\begin{figure*}
					\begin{subequations}\label{220}
						\begin{align}
							\rho_{\infty}=\widetilde{W}_\infty(\gamma^{\text{r}})+\min _{A_i, Z_i}\left\{\frac{q(\gamma^{\text{r}},Z_i,A_i)-\kappa\widetilde{W}_{\infty}(\gamma^{\text{r}})\cdotp\mathbb{E}[Y_i]+\kappa\mathbb{E}[Y_i]\sum_{\gamma'}p(\gamma'|\gamma^{\text{r}},Z_i,A_i)\cdot \widetilde{W}_\infty(\gamma')}{f(Z_i)}\right\},\label{eq192}\\
							0=\min _{A_i, Z_i}\left\{\frac{q(\gamma,Z_i,A_i)-\rho_{\infty}\cdot f(Z_i)+\kappa\mathbb{E}[Y_i]\sum_{\gamma'}p(\gamma'|\gamma,Z_i,A_i)\cdot \widetilde{W}_\infty(\gamma')-\kappa\widetilde{W}_{\infty}(\gamma)\cdotp\mathbb{E}[Y_i]}{f(Z_i)}\right\},\nonumber\\\forall \gamma\in\mathcal{S}\times\mathcal{Y}\times\mathcal{A},\label{eq196}
						\end{align}
					\end{subequations}
					\hrulefill
				\end{figure*}
				
				Since $f(Z_i)>0$, from \eqref{eq196} we have that
				\begin{equation}\label{eq194}
					\begin{aligned}
						0=&\min _{A_i, Z_i}\Big\{g(\gamma,Z_i,A_i;\rho_{\infty})+
						\\&\kappa\mathbb{E}[Y_i]\sum_{\gamma'}p(\gamma'|\gamma,Z_i,A_i)\widetilde{W}_{\infty}(\gamma')-	{\kappa\mathbb{E}[Y_i]}\widetilde{W}_{\infty}(\gamma)\Big\},\\&\quad\quad\quad\quad\quad\quad\quad\quad\quad\quad\quad\quad\quad\forall \gamma\in\mathcal{S}\times\mathcal{Y}\times\mathcal{A},
					\end{aligned}
				\end{equation}
				where $g(\gamma,Z_i,A_i;\rho_{\infty})=q(\gamma,Z_i,A_i)-\rho_{\infty}\cdot f(Z_i)$.
				By moving the term ${\kappa\mathbb{E}[Y_i]}\cdot\widetilde{W}_{\infty}(\gamma)$ on the left-hand side of \eqref{eq194} to the left-hand side and rewriting the summation terms in \eqref{eq194} into an expectation form, we have
				\begin{equation}\label{eq195}
					\begin{aligned}
						&{\kappa\mathbb{E}[Y_i]}\cdot\widetilde{W}_{\infty}(\gamma)=\min _{A_i, Z_i}\Big\{g(\gamma,Z_i,A_i;\rho_{\infty})\\&+\mathbb{E}\left[\kappa\mathbb{E}[Y_i]\cdot\widetilde{W}_{\infty}(\gamma')|\gamma,Z_i,A_i\right]\Big\},\forall \gamma\in\mathcal{S}\times\mathcal{Y}\times\mathcal{A}.
					\end{aligned}
				\end{equation}
				Meanwhile, substituting \eqref{185} into \eqref{eq192} and rewriting the summation terms in \eqref{eq192} yields:
				\begin{equation}\label{eq197}
					\rho_{\infty}=\min _{A_i, Z_i}\left\{\frac{q(\gamma^{\text{r}},Z_i,A_i)+\mathbb{E}\left[\kappa\mathbb{E}[Y_i]\cdot\widetilde{W}_{\infty}(\gamma')|\gamma,Z_i,A_i\right]}{f(Z_i)}\right\},
				\end{equation}
				Compare \eqref{eq195} and \eqref{eq197} with \eqref{eqfunction}, we have that  $\left(\rho_{\infty},\left\{\kappa\mathbb{E}[Y_i]\cdot\widetilde{W}_\infty(\gamma)\right\}_{\gamma\in\mathcal{S}\times\mathcal{Y}\times\mathcal{A}}\right)$ is a root of \eqref{eqfunction}. We thus establish the equivalence:
				\begin{equation}\label{eq227}
					\begin{aligned}
						\rho^\star&=\rho_{\infty},\\
						\kappa\mathbb{E}[Y_i]\cdot\widetilde{W}_\infty(\gamma)&=\widetilde{W}^{\star}(\gamma),\forall\gamma\in\mathcal{S}\times\mathcal{Y}\times\mathcal{A},
					\end{aligned}
				\end{equation}	
				\textcolor{black}{Substituting \eqref{eq227} into \eqref{eq155v2} completes the proof of Theorem~\ref{the6upperbound2}. 
					Moreover, substituting \eqref{eq227} into \eqref{eq201} and \eqref{eq155v2} and letting $K\to\infty$, we obtain}
				\begin{equation}\color{black}
					\begin{aligned}
						0\le\lim_{K\to\infty}|\rho^K-\rho^\star|\le\lim\limits_{K\to\infty}\frac{M(1-\epsilon)^{K/L}}{1-(1-\epsilon)^{1/L}}=0,\\
						0\le\lim_{K\to\infty}\left|\widetilde{{W}}_{K}(\gamma)-\frac{\widetilde{{W}}^{\star}(\gamma)}{\kappa\mathbb{E}[Y_i]}\right|\le\lim\limits_{K\to\infty}\frac{M(1-\epsilon)^{K/L}}{1-(1-\epsilon)^{1/L}}=0.
					\end{aligned}
				\end{equation}
				\textcolor{black}{Therefore, by the squeeze theorem, both limits are zero, which completes the proof of Theorem~\ref{Lemma6}.}
				\section{Proof of Lemma \ref{Lemma7v2}}\label{appendixkv3}
				By substituting the definitions of $\mathcal{Q}^{\lambda+\theta}$ and $\mathcal{F}^{\lambda+\theta}$ from \eqref{eq44v2} and \eqref{eq45v2} into \eqref{eq37} and \eqref{eq38}, we express $\Upsilon(\theta,\lambda;f_{\max})$ as:
				\begin{equation}
					\begin{aligned}
						\Upsilon(\theta,\lambda;f_{\max})&=\mathcal{Q}^{\lambda+\theta}-(\lambda+\theta)\mathcal{F}^{\lambda+\theta}+\frac{\theta}{f_{\max}}\\&=\mathcal{Q}^{\lambda+\theta}-(\lambda+\theta)(\mathcal{F}^{\lambda+\theta}-\frac{1}{f_{\max}})-\frac{\lambda}{f_{\max}}.
					\end{aligned}
				\end{equation}
				We now demonstrate that for a fixed value of $\lambda$, the function $\Upsilon(\theta, \lambda; f_{\max})$ is non-increasing with respect to $\theta$, provided that $\mathcal{F}^{\lambda+\theta} \geq 1/f_{\max}$. For any $\Delta \theta \geq 0$, we can establish \eqref{226} at the top of the next page,
				\begin{figure*}
					\begin{subequations}\label{226}
						\begin{align}
							&\Upsilon(\theta+\Delta \theta,\lambda;f_{\max})-\Upsilon(\theta,\lambda;f_{\max})\nonumber\\&\overset{(a)}{=}\frac{\Delta\theta}{f_{\max}}+\lim _{\mathrm{n} \rightarrow \infty}\frac{1}{n} {\sum_{i=0}^{n-1} \mathbb{E}_{\phi^\star_{\theta+\Delta \theta+\lambda}}\left[\sum_{t=\it{D}_{i}}^{\it{D}_{i+1}-1} \mathcal{C}(X_t, A_i)\right]}-{(\theta+\Delta \theta+\lambda)}\lim _{\mathrm{n} \rightarrow \infty}\frac{1}{n} {\sum_{i=0}^{n-1}\mathbb{E}_{\phi^\star_{\theta+\Delta \theta+\lambda}} \left[Z_i+Y_{i+1}\right]}\label{226a}\\&-\left(\lim _{\mathrm{n} \rightarrow \infty}\frac{1}{n} {\sum_{i=0}^{n-1} \mathbb{E}_{\phi^\star_{ \theta+\lambda}}\left[\sum_{t=\it{D}_{i}}^{\it{D}_{i+1}-1} \mathcal{C}(X_t, A_i)\right]}-{(\theta+\lambda)}\lim _{\mathrm{n} \rightarrow \infty}\frac{1}{n} {\sum_{i=0}^{n-1}\mathbb{E}_{\phi^\star_{\theta+\lambda}} \left[Z_i+Y_{i+1}\right]}\right)\notag\\
							&{\le}\frac{\Delta\theta}{f_{\max}}-\Delta\theta\lim _{\mathrm{n} \rightarrow \infty}\frac{1}{n} {\sum_{i=0}^{n-1}\mathbb{E}_{\phi^\star_{\theta+\lambda}} \left[Z_i+Y_{i+1}\right]}\\&=-\Delta\theta\left(\mathcal{F}^{\lambda+\theta}-\frac{1}{f_{\max}}\right)\\&\le0,\label{226b}
						\end{align}
					\end{subequations}
					\hrulefill
				\end{figure*}
				where \eqref{226a} follows from expanding $\mathcal{Q}^{\lambda+\theta}$ and $\mathcal{F}^{\lambda+\theta}$ according to \eqref{eq44v2} and \eqref{eq45v2}; and inequality \eqref{226b} is established by the optimality of policy $\phi_{\theta+\Delta\theta+\lambda}^\star$:
				\begin{equation}
					\mathcal{L}(\phi_{\theta+\Delta\theta+\lambda}^\star;\theta+\Delta\theta,\lambda,f_{\max})\le\mathcal{L}(\phi_{\theta+\lambda}^\star;\theta+\Delta\theta,\lambda,f_{\max}).
				\end{equation}
				This completes the proof.
				\section{Proof of Corollary \ref{coro1}}\label{appendixlv2}
				\subsection{Case 1: $\mathcal{F}^{\lambda^+} \geq 1/f_{\max}$}
				From part ($ii$) of Lemma \ref{Lemma7v2}, which indicates the non-decreasing property of $\mathcal{F}^{\lambda+\theta}$, it turns out that if $\mathcal{F}^{\lambda^+} \geq 1/f_{\max}$, then for any $\theta \geq 0$, \begin{equation}\mathcal{F}^{\lambda + \theta} \ge \mathcal{F}^{\lambda^+}\ge 1/f_{\max}.\end{equation} From this, it follows that the inequality $\mathcal{F}^{\lambda + \theta} \geq \frac{1}{f_{\max}}$ is always true, satisfying the KKT condition \eqref{KKT3}. Additionally, under this assumption, part ($iv$) of Lemma \ref{Lemma7v2} ensures that $\Upsilon(\theta, \lambda; f_{\max})$ is non-increasing for $\theta \geq 0$. Hence, the minimum value of $\Upsilon(\theta, \lambda; f_{\max})$ occurs at $\theta = 0$:
				\begin{equation}
					\min_{\theta\ge0}\Upsilon(\theta,\lambda;f_{\max})=\Upsilon(0,\lambda;f_{\max})=U(\lambda).
				\end{equation}
				Note that $\theta^\star_{\lambda}=0$ also satisfies the KKT condition \eqref{KKT3}, we accomplish the proof under this case.
				
				\subsection{Case 2: $\mathcal{F}^{\lambda^+} < 1/f_{\max}$}
				
				On the other hand, consider the case where $\mathcal{F}^{\lambda^+} < 1/f_{\max}$. According to part ($ii$) of Lemma \ref{Lemma7v2}, there exists a threshold value $\theta^{\text{tr}}$ such that \eqref{eq46} holds. When $\theta\in[0,\theta^{\text{tr}})$, it follows from \eqref{eq46} that $\mathcal{F}^{\lambda+\theta} \leq \frac{1}{f_{\max}}$, which contradicts the KKT condition \eqref{KKT3}, indicating that values of $\theta$ in this range do not satisfy the necessary optimality conditions. However, for $\theta \in [\theta^{\text{tr}}, \infty)$, we observe from \eqref{eq46} that the KKT condition \eqref{KKT3} is satisfied, as $\mathcal{F}^{\lambda+\theta} \geq \frac{1}{f_{\max}}$. Furthermore, from part $(iv)$ of Lemma \ref{Lemma7v2}, $\Upsilon(\theta,\lambda;f_{\max})$ is non-increasing with respect to $\theta$ in this feasible region. Therefore, the minimum value of $\Upsilon(\theta, \lambda; f_{\max})$, under the KKT conditions occurs at the smallest $\theta$ for which the KKT condition \eqref{KKT3} holds, namely $\theta^\star_{\lambda} = \theta^{\text{tr}}$. 
				\section{Proof of Theorem \ref{theorem5}}\label{appendixkM}
				\subsection{Case 1: $\mathcal{F}^{\lambda^+}\ge1/f_{\max}$} From part ($i$) of Corollary \ref{coro1}, we know that the optimal Lagrangian multiplier satisfies $\theta_\lambda^\star=0$. Substituting $\theta_\lambda^\star=0$ into \eqref{kkt2} yields 
				\begin{equation}
					\begin{aligned}
						&\phi^\star_{\lambda}=\\&\mathop{\arg\min}_{\phi}\lim _{n \rightarrow \infty} \frac{1}{n} \sum_{i=0}^{n-1}\mathbb{E}_{\phi}\Bigg\{\sum_{t=D_i}^{D_{i+1}-1} \mathcal{C}\left(X_t, A_i\right)-\lambda \left(Z_i+Y_{i+1}\right)\Bigg\}.
					\end{aligned}
				\end{equation}In this case, the optimal policy is equivalent to the optimal policy determined in \eqref{eq19}, which is stationary deterministic.
				\subsection{Case 2: $\mathcal{F}^{\lambda^+}<1/f_{\max}$ and $\mathcal{F}^{(\lambda+\theta^\star_\lambda)^+}=1/f_{\max}$}
				As $\mathcal{F}^{\lambda^+}<1/f_{\max}$, we can establish from Corollary \ref{coro1} that $\theta_\lambda^\star>0$. The condition $\mathcal{F}^{(\lambda+\theta^\star_\lambda)^+}=1/f_{\max}$ leads to the fact that \eqref{KKT3} and \eqref{kktT4} are satisfied naturally, and thus the optimal policy is exactly the stationary deterministic policy $\phi^\star_{\lambda+\theta_\lambda^\star}$. 
				\subsection{Case 3: $\mathcal{F}^{\lambda^+}<1/f_{\max}$ and $\mathcal{F}^{(\lambda+\theta^\star_\lambda)^+}>1/f_{\max}$}
				($iii$). For this case, it is proved in \cite[Theorem 4.4]{li2011finding} that the policy given in \eqref{eq52} and \eqref{eq53v2} is the optimal policy.
				\section{Proof of Theorem \ref{the6}}\label{appendixkv2}
				\subsection{Proof of Part ($i$)}\label{sec:proof-of-part-i}
				Substituting \eqref{eq37}, \eqref{eq44v2} and \eqref{eq45v2} into \eqref{eq38}, we have
				\begin{equation}\label{eq203}
					\Upsilon(\theta,\lambda;f_{\max})=\mathcal{Q}^{\lambda+\theta}-(\lambda+\theta)\mathcal{F}^{\lambda+\theta}+\frac{\theta}{f_{\max}}.
				\end{equation}
				Let $\lambda=h^\star$ and $\rho=h^\star+\theta$, \eqref{eq203} turns to
				\begin{equation}\label{eq44}
					\Upsilon(\rho-h^\star,h^\star;f_{\max})=\mathcal{Q}^{\rho}-\rho\mathcal{F}^{\rho}+\frac{\rho-h^\star}{f_{\max}}.
				\end{equation}
				Note that $d(h^\star;f_{\max})=U(h^\star;f_{\max})=0$, the Problem \ref{p6} turns to the following equation:
				\begin{equation}\label{eq45}
					0=\max_{\rho\ge h^\star}\left\{\underbrace{\mathcal{Q}^{\rho}-\rho\left(\mathcal{F}^{\rho}-\frac{1}{f_{\max}}\right)-\frac{h^\star}{f_{\max}}}_{\mathcal{J}(\rho)}\right\}.
				\end{equation}
				Similar to the proof of part ($iv$) of Lemma \ref{Lemma7v2}, we can establish the following Lemma.
				\begin{Lemma}\label{Lemma13}
					If $\mathcal{F}^{(\rho)^-}\ge1/f_{\max}$, then $\mathcal{J}(\rho)$ is non-increasing with respect to $\rho$.
				\end{Lemma}
				We next employ a proof by contradiction to establish this result. Specifically, we consider the following three cases: 
				\begin{itemize}
					\item \ul{\textit{Case 1}}: If $h^\star<\rho^{\star}$, we know that
					\begin{equation}
						\mathcal{J}(\rho^\star)=\frac{\rho^{\star}-h^{\star}}{f_{\max}}>0.
					\end{equation}
					This contradicts \eqref{eq45}, as \eqref{eq45} is equivalent to:
					\begin{equation}
						\mathcal{J}(\rho)\le0,\forall \rho\ge h^\star.
					\end{equation}
					\item \ul{\textit{Case 2}}: If $h^\star=\rho^{\star}$, we have that $\mathcal{F}^{(h^{\star})^-}\ge 1/f_{\max}$. Since $\mathcal{F}^{\lambda}$ is non-decreasing with $\lambda$, as indicated in Lemma \ref{Lemma7v2}, we know that 
					\begin{equation}
						\mathcal{F}^{(\rho)^-}\ge\frac{1}{f_{\max}},\forall\rho\ge h^{\star},
					\end{equation} which indicates that $\mathcal{J}(\rho)$ is non-increasing with respect to $\rho$, as shown in Lemma \ref{Lemma13}. As a result, we have that
					\begin{equation}
						\begin{aligned}
							0=\max_{\rho\ge h^{\star}}\left\{\mathcal{J}(\rho)\right\}&=\mathcal{J}(h^{\star})\\&=\mathcal{Q}^{h^{\star}}-h^{\star}\mathcal{F}^{h^\star}=U(h^\star).
						\end{aligned}
					\end{equation}
					This establishes that $U(h^{\star})=0$. Since the root of $U(\lambda)$ is unique, as detailed in part ($iii$) of Lemma \ref{l3}, it is sufficient to verify that $\rho^{\star}=h^\star$.
					\item \ul{\textit{Case 3}}: If $h^\star>\rho^{\star}$, we have that $\mathcal{F}^{(h^{\star})^-}{\ge}\mathcal{F}^{(\rho^{\star})^-}\ge 1/f_{\max}$ as $\mathcal{F}^{\lambda}$ is non-decreasing with $\lambda$ ( part ($ii$) of Lemma \ref{Lemma7v2}). Similar to \ul{\textit{Case 2}}, we know that $\mathcal{J}(\rho)$ is non-increasing with respect to $\rho$ and thus $U(h^{\star})=0$. Since the root of $U(\lambda)$ is unique (part ($iii$) of Lemma \ref{l3}), we have $\rho^\star=h^{\star}$. This contradicts the case where $h^\star>\rho^{\star}$.
				\end{itemize}
				As such, we establish that $h^\star=\rho^\star$.
				\subsection{Proof of Part ($ii$)}
				If $\mathcal{F}^{(\rho^{\star})^-}<1/f_{\max}$, \ul{\textit{Case 1}} in the proof of part ($i$) remains contradictory and thus $\rho^{\star}\ge h^\star$. Since $\mathcal{F}^{\lambda}$ is non-decreasing with $\lambda$, as indicated in Lemma \ref{Lemma7v2}, we know that 
				\begin{equation}
					\mathcal{F}^{(h^{\star})^-}\le\mathcal{F}^{(\rho^{\star})^-}<1/f_{\max}.
				\end{equation}
				
				\begin{Lemma}
					$\mathcal{J}(\rho)$ and $U(\rho)$ is uniformly absolutely continuous, with the derivative given as:
					\begin{subequations}
						\begin{align}
							\frac{\mathrm{d}U(\rho)}{\mathrm{d}\rho}&=-\mathcal{F}^{\rho},\label{243a}\\
							\frac{\mathrm{d}\mathcal{J}(\rho)}{\mathrm{d}\rho}&=-\mathcal{F}^{\rho}+\frac{1}{f_{\max}}.\label{243b}
						\end{align}
					\end{subequations} 
				\end{Lemma}
				\begin{proof}
					\eqref{243a} is a restatement of \cite[Lemma 3.1]{BEUTLER1985236}. \eqref{243b} is a corollary of \eqref{243a}.
				\end{proof}
				
				As a result, the maximum value of $\mathcal{J}(\rho)$ is taken where $\mathcal{F}^{\rho}=\frac{1}{f_{\max}}$. We denote one of the root as
				\begin{equation}
					\rho^{\downarrow}=\inf\left\{\rho\ge h^{\star}:\mathcal{F}^{\rho}=\frac{1}{f_{\max}}\right\}.
				\end{equation}
				The root $\rho^{\downarrow}$ always exists since:
				\begin{equation}
					\left.\frac{\mathrm{d}\mathcal{J}(\rho)}{\mathrm{d}\rho}\right|_{\rho=(h^{\star})^-}<0,
				\end{equation}
				and the monotonic property of $\mathcal{F}(\rho)$. As a result, the Problem \ref{p6} turns to the following problem with a stricter equation constraint:
				
				\begin{problem}[\textit{CMDP with equation constraint}]\label{p8}
					\begin{equation}
						\begin{aligned}
							&H(\lambda;f_{\max})\triangleq\\&\inf _{\phi} \lim _{\mathrm{n} \rightarrow \infty}\frac{1}{n} {\sum_{i=0}^{n-1}\left\{ \mathbb{E}_{\boldsymbol{\psi}}\left[\sum_{t=\it{D}_{i}}^{\it{D}_{i+1}-1} \mathcal{C}(X_t, A_i)\right]-\lambda\mathbb{E} \left[Z_i+Y_{i+1}\right]\right\}}\\
							&\text{s.t. } \lim_{T \to \infty} \frac{1}{T} \mathbb{E}_{\phi}\left[\sum_{i=1}^{T}(Y_{i+1}+Z_i)\right] = \frac{1}{f_{\text{max}}}.
						\end{aligned}
					\end{equation}
				\end{problem}
				
				As did in \cite[Chapter 8.8]{puterman2014markov}, we can utilize the \textit{dual relaxation} to reformulate Problem \ref{p8} as the following LP problem:
				
				\begin{problem}[\textit{LP Transformation of {Problem \ref{p8}}}]\label{p9}
					\begin{align}
						H(\lambda;f_{\max})=\min_{\mathbf{x}} &\sum_{\gamma,z,a}(q(\gamma,z,a)-\lambda f(z))x(\gamma,z,a)\nonumber  \\
						\mathrm{s.t.}&~~~ \eqref{eq53}-\eqref{eq56}
					\end{align}
				\end{problem}
				
				Our remaining focus is to explicitly express the root of $H(\lambda;f_{\max})$ in Problem \ref{p9}.
				By substituting the constraint \eqref{eq53} into the objective function of Problem \ref{p9}, this problem is decomposed as: 
				\begin{align}
					H(\lambda;f_{\max})=\min_{\mathbf{x}} &\sum_{\gamma,z,a}q(\gamma,z,a)x(\gamma,z,a)-\frac{\lambda}{f_{\max}}\nonumber   \\
					\mathrm{s.t.}&~~~ \eqref{eq53}-\eqref{eq56}
				\end{align}
				Since $\lambda/f_{\max}$ is independent of the decision variable $\mathbf{x}$, let $\lambda=h^\star$ and $H(h^\star;f_{\max})=0$, we establish the following LP program:
				
				\begin{problem}[\textit{Closed-form Root Solution}]\label{p10}
					\begin{align}
						\frac{h^\star}{f_{\max}}=\min_{\mathbf{x}} &\sum_{\gamma,z,a}q(\gamma,z,a)x(\gamma,z,a)\nonumber\\
						\mathrm{s.t.}&~~~ \eqref{eq53}-\eqref{eq56}
					\end{align}
				\end{problem}
				
				Comparing Problem \ref{p10} with Problem \ref{p7} establishes the relationship $h^\star=f_{\max}\cdot {Q}^\star(f_{\max})$.
				
				\section{Proof of Theorem \ref{Lemmasa} }\label{apposa}
				\subsection{Part ($i$)}
				Part ($i$) holds naturally since $h^{\star}=\rho^{\star}$ when $\mathcal{F}^{(\rho^\star)^-}\ge1/f_{\max}$. In this case, as discussed in Section \ref{sectionIV} and Section \ref{sectionIV2}, $\rho^{\star}$ corresponds to the optimal value without a sampling frequency constraint. Consequently, $h^{\star}$ is independent of $f_{\max}$.
				
				\subsection{Part ($ii$)}
				When $\mathcal{F}^{(\rho^\star)^-}<1/f_{\max}$, we have that $h^{\star}=f_{\max}\cdot Q^{\star}(f_{\max})$. To analyze the relationship between $h^{\star}$ and $f_{\max}$, we need to conduct a comprehensive \textit{sensitivity analysis} on the LP problem specified in Problem \ref{p7}. To facilitate the \textit{sensitivity analysis}, we first rewrite $h^{\star}$ as follows:
				
				\begin{problem}[\textit{Reformulation of $h^{\star}$}]\label{p11}
					\begin{subequations}
						\begin{align}
							h^{\star}=&f_{\max}Q^\star(f_{\max})\\=&\min_{\mathbf{x}} \sum_{\gamma,z,a}f_{\max}q(\gamma,z,a)x(\gamma,z,a) \label{eq242a}\\
							&\mathrm{ s.t.} \sum_{\gamma,z,a} f_{\max}f(z)x(\gamma,z,a)=1,\label{eq242b}\\
							&\quad\quad\eqref{eq53}-\eqref{eq56}\label{appoLemma9} 
						\end{align}
					\end{subequations}
				\end{problem}
				
				For easy notations, we define the feasible region that satisfies constraint \eqref{eq53}-\eqref{eq56} as $\mathcal{D}=\{\mathbf{x}:\eqref{eq53}-\eqref{eq56}\}$. Our goal is to conduct a sensitivity analysis on the new LP problem in Problem \ref{p11}, where the parameter affects both the objective \eqref{eq242a} and subjective \eqref{eq242b}. By applying the Lagrangian dual technique to Problem \ref{p10}, we establish the partial Lagrangian dual problem:
				
				\begin{problem}[\textit{Lagrangian Dual of Problem \ref{p11}}]\label{p12}
					\begin{equation}
						\begin{aligned}
							&\Theta^{\star}=\\&\max_{\lambda}\left\{
							\lambda+f_{\max}{\min_{\mathbf{x}\in\mathcal{D}} \sum_{\gamma,z,a}(q(\gamma,z,a)-\lambda f(z))x(\gamma,z,a)}
							\right\}. 
						\end{aligned}
					\end{equation}
				\end{problem}
				
				The \textit{weak duality} principle establishes that $\Theta^{\star}\le h^{\star}$, while the strong duality holds that $\Theta^{\star}=h^{\star}$. For LP problems, the \textit{Staler's constraint qualification} is always satisfied \cite[Chapter 5.2.3]{boyd2004convex}, which naturally ensures \textit{strong duality}. Thus, for our LP problem, we have $\Theta^{\star}=h^{\star}$. By applying the inverse transformation approach as demonstrated in the transformation from Problem \ref{p8} to Problem \ref{p9}, we obtain:
				\begin{equation}
					U(\lambda)=\min_{\mathbf{x}\in\mathcal{D}} \sum_{\gamma,z,a}(q(\gamma,z,a)-\lambda f(z))x(\gamma,z,a).
				\end{equation}
				This transformation, combined with strong duality, converts Problem \ref{p12} into: \begin{equation}
					h^{\star}=\max_{\lambda}\left\{\underbrace{\lambda+f_{\max}U(\lambda)}_{g(\lambda)}\right\}
				\end{equation}
				To analyze the derivative of $g(\lambda)$, we utilize the following lemma, which is a restatement of \cite[Lemma 3.1]{BEUTLER1985236}.
				\begin{Lemma}
					(Restatement \cite[Lemma 3.1]{BEUTLER1985236}). $U(\lambda)$ is uniformly absolutely continuous, with the derivative given as:
					\begin{equation}
						\frac{\mathrm{d}U(\lambda)}{\mathrm{d}\lambda}=-\mathcal{F}^{\lambda}.
					\end{equation} 
				\end{Lemma}
				Consequently, the derivative of $g(\lambda)$ is:
				\begin{equation}
					\frac{\mathrm{d}g(\lambda)}{\mathrm{d}\lambda}=	\frac{\mathrm{d}(\lambda+f_{\max}U(\lambda))}{\mathrm{d}\lambda}=1-f_{\max}\mathcal{F}^{\lambda}.
				\end{equation}
				Setting $\frac{\mathrm{d}g(\lambda)}{\mathrm{d}\lambda}=0$ and by utilizing the monotonically non-decreasing property of $\mathcal{F}^{\lambda}$ in terms of $\lambda$ (as established in Lemma \ref{Lemma7v2}-($ii$)), we have that the optimal $\lambda^{\star}$ satisfies:
				\begin{equation}
					\mathcal{F}^{(\lambda^{\star})^{-}}<\frac{1}{f_{\max}}\le\mathcal{F}^{(\lambda^{\star})^{+}}
				\end{equation}
				Given that $\mathcal{F}^{(\rho^{\star})^{-}}<1/f_{\max}$, and considering the monotonically non-decreasing property of $\mathcal{F}^{\lambda}$ with respect to $\lambda$, we establish that
				\begin{equation}
					\lambda^{\star}\ge\rho^{\star}.
				\end{equation}
				Thus, from the monotonically non-decreasing property of $U{(\lambda)}$ in terms of $\lambda$ (as established in Lemma \ref{Lemma7v2}-($i$)), we have that
				\begin{equation}
					U(\lambda^{\star})\le U(\rho^{\star})=0.
				\end{equation}
				As last, we conduct the \textit{sensitivity analysis} as follows:
				\begin{equation}
					\frac{\mathrm{d}{h^{\star}}}{\mathrm{d}f_{\max}}=\frac{\mathrm{d}{(\lambda^{\star}+f_{\max}U(\lambda^{\star}))}}{\mathrm{d}f_{\max}}=U(\lambda^{\star})\le0,
				\end{equation}
				which accomplished the part ($ii$).\finished
				\section{Proof of Lemma \ref{Lemma9}}\label{appendixK}
				First, in \eqref{sum1}, it has been proved that the sum of each row of $\widetilde{\mathbf{P}_{\kappa}(k)}\in\mathbb{R}^{|\mathcal{S}\times\mathcal{Y}\times\mathcal{A}|\times|\mathcal{S}\times\mathcal{Y}\times\mathcal{A}|}$ is 1. Define ${p_{\gamma\gamma'}}(Z^{(k)}(\gamma),A^{(k)}(\gamma))$ as the transition probability that $\gamma$ transitions to state $\gamma'$ under actions $(Z^{(k)}(\gamma),A^{(k)}(\gamma))$ in ${\mathbf{P}(k)}$, and $\widetilde{p_{\gamma\gamma'}}(Z^{(k)}(\gamma),A^{(k)}(\gamma))$ as the corresponding element in matrix $\widetilde{\mathbf{P}_{\kappa}(k)}$, we have that
				\begin{equation}\label{eq144}
					\begin{aligned}
						&\widetilde{p_{\gamma\gamma'}}(Z^{(k)}(\gamma),A^{(k)}(\gamma))=\\&\begin{cases}
							{\frac{\kappa\mathbb{E}[Y_i]{p_{\gamma\gamma'}}(Z^{(k)}(\gamma),A^{(k)}(\gamma))}{f(Z^{(k)}(\gamma))}} &\text{if $\gamma\ne \gamma'$,}\\
							1-\frac{\kappa\mathbb{E}[Y_i]-\kappa\mathbb{E}[Y_i]\cdot {p_{\gamma\gamma}}(Z^{(k)}(\gamma),A^{(k)}(\gamma))}{f(Z^{(k)}(\gamma))} &\text{if $\gamma= \gamma'$.}\\
						\end{cases}
					\end{aligned}
				\end{equation}
				Note that $0<\kappa<1$, $Z^{(k)}(\gamma)\ge0$ and $0\le{p_{\gamma\gamma'}}(Z^{(k)}(\gamma),A^{(k)}(\gamma))\le1$, we have
				\begin{equation}\label{case1}
					0\le{\frac{\kappa\mathbb{E}[Y_i]{p_{\gamma\gamma'}}(Z^{(k)}(\gamma),A^{(k)}(\gamma))}{f(Z^{(k)}(\gamma))}}<\frac{\mathbb{E}[Y_i]}{\mathbb{E}[Y_i]+Z^{(k)}(\gamma)}\le1,
				\end{equation}
				and 
				\begin{equation}\label{case2}
					\begin{aligned}
						0<\frac{\kappa\mathbb{E}[Y_i]-\kappa\mathbb{E}[Y_i]\cdot {p_{\gamma\gamma}}(Z^{(k)}(\gamma),A^{(k)}(\gamma))}{f(Z^{(k)}(\gamma))}<\\\frac{\mathbb{E}[Y_i](1-{p_{\gamma\gamma}}(Z^{(k)}(\gamma),A^{(k)}(\gamma)))}{\mathbb{E}[Y_i]+Z^{(k)}(\gamma)}\le1.
					\end{aligned}
				\end{equation}	 
				By combining \eqref{case1} and \eqref{case2}, we establish that the transition probabilities $\widetilde{p_{\gamma\gamma'}}(Z^{(k)}(\gamma),A^{(k)}(\gamma))$ satisfy the condition $0 \leq \widetilde{p_{\gamma\gamma'}}(Z^{(k)}(\gamma),A^{(k)}(\gamma)) \leq 1$. Given this and the result from \eqref{sum1}, it follows that $\widetilde{\mathbf{P}{\kappa}(k)}$ is indeed a stochastic matrix. 
				\setcounter{equation}{271}
								\begin{figure*}[h]
					\begin{subequations}\label{eq157}
						\begin{align}
							\widetilde{p_{\phi,\gamma'\gamma'}^{n+1}}&=\widetilde{p_{\phi,\gamma'\gamma'}^{1}}\widetilde{p_{\phi,\gamma'\gamma'}^{n}}+\sum_{\gamma_1\ne\gamma'}\widetilde{p_{\phi,\gamma'\gamma_1}^{1}}\widetilde{p_{\phi,\gamma_1\gamma'}^{n}}\\
							&=\widetilde{p_{\phi,\gamma'\gamma'}^{1}}\widetilde{p_{\phi,\gamma'\gamma'}^{n}}+\quad\sum_{\gamma_1\ne\gamma'}\widetilde{p_{\phi,\gamma'\gamma_1}^{1}}\sum_{{\gamma_2,\cdots,\gamma_{n}}}\widetilde{p_{\phi,\gamma_1\gamma_2}}\widetilde{p_{\phi,\gamma_2\gamma_3}}\cdots \widetilde{p_{\phi,\gamma_{n}\gamma'}}\\
							&=\widetilde{p_{\phi,\gamma'\gamma'}^{1}}\widetilde{p_{\phi,\gamma'\gamma'}^{n}}+\sum_{\gamma_1\ne\gamma',\gamma_2,\cdots,\gamma_n}\frac{\kappa\mathbb{E}[Y_i]\cdot p_{\phi,\gamma'\gamma_1}^{1}}{f(\phi(\gamma'))}\cdot\prod_{j\notin\mathcal{D}}
							\frac{\kappa\mathbb{E}[Y_i]\cdot p_{\phi,\gamma_{j}\gamma_{j+1}}^{1}}
							{f(\phi(\gamma_j))}\cdot\prod_{j\in\mathcal{D}}\left(1-\frac{\kappa\mathbb{E}[Y_i]-\kappa\mathbb{E}[Y_i]p_{\phi,\gamma_j\gamma_j}^1}{f(\phi(\gamma_j))}\right)\label{220b}\\
							&=\widetilde{p_{\phi,\gamma'\gamma'}^{1}}\widetilde{p_{\phi,\gamma'\gamma'}^{n}}+\sum_{\gamma_1\ne\gamma',\gamma_2,\cdots,\gamma_n}\frac{\left(\kappa\mathbb{E}[Y_i]\right)^{n-|\mathcal{D}+1|}\cdot\prod_{j\in\mathcal{D}}\left(1-\frac{\kappa\mathbb{E}[Y_i]-\kappa\mathbb{E}[Y_i]p_{\phi,\gamma_j\gamma_j}^1}{f(\phi(\gamma_j))}\right)}{f(\phi(\gamma'))\prod_{j\notin\mathcal{D}}f(\phi(\gamma_j))}\cdot p_{\phi,\gamma'\gamma_1}^{1}\prod_{j\notin\mathcal{D}} p_{\phi,\gamma_{j}\gamma_{j+1}}^{1}.
						\end{align}
					\end{subequations}
					\hrulefill
				\end{figure*}
				\setcounter{equation}{262}
				Second, we show that $\widetilde{\mathbf{P}_{\kappa}(k)}$ is always \textit{aperiodic}. This is verified by the following inequality:
				\begin{equation}
					\begin{aligned}	&\widetilde{p_{\gamma\gamma}}(Z^{(k)}(\gamma),A^{(k)}(\gamma))\\&=1-\frac{\kappa\mathbb{E}[Y_i]-\kappa\mathbb{E}[Y_i] {p_{\gamma\gamma}}(Z^{(k)}(\gamma),A^{(k)}(\gamma))}{Z^{(k)}(\gamma)+\mathbb{E}[Y_i]}\\
						&>1-\frac{\mathbb{E}[{Y_i}]-\mathbb{E}[{Y_i}] {p_{\gamma\gamma}}(Z^{(k)}(\gamma),A^{(k)}(\gamma))}{\mathbb{E}[Y_i]}\\&={p_{\gamma\gamma}}(Z^{(k)}(\gamma),A^{(k)}(\gamma))\ge0.
					\end{aligned}
				\end{equation}
				As the self-transition probability is always positive for each state $\gamma$, the chain is always \textit{aperiodic}. At last, we show that if MDP characterized by $p_{\gamma\gamma'}(z,a)$ is a \textit{unichain}, then the MDP characterized by $\widetilde{p_{\gamma\gamma'}}(z,a)$ is also a \textit{unichain}, this in given in the following lemma.
				\begin{Lemma}\label{Lemma10}
					If the MDP characterized by the transition probability $p_{\gamma\gamma'}(z,a)$ is a unichain, then the MDP characterized by  $\widetilde{p_{\gamma\gamma'}}(z,a)$ is also a unichain.
				\end{Lemma}
				\begin{proof}
					See Appendix \ref{appendixl}.
				\end{proof}
				
				According to Lemma \ref{Lemma10}, the Markov Decision Process (MDP) defined by $\widetilde{\mathbf{P}_{\kappa}(k)}$ is an \textit{aperiodic unichain}, satisfying the condition established in \cite{bertsekas2008introduction}. This completes the proof. 
				\section{Proof of Lemma \ref{Lemma10}}\label{appendixl}
				Since the MDP is a \textit{unichain}, it follows that for any policy $\phi:\mathcal{S}\times\mathcal{Y}\times\mathcal{A}\rightarrow\mathcal{Z}\times\mathcal{A}$, the markov chain constitutes some transient states and a single recurrent class. Let $\mathcal{X}_\phi$ denote the \textit{recurrent class} of the Markov chain under policy $\phi$, we have that if $\gamma' \in \mathcal{X}_{\phi}$ and $\gamma\in\mathcal{S}\times\mathcal{Y}\times\mathcal{A}$, there always exists a smallest positive integer $m$ such that:
				\begin{equation}
					[\mathbf{P}_{\phi}^{n}]_{\gamma\times\gamma'}\triangleq p^n_{\phi,\gamma\gamma'}\begin{cases}
						>0 & \text{if } n = m\\
						=0 & \text{if } n<m\\
					\end{cases},
				\end{equation}
				where $\mathbf{P}_{\phi}\in\mathbb{R}^{|\mathcal{S}\times\mathcal{Y}\times\mathcal{A}|\times|\mathcal{S}\times\mathcal{Y}\times\mathcal{A}|}$ is the transition probability matrix under policy $\phi$, with entries $p^1_{\phi,\gamma\gamma'}=p_{\gamma\gamma'}(\phi(\gamma))$. Note that ${p_{\phi,\gamma\gamma'}^{m-1}}=0$, we can apply the \textit{Chapman--Kolmogorov equation} to establish:
				\begin{equation}\label{eq149}
					\begin{aligned}
						{p_{\phi,\gamma\gamma'}^{m}}&=\sum_{\gamma_1}{p_{\phi,\gamma\gamma_1}^1}{p_{\phi,\gamma_1\gamma'}^{m-1}}\\
						&={p_{\phi,\gamma\gamma}^1}{p_{\phi,\gamma\gamma'}^{m-1}}+\sum_{\gamma_1\ne\gamma}{p_{\phi,\gamma\gamma_1}^1}{p_{\phi,\gamma_1\gamma'}^{m-1}}\\
						&=\sum_{\gamma\ne\gamma_1}{p_{\phi,\gamma\gamma_1}^1}{p_{\phi,\gamma_1\gamma'}^{m-1}}.
					\end{aligned}
				\end{equation}
				Iterating \eqref{eq149} for $m-1$ times yields:
				\begin{equation}\label{eq219}
					\begin{aligned}
						{p_{\phi,\gamma\gamma'}^{m}}&=\sum_{\substack{\gamma\ne\gamma_1,\gamma_2\ne\gamma_3,\\\cdots,\gamma_{m-1}\ne\gamma'}}{p_{\phi,\gamma\gamma_1}^1}{p_{\phi,\gamma_1\gamma_2}^1}\cdots{p_{\phi,\gamma_{m-1}\gamma'}^1}.\\
					\end{aligned}
				\end{equation}
				Similarly, define $\widetilde{p_{\phi,\gamma\gamma'}^{m}}$ as the $m$-step transition probability for the Markov chain characterized by $\widetilde{p_{\gamma\gamma'}}(\phi(\gamma))$, we can apply the \textit{Chapman--Kolmogorov equation} to have
				\begin{subequations}
					\begin{align}
						\widetilde{p_{\phi,\gamma\gamma'}^{m}}&=\sum_{{\gamma_1,\gamma_2,\cdots,\gamma_{m-1}}}\widetilde{p_{\phi,\gamma\gamma_1}^1}\widetilde{p_{\phi,\gamma_1\gamma_2}^1}\cdots\widetilde{p_{\phi,\gamma_{m-1}\gamma'}^1}\\&\ge\sum_{\substack{\gamma\ne\gamma_1,\gamma_2\ne\gamma_3,\\\cdots,\gamma_{m-1}\ne\gamma'}}\widetilde{p_{\phi,\gamma\gamma_1}^1}\widetilde{p_{\phi,\gamma_1\gamma_2}^1}\cdots\widetilde{p_{\phi,\gamma_{m-1}\gamma'}^1}\\
						&{=}\frac{\left(\kappa\mathbb{E}[Y_i]\right)^{m}}{f(\phi(\gamma))\prod_{i=1}^{m-1}f(\phi(\gamma_i))}\times\label{247c}\\&\notag\sum_{\substack{\gamma\ne\gamma_1,\gamma_2\ne\gamma_3,\\\cdots,\gamma_{m-1}\ne\gamma'}}{p_{\phi,\gamma\gamma_1}^1}{p_{\phi,\gamma_1\gamma_2}^1}\cdots{p_{\phi,\gamma_{m-1}\gamma'}^1}\\
						&\label{247d}{=}\frac{\left(\kappa\mathbb{E}[Y_i]\right)^{m}\cdot p_{\phi,\gamma\gamma'}^{m}}{f(\phi(\gamma))\prod_{i=1}^{m-1}f(\phi(\gamma_i))}>0,
					\end{align}
				\end{subequations}
				where \eqref{247c} is established by the first line of \eqref{eq144} and \eqref{247d} is obtained by substituting \eqref{eq219}.  
				As such, if $\gamma'$ belongs to a recurrent class of the Markov chain characterized by $p_{\phi,\gamma\gamma'}^1$, it will also belong to a recurrent class of the Markov chain under $\widetilde{p_{\phi,\gamma\gamma'}^1}$, indicating that $\mathcal{X}_\phi$ is a recurrent class of the new chain $\widetilde{p_{\phi,\gamma\gamma'}^1}$.
				
				We next show that if $\gamma'$ is a transient state of the Markov chain $p_{\phi,\gamma\gamma'}^1$, it is also a transient state of the Markov chain $\widetilde{p_{\phi,\gamma\gamma'}^1}$. If $\gamma'$ is a transient state of the Markov unichain $p_{\phi,\gamma\gamma'}^1$, it follows that
				\begin{equation}
					p_{\phi,\gamma'\gamma_1}^{k}p_{\phi,\gamma_1\gamma'}^{n}=0, \text{for } \forall \gamma_1\ne\gamma' ,k,n\in\mathbb{N}^+,
				\end{equation}
				otherwise, the chain will be a multichain or $\gamma'$ is not a transient state. By leveraging the \textit{Chapman--Kolmogorov equation} to expand $p_{\phi,\gamma_1\gamma'}^{n}$, we have
				\begin{equation}\label{eq162}
					\begin{aligned}
						&0=p_{\phi,\gamma'\gamma_1}^{k}p_{\phi,\gamma_1\gamma'}^{n}=\\&\sum_{{\gamma_2,\cdots,\gamma_{n}}}p_{\phi,\gamma'\gamma_1}^{k}p_{\phi,\gamma_1\gamma_2}^{1}p_{\phi,\gamma_2\gamma_3}^{1}\cdots p_{\phi,\gamma_{n}\gamma'}^{1}\ge0.
					\end{aligned}
				\end{equation}
				From the \textit{sandwich theorem}, we know that every element in the summand of \eqref{eq162} is zero:
				\begin{equation}\label{eq154}
					\begin{aligned}			&p_{\phi,\gamma'\gamma_1}^{k}p_{\phi,\gamma_1\gamma_2}^{1}p_{\phi,\gamma_2\gamma_3}^{1}\cdots p_{\phi,\gamma_{n}\gamma'}^{1}=0,\\&\text{for }\forall\gamma_1\ne\gamma', \gamma_2,\cdots,\gamma_{n},k,n\in\mathbb{N}^+,
					\end{aligned}
				\end{equation}
				Meanwhile, we can leverage the \textit{Chapman--Kolmogorov equation} to expand the $n+1$-step transition probability $\widetilde{p_{\phi,\gamma'\gamma'}^{n+1}}$ and obtain \eqref{eq157} at the top of the previous page, where the set $\mathcal{D}$ is defined as $\mathcal{D}\triangleq\{j|\gamma_j=\gamma_{j+1},j=1,2,\cdots,n\}$, with $\gamma_{n+1}=\gamma'$ and \eqref{220b} is obtained by substituting \eqref{eq144}.
				As \eqref{eq154} holds for all $n$, we know that
\stepcounter{equation}
				\begin{equation}\label{eq158}
					p_{\phi,\gamma'\gamma_1}^{1}\prod_{j\notin\mathcal{D}} p_{\phi,\gamma_{j}\gamma_{j+1}}^{1}=0.
				\end{equation}
				Applying \eqref{eq158} in  \eqref{eq157} yields:
				\begin{equation}\label{eq221}
					\widetilde{p_{\phi,\gamma'\gamma'}^{n+1}}=\widetilde{p_{\phi,\gamma'\gamma'}^{1}}\widetilde{p_{\phi,\gamma'\gamma'}^{n}}.
				\end{equation}
				Iterating \eqref{eq221} yields
				\begin{equation}
					\widetilde{p_{\phi,\gamma'\gamma'}^n}=\left(\widetilde{p_{\phi,\gamma'\gamma'}^{1}}\right)^n=\left(1-\frac{\kappa\mathbb{E}[Y_i]-\kappa\mathbb{E}[Y_i]p_{\phi,\gamma'\gamma'}^1}{f(\phi(\gamma'))}\right)^n.
				\end{equation}
				Because $\gamma'$ is a transient state of the chain characterized by $p_{\phi,\gamma\gamma'}^1$, we have $p_{\phi,\gamma'\gamma'}^1<1$, and thus
				\begin{equation}
					\lim_{n\rightarrow\infty}\widetilde{p_{\phi,\gamma'\gamma'}^n}=0,
				\end{equation}
				which indicates that $\gamma'$ is also a transient state of the new chain characterized by $\widetilde{p_{\phi,\gamma\gamma'}^1}$. We thus accomplish the proof.

				\bibliographystyle{IEEEtran}
				\bibliography{reference}

				\vspace{-1cm}
				\IEEEbiographynophoto{Aimin Li}
				(Member, IEEE) received the B.S. degree (Best Thesis Award) and the Ph.D. degree (Awarded Best Dissertation Nomination) in electronic engineering from Harbin Institute of Technology (Shenzhen) in 2020 and 2025, respectively. From 2023 to 2024, he was a visiting researcher with the Institute for Infocomm Research (I$^2$R), Agency for Science, Technology, and Research (A*STAR), Singapore. His research interests include advanced channel coding techniques, information theory, and wireless communications. He has served as a reviewer for IEEE TIT, IEEE JSAC, IEEE TWC, IEEE TMC, IEEE TCOM, IEEE ISIT, among others, and as a session chair for IEEE Information Theory Workshop 2024 and IEEE Globecom 2024.
				
				\vspace{-1cm}
				\IEEEbiographynophoto{Shaohua Wu} (Member, IEEE) received the Ph.D. degree in communication engineering from the Harbin Institute of Technology (HIT), Harbin, China, in 2009. From 2009 to 2011, he held a Postdoctoral position with the Department of Electronics and Information Engineering, Shenzhen Graduate School, HIT (Shenzhen), Shenzhen, China, where he has been an Associate Professor since 2012. He is also an Associate Professor of Peng Cheng Laboratory, Shenzhen. From 2014 to 2015, he was a Visiting Researcher with BBCR, University of Waterloo, Waterloo, ON, Canada. He holds more than 30 Chinese patents. He has authored or coauthored more than 100 papers in the above-mentioned areas. His current research interests include wireless image/video transmission, space communications, advanced channel coding techniques, and B5G wireless transmission technologies.
				\vspace{-1cm}
				\IEEEbiographynophoto{Gary C.F. Lee} (Member, IEEE) received his B.S. degree in Electrical Engineering from Stanford University in 2016, his M.S. degree in Electrical Engineering from the Massachusetts Institute of Technology (MIT) in 2019, and his Ph.D. degree in Electrical Engineering from the Massachusetts Institute of Technology (MIT) in 2023. He is currently a Senior Scientist at the Institute for Infocomm Research (I2R), Agency for Science, Technology, and Research (A*STAR), Singapore. His current research interests are in wireless communications, machine learning, and signal processing.
				
				
				\IEEEbiographynophoto{Sumei Sun} (Fellow, IEEE) is a Principal Scientist, Distinguished Institute Fellow, and Acting Executive Director of the Institute for Infocomm Research (I2R), Agency for Science, Technology, and Research (A*STAR), Singapore. She is also holding a joint 
				appointment with the Singapore Institute of Technology, and an adjunct appointment 
				with the National University of Singapore, both as a full professor. Her current 
				research interests are in next-generation wireless communications, cognitive 
				communications and networks, industrial internet of things, communications-computing-control integrative design, joint radar-communication systems, and signal 
				intelligence. She is Editor-in-Chief of the IEEE Open Journal Vehicular Technology, and 
				Steering Committee Chair of the IEEE Transactions on Machine Learning in 
				Communications and Networking. She is also Member-at-Large of 
				the IEEE Communications Society, a member of the IEEE Vehicular Technology Society Board of Governors (2022-2024), Fellow of the IEEE and the Academy of Engineering Singapore.

			\end{document}